\documentclass[11pt]{article}

\usepackage[margin=1in]{geometry}
\usepackage{amsmath,amssymb,amsthm,mathtools}
\usepackage{booktabs,tabularx,array,ragged2e,float}
\usepackage{xcolor}
\usepackage{microtype}
\usepackage{hyperref}
\usepackage{tikz}
\usetikzlibrary{arrows.meta}

\hypersetup{
plainpages = false,
bookmarksopen = true,
colorlinks = true,
citecolor = purple,
linkcolor = teal,
urlcolor  = brown,
}
\usepackage{doi}

\allowdisplaybreaks
\newcommand{\C}{\mathbb C}
\newcommand{\D}{\mathrm D}
\newcommand{\Sym}{\operatorname{Sym}}
\newcommand{\Tr}{\operatorname{Tr}}
\newcommand{\conv}{\operatorname{conv}}
\newcommand{\SEP}{\operatorname{SEP}}
\newcommand{\SEPH}{\operatorname{SEP}_{\mathrm H}}
\newcommand{\SEPP}{\operatorname{SEP}_{\mathrm P}}
\newcommand{\BOS}{\operatorname{BOS}}
\newcommand{\EXCH}{\operatorname{EXCH}}
\newcommand{\EXT}{\operatorname{EXT}}
\newcommand{\BEXT}{\operatorname{BEXT}}
\newcommand{\Dtr}{D_{\mathrm{tr}}}
\newcommand{\DHS}{D_{\mathrm{HS}}}
\newcommand{\DTV}{D_{\mathrm{TV}}}
\newcommand{\DHtr}{D_{\mathrm H}^{\mathrm{tr}}}
\newcommand{\DHHS}{D_{\mathrm H}^{\mathrm{HS}}}
\newcommand{\cH}{\mathcal H}
\newcommand{\cX}{\mathcal X}
\newcommand{\cY}{\mathcal Y}

\newcommand{\Lin}{\mathrm L}
\newcommand{\Herm}{\operatorname{Herm}}
\newcommand{\ket}[1]{\lvert #1\rangle}
\newcommand{\bra}[1]{\langle #1\rvert}
\newcommand{\braket}[2]{\langle #1,#2\rangle}
\newcommand{\ketbra}[1]{\lvert #1\rangle\langle #1\rvert}
\newcommand{\norm}[1]{\left\lVert #1\right\rVert}
\newcommand{\abs}[1]{\left\lvert #1\right\rvert}
\newcommand{\ot}{\otimes}
\newcommand{\Mat}{\operatorname{Mat}}

\DeclareMathOperator{\rank}{rank}
\renewcommand{\epsilon}{\varepsilon}

\theoremstyle{plain}
\newtheorem{theorem}{Theorem}[section]
\newtheorem{introtheorem}[theorem]{Theorem}
\newtheorem{introcorollary}[theorem]{Corollary}
\newtheorem{proposition}[theorem]{Proposition}
\newtheorem{lemma}[theorem]{Lemma}
\newtheorem{fact}[theorem]{Fact}
\newtheorem{corollary}[theorem]{Corollary}
\theoremstyle{definition}

\theoremstyle{remark}
\newtheorem{remark}[theorem]{Remark}
\numberwithin{equation}{section}

\newcommand{\proofparagraph}[1]{\par\medskip\noindent\textbf{#1}\enspace}

\title{Optimal Quantum de Finetti Theorems via Argmax Rounding}
\author{
Fernando Granha Jeronimo%
\thanks{University of Illinois Urbana--Champaign.
Email: \url{granha@illinois.edu}}
\and
Pei Wu%
\thanks{The Pennsylvania State University.
Email: \url{pei.wu@psu.edu}}
\and
Haochen Xu%
\thanks{The Pennsylvania State University.
Email: \url{hpx5065@psu.edu}}
}
\date{}

\begin{document}
\maketitle

\begin{abstract}
We prove optimal finite quantum de Finetti upper bounds.  Given a bosonic state
$\rho_N\in\D(\Sym^N(\C^d))$, there is a probability measure $\nu$ on
the unit sphere such that
\[
 \left\|
  \rho_N^{(2)}-\int_{\norm{u}=1}\ketbra{u}^{\ot2}\,d\nu(u)
 \right\|_1
 \le \frac{\sqrt{d-1}}{N-1}.
\]
By purification, the bosonic theorem also gives the optimal $O(d/N)$
upper bound for arbitrary exchangeable states.  These results settle
the dimension dependence left open by Christandl, K\"onig,
Mitchison, and Renner (CMP 2007).  The proof casts de Finetti approximation
as sum-of-squares rounding and applies the argmax method of Jeronimo, Wu,
and Xu (manuscript 2026).

Among the consequences of the optimal bounds, we record three advances
on long-standing problems in quantum complexity that have resisted
progress for decades.
\begin{enumerate}
\item We refute Watrous's disentangler conjecture.  A disentangler is a
quantum channel whose outputs are all close to separable states and
whose image approximates every separable state.  For every fixed
$\epsilon\in(0,1)$, we construct such a channel with exact coverage,
that is, an $(\epsilon,0)$-disentangler
$\D(\C^D)\to\D(\C^d\ot\C^d)$ with
\[
 D=\exp\bigl(O_\epsilon(\sqrt d\log d)\bigr)=\exp(o(d)).
\]
The conjecture, recorded by Aaronson, Beigi, Drucker, Fefferman, and Shor
(CCC 2008), predicted instead that every constant-error disentangler must
have $D=\exp(\Omega(d))$.

\item We give the first algorithm for general explicit Best Separable
State with arbitrary completeness and soundness parameters whose running
time is subexponential in the local dimension: at every fixed accuracy,
it runs in time $\exp(\widetilde O(\sqrt d))$.  This removes the
perfect-completeness assumption from the earlier algorithm of Barak,
Kothari, and Steurer (STOC 2017), which achieved such a running time
only in that case.

\item We give the first algorithm for trace-norm separability testing of
an explicitly given state whose running time is subexponential in the
local dimension: at every fixed additive accuracy, it runs in time
$\exp(\widetilde O(\sqrt d))$.
Computational separability has been studied since at least the work of
Gurvits (STOC 2003), yet no such algorithm was previously known for
trace-norm testing.

\end{enumerate}

We also prove the first dimension-free bosonic quantum de Finetti theorem
in Hilbert--Schmidt norm, with upper bound $O(N^{-1/2})$.
This dimension-free result follows from our
optimal trace-norm theorem by spectral truncation and recovers the
dimension-free classical phenomenon of Diaconis and
Freedman (\emph{The Annals of Probability} 1980), which is impossible in
quantum trace norm.
\end{abstract}

\newpage
\setcounter{tocdepth}{2}
\tableofcontents
\newpage

\section{Introduction and main results}
\label{sec:intro}

\subsection{Background}

The quantum de Finetti theorem is one of the fundamental structural
results in quantum information.  It gives a quantitative expression of
the monogamy of entanglement: highly symmetric quantum systems cannot
share much entanglement among many subsystems.  The theorem captures
this principle by showing that the marginal on any fixed number of
sites approaches a mixture of tensor-power states as the total number
of sites grows.

The standard finite bounds make this statement quantitative.  Let
$\rho_N\in\D((\C^d)^{\ot N})$ be permutation invariant, and write
$\rho_N^{(t)}$ for its $t$-site marginal.  The quantum de Finetti
theorem of Christandl, K\"onig, Mitchison, and
Renner~\cite[Theorems~II.7 and~II.8]{CKMR07} states that, for every
$1\le t\le N$, there is a
probability measure $\mu$ on $\D(\C^d)$ such that
\[
 \frac12
 \left\|
  \rho_N^{(t)}
  -
  \int_{\D(\C^d)}\sigma^{\ot t}\,d\mu(\sigma)
 \right\|_1
 \le
 \frac{2d^2t}{N}.
\]
If $\rho_N$ is supported on $\Sym^N(\C^d)$, there is instead a
probability measure $\nu$ on the unit sphere of $\C^d$ such that
\[
 \frac12
 \left\|
  \rho_N^{(t)}
  -
  \int_{\norm{u}=1}\ketbra{u}^{\ot t}\,d\nu(u)
 \right\|_1
 \le
 \frac{2dt}{N}.
\]

The infinite quantum de Finetti theorem goes back to St{\o}rmer and to
Hudson and Moody~\cite{Sto69,HM76}; Caves, Fuchs, and Schack later
recast it in quantum-information language~\cite{CFS02}.  Finite
quantitative versions developed in the early 2000s, with the work of
Christandl, K\"onig, Mitchison, and Renner forming a culmination of the
theory of approximation by mixtures of exact tensor
powers~\cite{KR05,CKMR07}.  In a complementary direction, Renner's
exponential de Finetti theorem obtained exponentially small error by
allowing almost-product states, with applications to quantum
cryptography~\cite{Ren05,Ren07}.

In quantum complexity, finite de Finetti bounds are a basic tool for
analyzing SDP hierarchies based on symmetric extensions: they certify
convergence to separable or product states and thereby yield algorithms
for separability membership and product-state
optimization~\cite{DPS04,NOP09,BCY11algorithm,BH13,FF19}.  They also
enter the study of QMA with multiple unentangled proofs and of
disentangler constructions~\cite{KMY09,ABDFS08,BH13,LS15,HNW2019limitations,JLW26}.

An important refinement changes the metric so that the approximation
need only fool a restricted class of measurements.  Brand\~ao,
Christandl, and Yard obtained an improved de Finetti-type bound with
only logarithmic local-dimension dependence under one-way LOCC
measurements, yielding quasipolynomial separability algorithms for the
one-way LOCC and Hilbert--Schmidt norms~\cite{BCY11,BCY11algorithm}.
Brand\~ao and Christandl extended these quantitative
symmetric-extension bounds to the multipartite setting~\cite{BC12}.
Brand\~ao and Harrow subsequently proved
stronger de Finetti bounds under local measurements, with applications
to quantum complexity and polynomial optimization~\cite{BH13}.  Li and
Smith extended this theory to fully adaptive one-way LOCC measurements;
among other consequences, they show that polynomially many unentangled
proofs add no power to QMA when the verifier is restricted to one-way
LOCC~\cite{LS15}.  Such results are closely related to distinguishability
norms and quantum data hiding~\cite{MWW09,BCY11,BH13,LS15}.

The quantitative content of a finite de Finetti theorem is its
convergence rate.  Classically, the $t$-variable marginal of an
exchangeable sequence of length $N$ is within $O(t^2/N)$ in total
variation of a mixture of i.i.d.\ laws, independently of the size of
the sample space~\cite{DF80}.  Quantum trace-norm bounds necessarily
depend on the local Hilbert-space dimension, but the correct dependence
was left open.  Christandl, K\"onig, Mitchison, and Renner devoted a
section to optimality, asked whether their bounds could be improved,
and concluded that the dimension dependence was
unclear~\cite[Secs.~II.C and~V]{CKMR07}.  For bosonic states on
$\Sym^N(\C^d)$, the standard constructive theorem gives error $O(d/N)$
for two-site marginals~\cite{CKMR07,Chi11,LNR15}; Lewin, Nam, and
Rougerie later singled out improving this $d$-dependence as an open
problem~\cite[Remark~2.2]{LNR15}.

We settle this question by proving optimal trace-norm upper bounds.  The
main result improves the two-site bosonic error from $O(d/N)$ to
$O(\sqrt d/N)$, and purification gives an $O(d/N)$ upper bound for
arbitrary exchangeable states.
The square-root improvement is exactly what moves several
quantum-complexity problems from exponential to subexponential
dependence on the local dimension.

The sharp bosonic theorem is the engine behind the paper, but it is not
the end of the development.  We next prove a de Finetti theorem for
two-sided Bose-symmetric $(n,m)$-extendible states: every such
bipartite state on $A\ot B$ is within
\[
 O\left(
 \sqrt{\frac{\min\{d_A,d_B\}}{nm}}
 \right)
\]
in trace distance of the separable states.  For balanced extensions
with equal local dimensions, this becomes $O(\sqrt d/n)$.  The
two-sided theorem provides a common interface for three complexity
applications.
Realizing the two-sided relaxation as the image of a quantum channel
gives a subexponential disentangler.  Optimizing an accepting operator
over the same relaxation gives the subexponential Best Separable State
algorithm, and the Harrow--Montanaro reduction then yields trace-norm
separability testing~\cite{HM13}.

The optimal bosonic trace-norm theorem is also the starting point for
the first dimension-free bosonic de Finetti theorem in
Hilbert--Schmidt norm, with upper bound $O(N^{-1/2})$.  We also obtain
trace-norm bounds for higher marginals, with error $O(t\sqrt d/N)$ for
bosonic states and $O(td/N)$ for exchangeable states.  We state these
results precisely next.

\subsection{Our results}

\subsubsection{Optimal trace-norm de Finetti}

Let $\cH=\C^d$.  We use the following notations for trace distance
\[
 \Dtr(\rho,\sigma):=\frac12\norm{\rho-\sigma}_1,
 \qquad
 \Dtr(\rho,\mathcal K):=\inf_{\tau\in\mathcal K}\Dtr(\rho,\tau).
\]
A state in $\D(\Sym^N\cH)$ is \emph{bosonic}, while a state on
$\cH^{\ot N}$ is \emph{exchangeable} if it is permutation invariant.
For either type, we write $\rho_N^{(t)}$ for the $t$-site marginal.

For $2\le t\le N$, define the Hartree mixtures and tensor-power
mixtures by
\begin{align}
 \SEPH^{(t)}(d)
 &:=\conv\{\ketbra{u}^{\ot t}:u\in\C^d,\ \norm{u}=1\},
 \label{eq:intro-SEPH-definition}\\
 \SEPP^{(t)}(d)
 &:=\conv\{\sigma^{\ot t}:\sigma\in\D(\C^d)\}.
 \label{eq:intro-SEPP-definition}
\end{align}
The subscripts $\mathrm H$ and $\mathrm P$ stand for Hartree and
tensor power, respectively.  The sets of $t$-site marginals of
$N$-site bosonic and exchangeable states are
\begin{align}
 \BOS_N^{(t)}(d)
 &:=
 \left\{\rho_N^{(t)}:
 \rho_N\in\D\bigl(\Sym^N(\C^d)\bigr)\right\},
 \label{eq:intro-BOS-definition}\\
 \EXCH_N^{(t)}(d)
 &:=
 \left\{\sigma_N^{(t)}:
 \begin{array}{l}
 \sigma_N\in\D((\C^d)^{\ot N}),\\[-2pt]
 U_\pi \sigma_N U_\pi^\dagger=\sigma_N
 \quad\text{for every }\pi\in\mathfrak S_N
 \end{array}
 \right\}.
 \label{eq:intro-EXCH-definition}
\end{align}
Every Hartree mixture has an $N$-site bosonic extension, and every
tensor-power mixture has an $N$-site exchangeable extension.  Thus
both marginal sets are outer relaxations:
\[
 \SEPH^{(t)}(d)\subseteq\BOS_N^{(t)}(d),
 \qquad
 \SEPP^{(t)}(d)\subseteq\EXCH_N^{(t)}(d).
\]
For the nested pair $\mathcal L\subseteq \mathcal K$ below, $\DHtr(\mathcal K,\,\mathcal L):=\sup_{\rho\in\mathcal K} \Dtr(\rho,\mathcal L)$ denotes
the worst-case trace distance from a state in $\mathcal K$ to
$\mathcal L$.  We call this the \emph{worst-case de Finetti error}.

\begin{introtheorem}[Bosonic two-site de Finetti upper bound;
cf.\ Theorem~\ref{thm:bos-main}]
\label{thm:intro-bos-two-site}
For every $N,d\ge2$,
\[
 \DHtr\Bigl(
  \BOS_N^{(2)}(d),\,
  \SEPH^{(2)}(d)
 \Bigr)
 \le
 \frac{\sqrt{d-1}}{2(N-1)}.
\]
\end{introtheorem}

\begin{introtheorem}[Exchangeable two-site de Finetti upper bound;
cf.\ Corollary~\ref{thm:perm-main}]
\label{thm:intro-ex-two-site}
For every $N,d\ge2$,
\[
 \DHtr\Bigl(
  \EXCH_N^{(2)}(d),\,
  \SEPP^{(2)}(d)
 \Bigr)
 \le
 \frac{\sqrt{d^2-1}}{2(N-1)}.
\]
\end{introtheorem}

The exchangeable upper bound follows by applying the standard
purification and then using the bosonic theorem with dimension $d^2$;
the rectangular Werner examples of Christandl, K\"onig, Mitchison,
and Renner give the matching $\Omega(d/N)$ lower bound, and their
symmetric purifications give the matching $\Omega(\sqrt d/N)$
bosonic lower bound~\cite[Lemma~II.5 and Corollary~III.9]{CKMR07}.

\paragraph{Higher marginals.}

The two-site upper-bound dependence on $N$ and $d$ persists for
arbitrary $t$, with only a linear loss in $t$.
The case $t=2$ is covered by the preceding two-site theorems; the
higher-marginal argument treats $t\ge3$.

\begin{introtheorem}[$t$-site de Finetti upper bounds;
cf.\ Theorem~\ref{thm:higher-main}]
\label{thm:intro-higher}
For every $d\ge2$ and $2\le t\le N$,
\begin{align*}
 \DHtr\Bigl(
 \BOS_N^{(t)}(d),\,
 \SEPH^{(t)}(d)
 \Bigr)
 &\le
 \frac{4t\sqrt{d-1}}{N},\\
 \DHtr\Bigl(
 \EXCH_N^{(t)}(d),\,
 \SEPP^{(t)}(d)
 \Bigr)
 &\le
 \frac{4t\sqrt{d^2-1}}{N}.
\end{align*}
\end{introtheorem}

Thus no higher tensor-power loss is necessary: for fixed $t$, the
two-site dependence on $N$ and $d$ persists, and the two-site examples
above show that this dependence is optimal.  The linear factor in $t$
is also unavoidable in a joint large-dimension regime; see
Theorem~\ref{thm:rectangular-higher-lower} in the appendix.

\subsubsection{Two-sided Bose-symmetric extendible states and
trace-norm applications}

Let $\cH_A=\C^{d_A}$ and $\cH_B=\C^{d_B}$.  Recall that a bipartite
state is \emph{separable} if it is a convex combination of product
states, and write $\SEP(d_A,d_B)$ for the set of such states.  For
$n,m\ge1$, let $\EXT_{n,m}(d_A,d_B)$ be the set of states
$\rho_{AB}$ admitting an extension on
$\cH_A^{\ot n}\ot\cH_B^{\ot m}$ that is invariant under permutations
within the two blocks and whose $A_1B_1$ marginal is
$\rho_{AB}$~\cite{TDS03,JV13,JSZ22}.  Its Bose-symmetric subfamily
$\BEXT_{n,m}(d_A,d_B)$ consists of the states admitting such an extension
supported on
$\Sym^n(\cH_A)\ot\Sym^m(\cH_B)$.  This is the natural two-sided
analogue of the one-sided Bose-symmetric extension
hierarchy~\cite{DPS04,NOP09}.  Here
``Bose-symmetric'' denotes the same support condition called
``bosonic'' above.  Thus $\BEXT_{n,m}$ differs from $\EXT_{n,m}$ only
by this symmetric-support requirement.  Since pure product states have tensor-power
extensions,
\[
 \SEP(d_A,d_B)
 \subseteq
 \BEXT_{n,m}(d_A,d_B)
 \subseteq
 \EXT_{n,m}(d_A,d_B).
\]
The familiar one-sided $k$-extendibility hierarchy is the $(1,k)$ case,
up to exchanging the two parties.  Our bound applies to every pair of
extension levels, depends on the smaller local dimension, and captures
the multiplicative gain when both extension levels grow.  At the
one-sided endpoints, specialized de Finetti bounds can be stronger.

\begin{introtheorem}[De Finetti upper bound for two-sided
Bose-symmetric extendible states;
cf.\ Theorem~\ref{thm:two-sided-bose-extension}]
\label{thm:intro-two-sided-bext}
For every $n,m,d_A,d_B\ge1$,
\[
 \DHtr\Bigl(
  \BEXT_{n,m}(d_A,d_B),\,
  \SEP(d_A,d_B)
 \Bigr)
 \le
 \frac12
 \sqrt{\frac{\min\{d_A,d_B\}-1}{nm}}.
\]
\end{introtheorem}

For balanced extension levels and equal local dimensions, the
$\sqrt d/n$ dependence is optimal up to universal constants.  The
rectangular Werner construction of Christandl, K\"onig, Mitchison, and
Renner gives the matching order after symmetric purification and
grouping into two equal blocks~\cite{CKMR07}.

\begin{introcorollary}[De Finetti upper bound for two-sided extendible
states;
cf.\ Corollary~\ref{cor:two-sided-extension}]
\label{cor:intro-two-sided-ext}
For every $n,m,d_A,d_B\ge1$,
\[
 \DHtr\Bigl(
  \EXT_{n,m}(d_A,d_B),\,
  \SEP(d_A,d_B)
 \Bigr)
 \le
 \frac12
 \sqrt{\frac{\min\{d_A^2,d_B^2\}-1}{nm}}.
\]
\end{introcorollary}

This follows by taking the canonical purification of the two-block
extension and then applying the preceding theorem with local
dimensions $d_A^2$ and $d_B^2$.
In the balanced case, the resulting $d/n$ dependence for arbitrary
two-sided extensions is also optimal in a joint high-dimensional
regime, as witnessed by the rectangular Werner examples
of~\cite{CKMR07}.

The Bose-symmetric theorem above is the common interface for the three
trace-norm applications.  For the disentangler,
$\BEXT_{n,n}(d_A,d_B)$ is realized as the image of a channel: the
theorem controls every possible output, while tensor-power inputs give
exact coverage of the separable states.  For BSS, we optimize the
verifier's accepting operator over $\BEXT_{n,n}(d_A,d_B)$.
Because every feasible state is close to separable, the resulting
relaxation approximates $h_{\SEP}(M)$ with an additive error independent
of the accepting operator $M$.  Trace-norm weak membership then follows
from this BSS approximation.

All three applications therefore have resource bounds at the same
scale, subexponential in the local dimension.  When $d_A=d_B=d$, the
BSS SDP acts on
$\Sym^n(\C^d)\ot\Sym^n(\C^d)$, whose dimension is
\[
 \binom{d+n-1}{n}^{2}.
\]
The disentangler instead takes $2n$ unrestricted $d$-dimensional input
registers, grouped as $A_1,\ldots,A_n$ and $B_1,\ldots,B_n$, so its
input dimension is $D=d^{2n}$.  Both resources have logarithm
$O(n\log d)$.
The sharp error $O(\sqrt d/n)$ allows
$n=\Theta(\sqrt d/\epsilon)$, so both resource bounds are
$\exp(\widetilde O(\sqrt d/\epsilon))$.  An $O(d/n)$ convergence bound
would instead require $n=\Theta(d/\epsilon)$, returning to the
exhaustive-search scale.

\paragraph{Disentanglers.}
A disentangler was proposed as a way to replace two unentangled quantum
proofs by one ordinary quantum proof.  It is a channel
$\Lambda:\D(\C^D)\to\D(\C^d\ot\C^d)$ with two guarantees.  Every output
of $\Lambda$ is $\epsilon$-close to some separable state, so an
entangled input cannot create much entanglement at the output.
Conversely, every separable state is $\delta$-close to some output, so
the image covers the entire separable set.  Such a channel is called an
$(\epsilon,\delta)$-disentangler.  For sufficiently small errors, if
$\log D$ were polylogarithmic in $d$ and the channel could be
implemented efficiently, composing it with a $\mathrm{QMA}(2)$ verifier would
simulate the two unentangled proofs using one ordinary QMA witness.

Watrous's conjecture, recorded by Aaronson, Beigi, Drucker, Fefferman,
and Shor in 2008~\cite{ABDFS08}, proposed an information-theoretic
obstruction: it predicted that
$D=\exp(\Omega(d))$ whenever the fixed errors satisfy
$\epsilon+\delta<1$, where $d$ is the local output dimension.  The
conjecture had survived for almost two decades.  The two-sided
Bose-symmetric de Finetti theorem gives a channel with exact coverage,
$\delta=0$, whose input dimension is subexponential in the local output
dimension $d$.

\begin{introtheorem}[Subexponential disentanglers]
\label{thm:intro-disentangler}
For every fixed $\epsilon\in(0,1)$ and every $d\ge2$, there is an
$(\epsilon,0)$-disentangler
\[
 \Lambda:\D(\C^D)\longrightarrow\D(\C^d\ot\C^d)
\]
with
\begin{equation*}
 \log D=O_{\epsilon}(\sqrt d\log d).
\end{equation*}
In particular, there is a $(1/4,0)$-disentangler satisfying
the same dimension bound.
\end{introtheorem}

The image of the channel is $\BEXT_{n,n}(d,d)$.
Theorem~\ref{thm:intro-two-sided-bext} gives
soundness, while tensor-power inputs give exact coverage; the explicit
construction appears in Subsection~\ref{sec:disentanglers}.  For an
$m$-qubit output register, where $d=2^m$, our channel still needs
$\Theta_\epsilon(2^{m/2}m)$ input qubits, so the result does not imply
$\mathrm{QMA}(2)=\mathrm{QMA}$.  Its significance is instead structural: the
conjectured exponential obstruction in the local dimension fails,
leaving open the possibility of substantially smaller disentanglers.

Akibue, Kato, and Tani recently proved the conjectured exponential
lower bound for \emph{strong disentanglers}~\cite{AKT24}.  For every
reference system $R$ and every joint input, their definition requires
the output on $ABR$ to be close to a state separable across $A:BR$;
that is, one output must be approximately disentangled from both the
other output and the reference.  Our channel guarantees only that the
reduced output on $AB$ is close to separable across $A:B$ and need not
satisfy this stronger requirement.  Thus the two results separate the
original and strong notions of disentangling rather than contradicting
one another.

\paragraph{Best Separable State.}
The second application asks for the largest acceptance probability of
a measurement on a separable state.  Given its accepting operator
$M$, satisfying $0\preceq M\preceq I$ on $\C^d\ot\C^d$, the
\emph{Best Separable State} (BSS) value is
\[
 h_{\SEP}(M)
 :=
 \max_{\sigma\in\SEP(d,d)}\Tr(M\sigma)
 =
 \max_{\substack{\norm{x}=1\\\norm{y}=1}}
 \bra{x\ot y}M\ket{x\ot y}.
\]
The equality holds because the linear objective
$\Tr(M\sigma)$ is maximized over the separable set by a
pure product state.  Operationally, $M$ is the accepting operator of a
verifier, and $h_{\SEP}(M)$ is the largest acceptance probability
achievable by two unentangled messages.  This is exactly the
optimization core of $\mathrm{QMA}(2)$.  From a classical perspective,
it is also a basic tensor-optimization problem, linked through standard
reductions to the landscape surrounding the $2\to4$ norm and Unique
Games~\cite{BBHKSZ12}.

For constant accuracy, the direct net-search algorithm takes
$\exp(O(d))$ time.  Before this work, no $\exp(o(d))$ algorithm was
known for general explicit BSS---in other words, no nontrivial
improvement over exhaustive search was known.  Barak, Kothari, and
Steurer came closest: under the perfect-completeness promise
$h_{\SEP}(M)=1$, they obtained an algorithm running in time
subexponential in the local dimension and explicitly asked whether
that assumption could be removed~\cite[Remark~1.4]{BKS17}.  The
two-sided Bose-symmetric extension SDP hierarchy removes it.

\begin{introtheorem}[Best separable state without perfect completeness;
cf.\ Theorem~\ref{thm:bss-main}]
\label{thm:intro-bss}
Let $0\preceq M\preceq I$ be explicitly given on
$\C^d\ot\C^d$.  For every $\epsilon\in(0,1]$, the value
$h_{\SEP}(M)$ can be approximated to additive error $\epsilon$
deterministically.  Up to a factor polynomial in the input bit
complexity, the running time is
\begin{equation*}
 \exp\left(
 \widetilde O\left(\frac{\sqrt d}{\epsilon}\right)
 \right).
\end{equation*}
\end{introtheorem}

The guarantee in Theorem~\ref{thm:intro-bss} is proved in
Subsection~\ref{sec:qma2-applications} by evaluating
$h_{\BEXT_{n,n}}(M)$.  The additive error is independent of the value
of $h_{\SEP}(M)$, so the guarantee applies to every choice of
completeness and soundness parameters $c>s$ and removes perfect
completeness for explicit BSS.
The generic de Finetti argument does not exploit perfect completeness
and therefore does not recover the ETH-optimal bound established by the
direct SoS analysis~\cite{JWX26}.

\paragraph{Trace-norm separability testing.}
The third application is separability testing: given an explicitly
described bipartite state $\rho$, decide whether it is separable.
Since exact membership is not numerically robust, we use the standard
weak-membership formulation, with distance to $\SEP(d,d)$ measured in
trace distance.  Given a tolerance $\epsilon>0$, the task is to
distinguish
\[
 \rho\in\SEP(d,d)
 \qquad\text{from}\qquad
 \Dtr\bigl(\rho,\SEP(d,d)\bigr)\ge\epsilon,
\]
with no requirement when $\rho$ lies in the intervening
$\epsilon$-neighborhood.  Approximating the distance to the separable
set is a stronger, quantitative version of this problem.

The direction we need---from BSS optimization to this weak-membership
problem---is due to Harrow and
Montanaro~\cite[Proposition~16]{HM13}.  They show that an additive
approximation to $h_{\SEP}(M)$ valid for every accepting operator $M$
gives trace-norm weak membership for the separable set with only a
constant-factor loss in accuracy.
This is sharper than generic convex-optimization reductions, which lose
polynomial factors in the dimension~\cite{GLS93}.  In our setting, the reduction is
visible from trace-distance duality:
\[
 \Dtr\bigl(\rho,\SEP(d,d)\bigr)
 =
 \max_{0\preceq M\preceq I}
 \left\{\Tr(M\rho)-h_{\SEP}(M)\right\}.
\]
Here $M$ is the accepting operator of the two-outcome measurement that
best separates $\rho$ from every separable state, while
$h_{\SEP}(M)$ is the largest acceptance probability of that same
measurement on a separable state.  Our new input to the
Harrow--Montanaro reduction is a
subexponential-time approximation to $h_{\SEP}(M)$ valid for every
accepting operator $0\preceq M\preceq I$.  Optimizing the resulting
estimate over $M$ gives the distance to separability.

Weak membership for the separable set with promise gap
$1/\operatorname{poly}(d)$ is NP-hard~\cite{gurvits03sep,Gharibian10sep}.
Much less is known in the constant-gap regime.  In particular, before
this work no trace-norm algorithm running in time subexponential in the
local dimension was known nor ruled out.  The standard
$O(d/N)$ de Finetti analysis reaches constant error only at level
$N=\Omega(d)$, again at exponential cost; the quasipolynomial
algorithms in weaker LOCC and Hilbert--Schmidt geometries do not imply a
trace-norm algorithm~\cite{BCY11algorithm}. Our subexponential algorithm for
$h_{\SEP}(M)$ gives the first such subexponential
trace-norm weak-membership algorithm.

\begin{introtheorem}[Trace-norm separability testing;
cf.\ Theorem~\ref{thm:trace-norm-separability-algorithm}]
\label{thm:intro-trace-testing}
Let $\rho$ be an explicitly given state on $\C^d\ot\C^d$.  For every
$\epsilon\in(0,1]$, the trace distance from $\rho$ to the separable
states can be approximated to additive error $\epsilon$
deterministically.  Up to a factor polynomial in the input bit
complexity, the running time is
\begin{equation*}
 \exp\left(
 \widetilde O\left(\frac{\sqrt d}{\epsilon}\right)
 \right).
\end{equation*}
\end{introtheorem}

Harrow and Montanaro connected BSS through reductions with more than a
dozen problems, including injective tensor norms, channel output norms,
mean-field ground-state energy, minimum output R\'enyi-$2$ and
min-entropies, and planted clique~\cite[Sec.~4.2]{HM13}.  Our algorithm
therefore transfers to every problem in this network that reduces to
BSS, with the resulting running time determined by the dimension and
accuracy losses of the reduction.

\subsubsection{From optimal trace norm to dimension-free
Hilbert--Schmidt norm}

The classical finite de Finetti theorem is dimension free, while
dimension-free quantum convergence cannot hold in trace norm.
More practically, the familiar $O(\sqrt d/N)$ trace-norm bounds require
$N=\Omega(\sqrt d)$ even for constant accuracy, making de Finetti-based
algorithms prohibitively expensive in many applications.  These two
obstacles motivated de Finetti theorems under restricted
measurements~\cite{BC12,BH13,LS15}.  Rougerie further asked whether a
dimension-independent quantum analogue of the classical
Diaconis--Freedman theorem could be
found~\cite[Remark~4.2(4)]{Rou15}. The optimal trace-norm theorem motivates another
look at this question.  We show that Hilbert--Schmidt geometry restores
dimension-free convergence, with the optimal dimension-independent
rate.

We use Hilbert--Schmidt distance
\[
 \DHS(\rho,\sigma):=\norm{\rho-\sigma}_2,
 \qquad
 \DHS(\rho,\mathcal K)
 :=\inf_{\tau\in\mathcal K}\DHS(\rho,\tau),
\]
and use $\DHHS(\mathcal K,\,\mathcal L)$ for the analogous worst-case
de Finetti error in Hilbert--Schmidt distance.
\begin{introtheorem}[Dimension-free bosonic Hilbert--Schmidt de
Finetti upper bound; cf.\ Theorem~\ref{thm:HS-main}]
\label{thm:intro-HS-deFinetti}
For every $N,d\ge2$,
\[
 \DHHS\Bigl(
  \BOS_N^{(2)}(d),\,
  \SEPH^{(2)}(d)
 \Bigr)
 =O(N^{-1/2}),
\]
with a universal implied constant.
\end{introtheorem}

This gives a qualitatively new, dimension-free de Finetti theorem for
bosonic states.  Its $N^{-1/2}$ dependence is optimal when the
dimension is allowed to grow: the lower bound in
Theorem~\ref{thm:HS-lower} adapts the bosonic pair construction of
Jiang, Tacla, and Caves~\cite{JTC17}.  Unlike in trace norm, the
purification argument does not give the same conclusion for arbitrary
exchangeable states:
tracing out the purifying registers need not contract
Hilbert--Schmidt norm.  Whether an analogous dimension-free theorem
holds in the exchangeable setting remains open.

\begin{introtheorem}[Dimension-free Hilbert--Schmidt upper bound for
two-sided Bose-symmetric extendible states; cf.\
Corollary~\ref{cor:HS-two-sided-bose-extension}]
\label{thm:intro-HS-two-sided-bext}
For every $n,d_A,d_B\ge1$,
\[
 \DHHS\Bigl(
  \BEXT_{n,n}(d_A,d_B),\,
  \SEP(d_A,d_B)
 \Bigr)
 =O(n^{-1/2}),
\]
with a universal implied constant.
\end{introtheorem}

This rate is likewise optimal when the local dimensions may grow.
Theorem~\ref{thm:HS-two-sided-bose-extension-lower} groups the same
adapted pair construction into two equal blocks.

Because the Hilbert--Schmidt norm is its own dual, the state
approximation also controls $h_{\SEP}(M)$ for every Hermitian objective
matrix $M$ of bounded Hilbert--Schmidt norm.  Adding a multiple of the
identity to $M$ merely shifts every feasible value by the same known
constant, so only its traceless part matters.  For positive-semidefinite
objectives, Shi and Wu previously obtained a polynomial-time
approximation at fixed accuracy and bounded Hilbert--Schmidt norm by an
$\epsilon$-net method~\cite{SW15}.  Here the conclusion for arbitrary
Hermitian objectives follows instead from a constant-level
symmetric-extension SDP: at every fixed accuracy, it suffices to take
$n=O(\epsilon^{-2})$, so the resulting semidefinite program is
polynomial in the local dimensions.  The same de Finetti estimate also
improves the quasipolynomial weak-membership guarantee of Brand\~ao,
Christandl, and Yard~\cite{BCY11algorithm}.
Concurrent work of Malavolta~\cite{Malavolta26} gives a randomized
polynomial-time algorithm for constant-gap Euclidean separability.

\begin{introtheorem}[Polynomial-time $h_{\SEP}$ approximation;
cf.\ Theorem~\ref{thm:HS-BSS}]
\label{thm:intro-HS-hsep}
For every fixed $\epsilon>0$, $h_{\SEP}(M)$ admits a deterministic
additive-$\epsilon$ approximation in time polynomial in the local
dimensions whenever $M$ is explicitly given and its traceless part has
Hilbert--Schmidt norm $O(1)$.
\end{introtheorem}

\begin{introtheorem}[Polynomial-time Hilbert--Schmidt weak membership;
cf.\ Theorem~\ref{thm:HS-separability-testing}]
\label{thm:intro-HS-weak-membership}
For every fixed $\epsilon>0$, given an explicitly described state
$\rho\in\D(\C^{d_A}\ot\C^{d_B})$, its Hilbert--Schmidt distance
$\DHS\bigl(\rho,\SEP(d_A,d_B)\bigr)$ can be approximated to additive
error $\epsilon$ in deterministic time
polynomial in the local dimensions.  In the same running time one can
solve Hilbert--Schmidt weak membership, distinguishing
\[
 \rho\in\SEP(d_A,d_B)
 \qquad\text{from}\qquad
 \DHS\bigl(\rho,\SEP(d_A,d_B)\bigr)\ge\epsilon,
\]
with no requirement in the intervening regime.
\end{introtheorem}

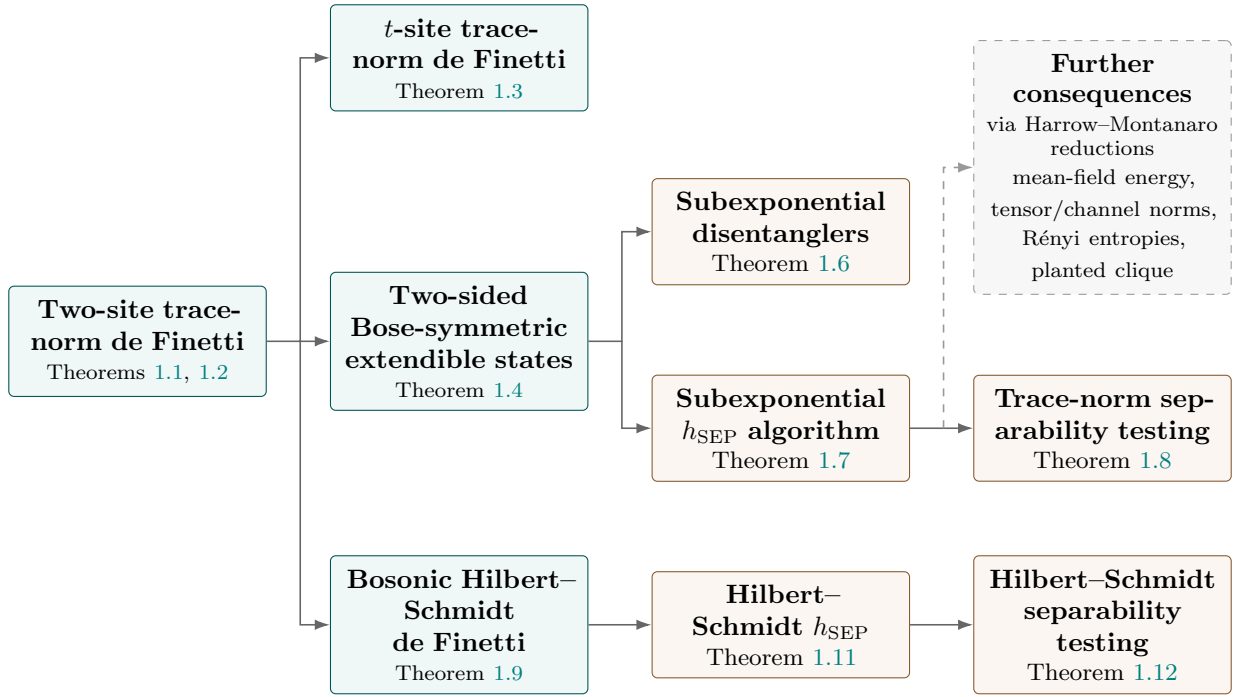
\begin{figure}[H]
\centering
\begin{tikzpicture}[
 theorem/.style={
  draw=teal!65!black,
  fill=teal!6,
  rounded corners=2pt,
  align=center,
  inner sep=5pt,
  minimum height=1.15cm,
  text width=3.05cm,
  font=\small
 },
 application/.style={
  draw=brown!70!black,
  fill=brown!7,
  rounded corners=2pt,
  align=center,
  inner sep=5pt,
  minimum height=1.15cm,
  text width=3.05cm,
  font=\small
 },
 transferred/.style={
  draw=black!40,
  fill=black!3,
  dashed,
  rounded corners=2pt,
  align=center,
  inner sep=5pt,
  minimum height=1.15cm,
  text width=3.05cm,
  font=\small
 },
 resultarrow/.style={
  -{Latex[length=2mm,width=1.4mm]},
  semithick,
  draw=black!60
 },
 transferarrow/.style={
  -{Latex[length=2mm,width=1.4mm]},
  semithick,
  dashed,
  draw=black!40
 }
]
\node[theorem] (trace) at (0,0) {
 {\bfseries Two-site trace-norm de Finetti}\\[-1pt]
 {\scriptsize Theorems~\ref{thm:intro-bos-two-site},
 \ref{thm:intro-ex-two-site}}
};

\node[theorem] (higher) at (4.25,3.75) {
 {\bfseries \(t\)-site trace-norm de Finetti}\\[-1pt]
 {\scriptsize Theorem~\ref{thm:intro-higher}}
};
\node[theorem] (twosided) at (4.25,0) {
 {\bfseries Two-sided Bose-symmetric extendible states}\\[-1pt]
 {\scriptsize Theorem~\ref{thm:intro-two-sided-bext}}
};
\node[theorem] (hsdf) at (4.25,-3.75) {
 {\bfseries Bosonic Hilbert--Schmidt de Finetti}\\[-1pt]
 {\scriptsize Theorem~\ref{thm:intro-HS-deFinetti}}
};

\node[application] (disentangler) at (8.50,1.45) {
 {\bfseries Subexponential disentanglers}\\[-1pt]
 {\footnotesize Theorem~\ref{thm:intro-disentangler}}
};
\node[application] (bss) at (8.50,-1.15) {
 {\bfseries Subexponential $h_{\SEP}$ algorithm}\\[-1pt]
 {\footnotesize Theorem~\ref{thm:intro-bss}}
};
\node[application] (tracetest) at (12.75,-1.15) {
 {\bfseries Trace-norm separability testing}\\[-1pt]
 {\footnotesize Theorem~\ref{thm:intro-trace-testing}}
};
\node[transferred] (hmnetwork) at (12.75,2.30) {
 {\bfseries Further}\\[-1pt]
 {\bfseries consequences}\\[-1pt]
 {\scriptsize via Harrow--Montanaro}\\[-3pt]
 {\scriptsize reductions}\\[-1pt]
 {\scriptsize mean-field energy, tensor/channel norms,}\\[-1pt]
 {\scriptsize R\'enyi entropies, planted clique}
};
\node[application] (hsbss) at (8.50,-3.75) {
 {\bfseries Hilbert--Schmidt $h_{\SEP}$}\\[-1pt]
 {\footnotesize Theorem~\ref{thm:intro-HS-hsep}}
};
\node[application] (hstest) at (12.75,-3.75) {
 {\bfseries Hilbert--Schmidt\\separability\\testing}\\[-1pt]
 {\footnotesize Theorem~\ref{thm:intro-HS-weak-membership}}
};

\coordinate (fork) at (2.15,0);
\coordinate (applicationfork) at (6.40,0);
\coordinate (bssfork) at (10.65,-1.15);
\draw[resultarrow] (trace.east) -- (fork) |- (higher.west);
\draw[resultarrow] (fork) |- (twosided.west);
\draw[resultarrow] (fork) |- (hsdf.west);
\draw[resultarrow] (twosided.east) -- (applicationfork) |- (disentangler.west);
\draw[resultarrow] (applicationfork) |- (bss.west);
\draw[resultarrow] (bss.east) -- (bssfork) -- (tracetest.west);
\draw[transferarrow] (bssfork) |- (hmnetwork.west);
\draw[resultarrow] (hsdf.east) -- (hsbss.west);
\draw[resultarrow] (hsbss.east) -- (hstest.west);
\end{tikzpicture}
\begin{minipage}{0.8\textwidth}
\small
\caption{Dependency structure of the main results.  Solid arrows show
the proof organization used in the main body.  Blue nodes are de Finetti
statements; tan nodes are applications.  The dashed gray branch records further
consequences available through the reductions of
Harrow--Montanaro~\cite[Sec.~4.2]{HM13}; their parameters depend on the
losses in those reductions.}
\label{fig:main-results-tree}
\end{minipage}
\end{figure}

\subsection{Technical overview}

The sharp two-site bosonic trace-norm theorem is the engine of the
paper.  We first explain its argmax proof, obtain the exchangeable
theorem by purification, and develop the rectangular version of
argmax rounding.  The rectangular estimate directly gives the
two-sided de Finetti theorem and its trace-norm applications.  The
higher-marginal proof then combines the same estimate with a symmetry
test.  Finally, a separate spectral-truncation argument builds on the
sharp trace-norm bound to prove the dimension-free bosonic
Hilbert--Schmidt theorem.

\subsubsection{Optimal two-site trace-norm de Finetti}

We begin with the optimal bosonic and exchangeable two-site bounds
stated in Theorems~\ref{thm:intro-bos-two-site}
and~\ref{thm:intro-ex-two-site}.

\paragraph{From de Finetti to a symmetric-tensor SDP integrality gap.}
The route to the proof began with a different problem.  In our recent
work on sum-of-squares algorithms~\cite{JWX26}, we developed argmax rounding in
pursuit of an ETH-optimal analysis of the SoS algorithm for Best Separable State
in the perfect-completeness case.  There the idea is quite natural: build a
high-moment polynomial from the SoS solution, choose the direction
where it is largest, and use the local optimality conditions at this
direction for rounding.

The connection between symmetric extensions and SoS/SDP relaxations
reaches back at least to the Doherty--Parrilo--Spedalieri (DPS)
hierarchy, which casts symmetric extendibility as an SDP test for
separability~\cite{DPS04}.  The duality between SoS relaxations on the
sphere and the DPS hierarchy is explicitly recorded by Fang and
Fawzi~\cite{FF19}.  By comparison, the associated rounding viewpoint
seems to have been explored less extensively; see, however,
Rao~\cite{Rao24} for related recent work on Hermitian SoS rounding and
quantum de Finetti theorems.

In our setting, trace-norm duality makes this interface concrete: the
Hartree optimum of a distinguishing witness is a polynomial
optimization over the sphere, while the level-$N$ bosonic optimum
is the corresponding symmetric-tensor SDP relaxation.  The argmax
viewpoint is especially natural here because the de Finetti theorem is
information-theoretic: unlike in the algorithmic setting, there is no
need to find the maximizing direction efficiently.  We may choose it
exactly and exploit its local optimality conditions without regard to
running time.  This shared geometry leads to the optimal de Finetti
rate, and the parallel with the earlier analyses goes further.

The BKS argument obtains its rounding direction by Gaussian rounding,
averaging over random directions and then conditioning~\cite{BKS17}.
Standard constructive de Finetti proofs likewise first choose a
universal rounding map, usually by averaging tensor powers over
Haar-random directions.  Given a bosonic state, this fixed map produces
one mixture that must approximate its marginal against every
measurement at once.  The resulting argument gives the familiar
$O(d/N)$ bosonic trace-norm bound~\cite{CKMR07,Chi11,Har13,LNR15}.

Our proof replaces this averaging by an objective-dependent extremal
choice: the rounding is tailored to the distinguishing witness.  To make
that choice, we reverse the order.  Trace-norm duality first exposes a
bounded Hermitian witness $M$, with $\norm M_\infty\le1$, for
the approximation error; only then do we
round an optimizer of the symmetric-tensor SDP relaxation for $M$.
That such dependence can reveal
extra information is already suggested by measurement-restricted de
Finetti theorems~\cite{BH13,LS15} and the value de Finetti theorem for
stoquastic verification~\cite{GGJ26}.  We do not need a single map that
succeeds simultaneously for all witnesses.  This is a weaker demand, but
the dual formulation below says that it is exactly enough.

To make this precise, let $M\in\Herm(\Sym^2\cH)$ and define
\begin{align*}
 h(M)
 &:=
 \max_{\norm{u}=1}
 \bra{u^{\ot2}}M\ket{u^{\ot2}},\\
 M^{[N]}
 &:=
 \binom N2^{-1}
 \sum_{1\le i<j\le N}M_{ij}
 \quad\text{on }\Sym^N\cH,
\end{align*}
where $M_{ij}$ acts as $M$ on sites $i,j$ and as the identity on every
other site.  Thus $h(M)$ is the original quartic optimization, while
\[
 \lambda_{\max}(M^{[N]})
 =
 \max_{\rho_N\in\D(\Sym^N\cH)}
 \Tr(M^{[N]}\rho_N)
\]
is the value of its level-$N$ symmetric-tensor SDP relaxation, the
bosonic form of the corresponding SoS relaxation.  After dualizing trace distance
and using permutation symmetry, the worst-case de Finetti error is exactly
\begin{equation}
 \DHtr\Bigl(
  \BOS_N^{(2)}(d),\,
  \SEPH^{(2)}(d)
 \Bigr)
 =
 \frac12\sup_{\norm{M}_\infty\le1}
 \left\{\lambda_{\max}(M^{[N]})-h(M)\right\}.
 \label{eq:intro-dual-bos}
\end{equation}
So the de Finetti problem has become a familiar TCS question: bound the
integrality gap of an SDP relaxation.  We no longer need to construct
the approximating mixture.  For each bounded Hermitian witness $M$, it
is enough to round an optimizer of the corresponding relaxation
with loss $O(\sqrt d/N)$.

\paragraph{The argmax rounding object.}
Fix $M$ and let $\psi\in\Sym^N(\cH)$ be a unit top eigenvector of
$M^{[N]}$.  The argmax principle gives a natural rounding object:
choose the tensor power $u^{\ot N}$ having maximum overlap with
$\psi$.  After fixing its phase, write
\[
 c:=\braket{u^{\ot N}}{\psi}>0.
\]
The argmax lemma provides the basic tools for analyzing this rounding
object.  Take the partial inner product of $\psi$ with $u$ on all but
two sites:
\[
 \eta_u
 :=
 (\bra u^{\ot(N-2)}\ot I_{\cH^{\ot2}})\psi
 \in\Sym^2(\cH).
\]
Concretely, the lemma says
\begin{equation}
 \eta_u=c\,u^{\ot2}+\zeta_u,
 \qquad
 \zeta_u\in\Sym^2(u^\perp),
 \qquad
 \sup_{\substack{v\perp u\\\norm{v}=1}}
 \abs{\braket{v^{\ot2}}{\zeta_u}}
 \le\frac{c}{N-1}.
 \label{eq:intro-argmax-contraction}
\end{equation}
Here is why.  Perturb $u$ toward any orthogonal direction $v$ and
differentiate the overlap with $\psi$.  The first derivative says that
the component with one copy of $v$ vanishes.  The second derivative
says that the component with two copies of $v$ is only $O(c/N)$.
Doing this for every $v\perp u$ gives
\eqref{eq:intro-argmax-contraction}.  In words, the partial inner
product is the tensor square $u^{\ot2}$, scaled by $c$, plus an error
living entirely in directions orthogonal to $u$ and having small
overlap with every orthogonal square $v^{\ot2}$.  These are precisely
the two tools
supplied by the argmax lemma: the mixed component vanishes, and the
fully orthogonal component is small against every square.  This is the
bosonic counterpart of the pseudo-covariance lemma underlying our SoS
argmax rounding~\cite{JWX26}.

\paragraph{Bounding the integrality gap.}
We now use the information supplied by the argmax lemma to compare
$\lambda_{\max}(M^{[N]})$ with $h(M)$.  By symmetry, the lift
$M^{[N]}$ can be evaluated on sites $1,2$, and hence
\begin{equation}
 \lambda_{\max}(M^{[N]})c
 =
 \braket{Mu^{\ot2}}{\eta_u}.
 \label{eq:intro-eigen-contraction}
\end{equation}
Let
\[
 B_u:=P_{\Sym^2(u^\perp)}Mu^{\ot2}.
\]
The component of $Mu^{\ot2}$ along $u^{\ot2}$ contributes the feasible
tensor-square value
$\bra{u^{\ot2}}M\ket{u^{\ot2}}\le h(M)$; the component in
$u\vee u^\perp$ is orthogonal to $\eta_u$.  Equations
\eqref{eq:intro-argmax-contraction} and
\eqref{eq:intro-eigen-contraction} therefore leave the exact identity
\[
 \bigl(
  \lambda_{\max}(M^{[N]})
  -\bra{u^{\ot2}}M\ket{u^{\ot2}}
 \bigr)c
 =
 \braket{B_u}{\zeta_u}.
\]

Choose an orthonormal basis of $u^\perp$ and identify symmetric
two-tensors with complex symmetric matrices.  Write $C_u$ for the
matrix representing $B_u$ and $Z_u$ for the matrix representing
$\zeta_u$.  The argmax bound
\eqref{eq:intro-argmax-contraction} says
$\norm{Z_u}_\infty\le c/(N-1)$.  Moreover,
$C_u$ is a $(d-1)\times(d-1)$ matrix and
$\norm{C_u}_2=\norm{B_u}\le\norm M_\infty$.  Its trace norm is
therefore at most $\sqrt{d-1}\norm M_\infty$.  Schatten duality now gives
\[
 \abs{\braket{B_u}{\zeta_u}}
 =
 \abs{\Tr(C_u^\dagger Z_u)}
 \le
 \frac{c\sqrt{d-1}}{N-1}\norm M_\infty.
\]
Consequently,
\[
 \lambda_{\max}(M^{[N]})
 \le
 h(M)+\frac{\sqrt{d-1}}{N-1}\norm M_\infty.
\]
Together with \eqref{eq:intro-dual-bos}, this proves the sharp
$O(\sqrt d/N)$ theorem.  The important point is that argmax rounding
leaves the single matrix pairing $\Tr(C_u^\dagger Z_u)$.  Passing from
the Hilbert--Schmidt norm of the $(d-1)\times(d-1)$ matrix $C_u$ to its
trace norm costs only $\sqrt{d-1}$, which is exactly the improvement
we need.

%
The exchangeable theorem follows by purifying each site into
dimension $d^2$, applying the bosonic theorem, and tracing out the
purifying registers.

\paragraph{Two-sided extendibility.}
The same proof extends directly to two symmetric blocks.  For a
Hermitian operator $M$ on $\cH_A\ot\cH_B$, average $M$ over the $nm$
cross-block pairs to form the rectangular lift $M^{[n,m]}$.  Maximizing
the overlap of a top eigenvector with $x^{\ot n}\ot y^{\ot m}$, the
same first- and second-order argument gives
\[
 \lambda_{\max}\bigl(M^{[n,m]}\bigr)-h_{\SEP}(M)
 \le
 \sqrt{\frac{\min\{d_A,d_B\}-1}{nm}}
 \norm{M}_\infty.
\]
Trace-distance duality yields
Theorem~\ref{thm:intro-two-sided-bext}, and a canonical purification
of the two blocks yields Corollary~\ref{cor:intro-two-sided-ext}.

\paragraph{The disentangler from the two-sided theorem.}
The disentangler in Theorem~\ref{thm:intro-disentangler} is obtained by
realizing $\BEXT_{n,n}(d,d)$ as the image of a channel: project the
$A_1,\ldots,A_n$ registers and the $B_1,\ldots,B_n$ registers onto
$\Sym^n(\C^d)$, sending rejected mass to a fixed product state, and
then retain only $A_1B_1$.  Every output is Bose-symmetric
$(n,n)$-extendible, while every separable state has an exact preimage
obtained from tensor powers of its product-state decomposition.
Choosing $n=O(\sqrt d/\epsilon)$ gives an $(\epsilon,0)$-disentangler
with input dimension $D=d^{2n}$, and hence
$\log D=O_\epsilon(\sqrt d\log d)$.

\paragraph{BSS and weak membership from the same theorem.}
For an accepting operator $0\preceq M\preceq I$ on $AB$, optimize
over $\BEXT_{n,n}(d_A,d_B)$ and write
\[
 h_{\BEXT_{n,n}}(M)
 :=
 \max_{\rho\in\BEXT_{n,n}(d_A,d_B)}\Tr(M\rho).
\]
Every separable state is feasible, while the two-sided theorem places
every feasible state within
$O(\sqrt{\min\{d_A,d_B\}}/n)$ of the separable states.
Consequently,
\[
 h_{\SEP}(M)
 \le
 h_{\BEXT_{n,n}}(M)
 \le
 h_{\SEP}(M)
 +O\left(\frac{\sqrt{\min\{d_A,d_B\}}}{n}\right).
\]
This bound, valid for every $M$, gives the BSS approximation in
Theorem~\ref{thm:intro-bss}.  The Harrow--Montanaro
reduction~\cite[Proposition~16]{HM13} then gives trace-norm weak
membership, and the equivalent trace-distance dual formulation gives
the additive distance approximation in
Theorem~\ref{thm:intro-trace-testing}.

\subsubsection{\texorpdfstring{$t$}{t}-site trace-norm de Finetti}

\paragraph{The \texorpdfstring{$1:(t-1)$}{1:(t-1)} cut.}
For the $t$-site bounds in Theorem~\ref{thm:intro-higher}, trace-distance
duality asks us to compare the value of a witness $M$ with its Hartree
optimum.  For $0\preceq M\preceq I$ on $\Sym^t(\cH)$, this optimum is
\[
 h_t(M):=\max_{\norm{u}=1}\bra{u^{\ot t}}M\ket{u^{\ot t}}.
\]
The two-site argument suggests first enforcing unentanglement across a
bipartition.  The natural choice is the $1:(t-1)$ cut
\[
 \cH_A:=\cH,
 \qquad
 \cH_B:=\Sym^{t-1}(\cH),
 \qquad
 \Sym^t(\cH)\subseteq\cH_A\ot\cH_B.
\]
This cut exposes one factor of dimension $d$, so the loss in rectangular
argmax rounding depends on $\sqrt{d-1}$ rather than on the much larger
dimension of $\Sym^{t-1}(\cH)$.  It also fits the bosonic extension naturally:
organizing the $N$ sites into $n$ individual sites and $m$ blocks of
$t-1$ sites makes $\rho_N^{(t)}$ Bose-symmetric $(n,m)$-extendible across
$\cH_A:\cH_B$.  The associated rectangular lift is therefore the two-sided
symmetric-extension relaxation of product optimization across this cut.

The argmax principle supplies the main rounding step.  The value of the
relaxation is the largest eigenvalue of the rectangular lift.  Let $\psi$ be
a unit top eigenvector of the rectangular lift.  Argmax rounding chooses
\[
 (x,y)
 \in
 \operatorname*{arg\,max}_{\substack{\norm{v}=1\\\norm{w}=1}}
 \abs{\braket{v^{\ot n}\ot w^{\ot m}}{\psi}},
\]
where $v,x\in\cH_A$ and $w,y\in\cH_B$, and returns the product vector
$x\ot y$.  This enforces unentanglement across the $1:(t-1)$ cut, but
it does not yet imply that $y=x^{\ot(t-1)}$, as required for a Hartree
vector.

\paragraph{Symmetry propagates unentanglement.}
To propagate this unentanglement to the remaining cuts, we use the
symmetry test in the spirit of the disentangler from unentanglement of
Jeronimo and Wu~\cite{JW24}.  Let $\Pi_t$ be the orthogonal projector from
$\cH_A\ot\cH_B$ onto $\Sym^t(\cH)$.  A product vector in this symmetric
subspace must be a tensor power.  Define the symmetry-test rejection
probability by
\[
 \delta_{\mathrm{sym}}
 :=
 1-\bra{x\ot y}\Pi_t\ket{x\ot y}.
\]
Permutation averaging gives
\[
 1-\abs{\braket{x^{\ot(t-1)}}{y}}^2
 \le t\delta_{\mathrm{sym}}.
\]
Thus, when the rounded product passes the symmetry test with high
probability, its unentanglement across the $1:(t-1)$ cut propagates to
the remaining cuts, and $x\ot y$ is close to the tensor power $x^{\ot t}$.
This is the content of
Lemma~\ref{lem:symmetry-projector}; its linear dependence on $t$ sharpens
the corresponding one-versus-many estimate in~\cite[Theorem~23]{JW24}.

\paragraph{Bounding the integrality gap.}
We incorporate the symmetry test before rectangular argmax rounding so
that the witness value and the rejection probability are
controlled through the same rounded product.  Extend $M$ by zero outside
$\Sym^t(\cH)$ and introduce the mixed test
\[
 M_\theta:=(1-\theta)\Pi_t+\theta M,
 \qquad 0<\theta<1.
\]
It performs the symmetry test with weight $1-\theta$ and evaluates $M$
with weight $\theta$.  Every bosonic $t$-site marginal passes the
symmetry test with certainty.  If $\widetilde\rho_{n,m}$ is its grouped
extension, then
\[
 1-\theta+\theta\Tr\bigl(M\rho_N^{(t)}\bigr)
 =
 \Tr\bigl(M_\theta^{[n,m]}\widetilde\rho_{n,m}\bigr)
 \le
 \lambda_{\max}\bigl(M_\theta^{[n,m]}\bigr).
\]
The right-hand side is the rectangular symmetric-extension relaxation
of $h_{\SEP}(M_\theta)$.  Applying rectangular argmax rounding to this
lift produces $x\ot y$.  On that product vector,
\[
 \bra{x\ot y}M_\theta\ket{x\ot y}
 =
 (1-\theta)(1-\delta_{\mathrm{sym}})
 +\theta\bra{x\ot y}M\ket{x\ot y},
\]
so rejection lowers the mixed-test value.  Lemma~\ref{lem:symmetry-projector}
controls the replacement of $x\ot y$ by $x^{\ot t}$ in terms of the same
quantity; the proof keeps these two contributions coupled rather than
first establishing a separate bound on $\delta_{\mathrm{sym}}$.

Let $\beta$ denote the error in rectangular argmax rounding.  Combining
the rectangular argmax estimate with the symmetry-test bounds gives
\begin{equation}
 \theta\left(\Tr\bigl(M\rho_N^{(t)}\bigr)-h_t(M)\right)
 \le
 -(1-\theta)\delta_{\mathrm{sym}}
 +\bigl(\theta\sqrt t+\beta(1-\theta)\bigr)
  \sqrt{\delta_{\mathrm{sym}}}
 +\beta\theta.
 \label{eq:intro-higher-master}
\end{equation}
Choosing the two extension levels and $\theta$ appropriately in this
inequality gives the claimed $O(t\sqrt d/N)$ bound.  Trace-distance
duality proves the bosonic theorem, and symmetric purification gives
the exchangeable theorem.

\subsubsection{From optimal trace norm to dimension-free
Hilbert--Schmidt norm}

The dimension-free bosonic Hilbert--Schmidt de Finetti theorem stated
in Theorem~\ref{thm:intro-HS-deFinetti} uses the sharp trace-norm
theorem in a different way.  On a subspace of dimension less than $N$,
the $\sqrt d/N$ trace-norm bound is already $O(N^{-1/2})$.  The spectral
subspace of the one-site marginal corresponding to eigenvalues above
$1/N$ has dimension less than $N$.  The proof applies the trace-norm
theorem on this spectral head and controls every block meeting the
complementary tail directly in Hilbert--Schmidt norm.

\paragraph{Spectral truncation.}
We first reduce to a pure bosonic state $\rho_N=\ketbra\psi$ and look
at its one-site marginal $\rho_1$.  Define
\[
 P:=\mathbf 1_{(1/N,\infty)}(\rho_1),
 \qquad
 Q:=I-P.
\]
We call $P\cH$ the spectral head and $Q\cH$ the spectral tail.  Since
$\Tr\rho_1=1$,
\[
 \rank P<N,
 \qquad
 Q\rho_1Q\preceq\frac1NQ.
\]
Let
\[
 \rho_2^{\mathrm{head}}
 :=
 (P\ot P)\rho_2(P\ot P)
\]
be the compression of the two-site marginal to the head.
The proof constructs a Hartree mixture
$\tau_{\mathrm{head}}$, supported on
$\Sym^2(P\cH)$ and having the same trace as
$\rho_2^{\mathrm{head}}$.  A second Hartree mixture
$\tau_{\mathrm{res}}$, supported on $\Sym^2(Q\cH)$, fills the missing
trace.  Thus
\[
 \tau:=\tau_{\mathrm{head}}+\tau_{\mathrm{res}}
 \in\SEPH^{(2)}(d).
\]
The three estimates enter through the error decomposition
\begin{equation}
 \norm{\rho_2-\tau}_2
 \le
 \underbrace{
  \norm{\rho_2^{\mathrm{head}}-\tau_{\mathrm{head}}}_2
 }_{\text{round the high-eigenvalue compression}}
 +
 \underbrace{
  \norm{\rho_2-\rho_2^{\mathrm{head}}}_2
 }_{\substack{\text{remove all blocks meeting}\\
               \text{the spectral tail}}}
 +
 \underbrace{\norm{\tau_{\mathrm{res}}}_2}_{\text{restore the missing trace}}.
 \label{eq:intro-HS-decomposition}
\end{equation}

The head compression is not itself an $N$-site marginal, so we decompose
$\psi$ according to the number $k$ of sites in the head.  In the
component with exactly $k$ sites in $P\cH$, the probability that a
prescribed pair lies in the head is
$\frac{k(k-1)}{N(N-1)}$, whereas the sharp trace-norm theorem costs only the
square root of the head dimension divided by $k-1$.  The factor $k-1$
cancels, and the fact that the head dimension is below $N$ leaves an
$O(N^{-1/2})$ Hilbert--Schmidt error.  This cancellation is exactly
where the optimal two-site theorem enters the dimension-free argument.

For the remaining terms, decompose $\Sym^2\cH$ as
\[
 \Sym^2(P\cH)
 \oplus(P\cH\vee Q\cH)
 \oplus\Sym^2(Q\cH).
\]
The cutoff $Q\rho_1Q\preceq Q/N$ gives operator-norm control on the two
diagonal compressions containing at least one $Q$ direction.  An ordinary
Schmidt decomposition handles $P\cH\vee Q\cH$, while its symmetric form,
or Takagi decomposition, handles $\Sym^2(Q\cH)$.  These decompositions
reduce the estimates to partial inner products of
$\psi$ with orthogonal vectors in $Q\cH$, each controlled by the cutoff.
Positive-block estimates first combine the mixed and tail--tail
diagonal bounds into an operator-norm bound on the full tail block and
then control the off-diagonal head--tail blocks.  Thus removing every
block meeting $Q\cH$ costs only $O(N^{-1/2})$ in Hilbert--Schmidt norm.
Finally, restore the missing
trace by adding the Haar-averaged Hartree mixture on $\Sym^2(Q\cH)$; its
Hilbert--Schmidt norm is $O(1/N)$.  Substitution into
\eqref{eq:intro-HS-decomposition} proves the dimension-free
$O(N^{-1/2})$ theorem.

\paragraph{Two-sided extensions.}
The dimension-free Hilbert--Schmidt de Finetti theorem above extends
to two-sided Bose-symmetric extendible states:
\[
 \DHHS\bigl(
  \BEXT_{n,n}(d_A,d_B),\SEP(d_A,d_B)
 \bigr)
 =O(n^{-1/2}).
\]
This is the content of Theorem~\ref{thm:intro-HS-two-sided-bext}.
Since the Hilbert--Schmidt norm is self-dual, the two-sided
Hilbert--Schmidt bound gives analogues of the BSS and weak-membership
results above.
Taking $n=O(\epsilon^{-2})$ yields the polynomial-time
$h_{\SEP}$ approximation and separability testing in
Theorems~\ref{thm:intro-HS-hsep}
and~\ref{thm:intro-HS-weak-membership}.

\section{Preliminaries}
All Hilbert spaces in this note are finite-dimensional and complex.  We denote
them by $\cH$, adding subscripts to indicate their roles, and take their inner
products to be conjugate-linear in the first argument.  For $z\in\C$, we write
$\Re(z)$ and $\Im(z)$ for its real and imaginary parts, respectively.  We write
$\Lin(\cH)$ for the space of linear operators on $\cH$ and
$\Herm(\cH)$ for its real subspace of Hermitian operators; and
$\D(\cH)\subseteq\Herm(\cH)$ the density operators.  For a vector $v$, the notation $\norm{v}$
denotes its Hilbert-space norm.  For $u\in\cH$, we write
$u^\perp:=\{v\in\cH:\braket{u}{v}=0\}$; more generally,
$\mathcal U^\perp$ denotes the orthogonal complement of a subspace
$\mathcal U$, with the ambient Hilbert space understood from context.
For $X\in\Lin(\cH)$, let
$\abs{X}:=\sqrt{X^\dagger X}$ and define the Schatten norms by
\[
 \norm{X}_p:=\bigl(\Tr\abs{X}^p\bigr)^{\frac{1}{p}}
 \quad(1\le p<\infty),
 \qquad
 \norm{X}_\infty:=\max_j s_j(X),
\]
where $s_j(X)$ are the singular values of $X$.  Thus $\norm{X}_1$,
$\norm{X}_2$, and $\norm{X}_\infty$ are the trace, Hilbert--Schmidt, and
operator norms, respectively.
We use normalized trace distance and unnormalized Hilbert--Schmidt
distance:
\begin{equation}
 \Dtr(\rho,\sigma):=\frac12\norm{\rho-\sigma}_1,
 \label{eq:trace-distance}
\end{equation}
and
\begin{equation}
 \DHS(\rho,\sigma):=\norm{\rho-\sigma}_2.
 \label{eq:HS-distance}
\end{equation}
For probability distributions, $\DTV$ denotes total-variation
distance.
Thus every displayed trace-distance constant is twice as large when
the error is measured by $\norm{\rho-\sigma}_1$ instead.

For $k\ge1$, let $\mathfrak S_k$ act on $\cH^{\ot k}$ by permuting
the tensor factors, and write $U_\pi$ for the operator corresponding
to $\pi\in\mathfrak S_k$.  The symmetric tensor power is
\[
 \Sym^k\cH
 :=
 \left\{
  \Psi\in\cH^{\ot k}:
  U_\pi\Psi=\Psi\ \text{for every }\pi\in\mathfrak S_k
 \right\},
\]
and we set $\Sym^0\cH:=\C$.

\subsection{Separable and extendible states}

The de Finetti theorems compare $t$-site marginals of bosonic or
exchangeable states, and bipartite states admitting two-sided
extensions, with specific separable state families.  We now define the
states appearing in these comparisons.

\paragraph{Separable state families.}
Fix $d\ge1$ and $t\ge1$.  The two separable families used below are
\[
\begin{aligned}
 \SEPH^{(t)}(d)
 &:={}
 \conv\left\{
  \ketbra{u}^{\ot t}:u\in\C^d,\ \norm{u}=1
 \right\},\\
 \SEPP^{(t)}(d)
 &:={}
 \conv\left\{
  \sigma^{\ot t}:\sigma\in\D(\C^d)
 \right\}.
\end{aligned}
\]
The subscripts $\mathrm H$ and $\mathrm P$ stand for Hartree and
tensor power, respectively.  Since pure states are density operators,
\[
 \SEPH^{(t)}(d)
 \subseteq
 \SEPP^{(t)}(d).
\]
The families $\SEPH^{(t)}(d)$ and $\SEPP^{(t)}(d)$ are the approximants
in the bosonic and exchangeable de Finetti theorems, respectively.

\paragraph{Bosonic states.}
Let $N\ge t$.  A state $\rho_N$ is \emph{bosonic} if
$\rho_N\in\D(\Sym^N(\C^d))$.  Its marginal on the first $t$ sites is
\[
 \rho_N^{(t)}
 :=
 \Tr_{t+1,\ldots,N}\rho_N.
\]
The set of these marginals is
\[
 \BOS_N^{(t)}(d)
 :=
 \left\{
  \rho_N^{(t)}:
  \rho_N\in\D\bigl(\Sym^N(\C^d)\bigr)
 \right\}.
\]
Extending Hartree states by tensor powers and tracing out extension
sites give
\[
 \SEPH^{(t)}(d)
 \subseteq
 \BOS_{N+1}^{(t)}(d)
 \subseteq
 \BOS_N^{(t)}(d).
\]

\paragraph{Exchangeable states.}
A state $\rho_N\in\D((\C^d)^{\ot N})$ is \emph{exchangeable} if
\[
 U_\pi\rho_NU_\pi^\dagger=\rho_N
 \qquad\text{for every }\pi\in\mathfrak S_N.
\]
The set of all such $t$-site marginals is
\[
 \EXCH_N^{(t)}(d)
 :=
 \left\{
  \rho_N^{(t)}:
  \begin{array}{l}
   \rho_N\in\D\bigl((\C^d)^{\ot N}\bigr),\\[-2pt]
   U_\pi\rho_NU_\pi^\dagger=\rho_N
   \quad\text{for every }\pi\in\mathfrak S_N
  \end{array}
 \right\}.
\]
Tensor powers extend exchangeably, and every bosonic state is
exchangeable.  Thus
\[
\begin{gathered}
 \SEPP^{(t)}(d)
 \subseteq
 \EXCH_{N+1}^{(t)}(d)
 \subseteq
 \EXCH_N^{(t)}(d),\\
 \BOS_N^{(t)}(d)
 \subseteq
 \EXCH_N^{(t)}(d).
\end{gathered}
\]

\paragraph{Two-sided extendible states.}
For $d_A,d_B\ge1$, define the bipartite separable states by
\[
 \SEP(d_A,d_B)
 :=
 \conv\left\{
  \ketbra{u}\ot\ketbra{v}:
  u\in\C^{d_A},\ v\in\C^{d_B},\
  \norm{u}=\norm{v}=1
 \right\}.
\]
For $n,m\ge1$, the $(n,m)$-extendible states are the marginals
\[
 \EXT_{n,m}(d_A,d_B)
 :=
 \left\{
  \widetilde\rho_{A_1B_1}:
  \begin{array}{l}
   \widetilde\rho
   \in\D\bigl((\C^{d_A})^{\ot n}\ot(\C^{d_B})^{\ot m}\bigr),\\[-2pt]
   \widetilde\rho\text{ is invariant under permutations within each block}
  \end{array}
 \right\}.
\]
This is the standard notion of two-sided
extendibility~\cite{TDS03,JV13,JSZ22}.  Requiring the extension to be
supported on the symmetric subspace of each block gives the
Bose-symmetric extendible states
\begin{equation}
 \BEXT_{n,m}(d_A,d_B)
 :=
 \left\{
  \widetilde\rho_{A_1B_1}:
  \widetilde\rho
  \in\D\bigl(\Sym^n(\C^{d_A})\ot\Sym^m(\C^{d_B})\bigr)
 \right\}.
 \label{eq:two-sided-bext-definition}
\end{equation}
This is the two-sided analogue of Bose-symmetric extendibility
in~\cite{NOP09}.  If $n'\ge n$ and $m'\ge m$, tracing out extension
sites and extending every term in a pure-product decomposition by
tensor powers give
\begin{equation}
 \begin{aligned}
  \SEP(d_A,d_B)
  &\subseteq
  \BEXT_{n',m'}(d_A,d_B)
  \subseteq
  \BEXT_{n,m}(d_A,d_B)
  \subseteq
  \EXT_{n,m}(d_A,d_B),\\
  \EXT_{n',m'}(d_A,d_B)
  &\subseteq
  \EXT_{n,m}(d_A,d_B).
 \end{aligned}
 \label{eq:sep-in-two-sided-bext}
\end{equation}

\begin{figure}[H]
\centering
\begin{tikzpicture}[
  font=\scriptsize,
  every node/.style={align=center}
]
\begin{scope}[xshift=-2.05cm]
 \draw[semithick, fill=black!4]
  (5.05,0) ellipse (2.55cm and 1.32cm);
 \draw[semithick, fill=green!10, fill opacity=0.78]
  (4.68,-0.08) ellipse (1.74cm and 0.78cm);
 \draw[semithick, fill=blue!8, fill opacity=0.78]
  (5.42,-0.08) ellipse (1.74cm and 0.78cm);
 \draw[semithick, fill=orange!14]
  (5.05,-0.08) ellipse (0.68cm and 0.31cm);
 \node at (5.05,1.10) {$\EXCH_N^{(t)}(d)$};
 \node at (6.08,0.34) {$\BOS_N^{(t)}(d)$};
 \node at (4.02,0.34) {$\SEPP^{(t)}(d)$};
 \node at (5.05,-0.08) {$\SEPH^{(t)}(d)$};
 \draw[
  {Stealth[length=2.5pt]}-{Stealth[length=2.5pt]},
  densely dotted, black!65, line width=0.5pt
 ]
  (2.57,-0.08) -- (2.91,-0.08);
 \draw[
  {Stealth[length=2.5pt]}-{Stealth[length=2.5pt]},
  densely dotted, black!65, line width=0.5pt
 ]
  (5.73,-0.08) -- (7.08,-0.08);
 \node[font=\scriptsize] at (5.05,-1.58)
  {(a) Bosonic and exchangeable states};
\end{scope}

\begin{scope}[xshift=-1.15cm]
 \draw[semithick, fill=black!4]
  (10.15,0) ellipse (2.35cm and 1.26cm);
 \draw[semithick, fill=blue!8]
  (10.15,-0.12) ellipse (1.68cm and 0.84cm);
 \draw[semithick, fill=orange!14]
  (10.15,-0.32) ellipse (0.98cm and 0.40cm);
 \node at (10.15,0.92) {$\EXT_{n,m}(d_A,d_B)$};
 \node at (10.15,0.32) {$\BEXT_{n,m}(d_A,d_B)$};
 \node at (10.15,-0.32) {$\SEP(d_A,d_B)$};
 \draw[
  {Stealth[length=2.5pt]}-{Stealth[length=2.5pt]},
  densely dotted, black!65, line width=0.5pt
 ]
  (8.54,-0.08) -- (9.28,-0.08);
 \draw[
  {Stealth[length=2.5pt]}-{Stealth[length=2.5pt]},
  densely dotted, black!65, line width=0.5pt
 ]
  (11.02,-0.08) -- (12.42,-0.08);
 \node[font=\scriptsize] at (10.15,-1.58)
  {(b) Two-sided extendible states};
\end{scope}
\end{tikzpicture}
\begin{minipage}{0.8\textwidth}
\scriptsize
\caption{State-set containments underlying the de Finetti theorems.
Nesting records containment.  From left to right, the dotted markers
represent the worst-case de Finetti errors for
$(\EXCH_N^{(t)},\SEPP^{(t)})$, $(\BOS_N^{(t)},\SEPH^{(t)})$,
$(\BEXT_{n,m},\SEP)$, and $(\EXT_{n,m},\SEP)$.  Marker lengths are
schematic.}
\label{fig:state-set-relations}
\end{minipage}
\end{figure}
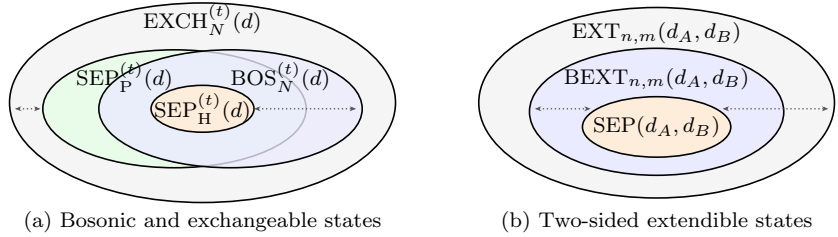

\subsection{Hausdorff notation and convex duality}

For $\star\in\{\mathrm{tr},\mathrm{HS}\}$ and a nonempty compact set
$\mathcal K$ of states, write
\begin{equation}
 D_\star(\rho,\mathcal K)
 :=
 \inf_{\sigma\in\mathcal K}D_\star(\rho,\sigma).
 \label{eq:distance-to-set}
\end{equation}
For two nonempty compact sets $\mathcal K$ and $\mathcal L$ of states,
their Hausdorff distance is
\begin{equation}
\begin{aligned}
 D_{\mathrm H}^{\star}(\mathcal K,\,\mathcal L)
 :=
 \max\left\{
  \sup_{\rho\in\mathcal K}D_\star(\rho,\mathcal L),
  \sup_{\sigma\in\mathcal L}D_\star(\sigma,\mathcal K)
 \right\}.
\end{aligned}
 \label{eq:Hausdorff-distance}
\end{equation}
In particular, if $\mathcal L\subseteq\mathcal K$, then
\begin{equation}
 D_{\mathrm H}^{\star}(\mathcal K,\,\mathcal L)
 =
 \sup_{\rho\in\mathcal K}D_\star(\rho,\mathcal L).
 \label{eq:Hausdorff-outer-relaxation}
\end{equation}

Figure~\ref{fig:state-set-relations} summarizes the relevant
containments.  Consequently, every de Finetti error in the paper has
the form $D_{\mathrm H}^{\star}(\mathcal K,\,\mathcal L)$, with the
larger finite-extension set $\mathcal K$ first and its separable subset
$\mathcal L$ second.  The four instances are
\[
 \begin{gathered}
  \bigl(\BOS_N^{(t)}(d),\SEPH^{(t)}(d)\bigr),\qquad
  \bigl(\EXCH_N^{(t)}(d),\SEPP^{(t)}(d)\bigr),\\
  \bigl(\BEXT_{n,m}(d_A,d_B),\SEP(d_A,d_B)\bigr),\qquad
 \bigl(\EXT_{n,m}(d_A,d_B),\SEP(d_A,d_B)\bigr).
 \end{gathered}
\]
These quantities are exactly the corresponding worst-case de Finetti
errors.

For a nonempty compact set $\mathcal C\subseteq\D(\cH)$, write
\begin{equation}
 h_{\mathcal C}(M)
 :=
 \max_{\rho\in\mathcal C}\Tr(M\rho)
 \qquad
 (M\in\Herm(\cH))
 \label{eq:state-set-support-function}
\end{equation}
for its support function.  The following lemma collects the distance
dualities used throughout the paper.

\begin{lemma}[Distance duality for compact convex state sets]
\label{lem:compact-convex-distance-duality}
Let $\mathcal L\subseteq\D(\cH)$ be nonempty, compact, and convex.
For every $\rho\in\D(\cH)$,
\begin{align}
 \Dtr(\rho,\mathcal L)
 &=
 \max_{0\preceq M\preceq I}
 \left\{\Tr(M\rho)-h_{\mathcal L}(M)\right\}
 \notag\\
 &=
 \frac12
 \max_{\substack{M\in\Herm(\cH)\\\norm{M}_\infty\le1}}
 \left\{\Tr(M\rho)-h_{\mathcal L}(M)\right\},
 \label{eq:compact-duality-state-trace}\\
 \DHS(\rho,\mathcal L)
 &=
 \max_{\substack{M\in\Herm(\cH)\\\norm{M}_2\le1}}
 \left\{\Tr(M\rho)-h_{\mathcal L}(M)\right\}
 \notag\\
 &=
 \max_{\substack{M\in\Herm(\cH)\\
                   \Tr M=0,\ \norm{M}_2\le1}}
 \left\{\Tr(M\rho)-h_{\mathcal L}(M)\right\}.
 \label{eq:compact-duality-state-HS}
\end{align}
If $\mathcal K\subseteq\D(\cH)$ is also nonempty, compact, and convex,
and $\mathcal L\subseteq\mathcal K$, then
\begin{align}
 D_{\mathrm H}^{\mathrm{tr}}(\mathcal K,\,\mathcal L)
 &=
 \max_{0\preceq M\preceq I}
 \left\{h_{\mathcal K}(M)-h_{\mathcal L}(M)\right\}
 \notag\\
 &=
 \frac12
 \max_{\substack{M\in\Herm(\cH)\\\norm{M}_\infty\le1}}
 \left\{h_{\mathcal K}(M)-h_{\mathcal L}(M)\right\},
 \label{eq:compact-duality-sets-trace}\\
 D_{\mathrm H}^{\mathrm{HS}}(\mathcal K,\,\mathcal L)
 &=
 \max_{\substack{M\in\Herm(\cH)\\
                   \Tr M=0,\ \norm{M}_2\le1}}
 \left\{h_{\mathcal K}(M)-h_{\mathcal L}(M)\right\}.
 \label{eq:compact-duality-sets-HS}
\end{align}
\end{lemma}

\begin{proof}
For states $\rho$ and $\sigma$, trace-norm duality and
$\Tr(\rho-\sigma)=0$ give
\[
 \Dtr(\rho,\sigma)
 =
 \max_{0\preceq M\preceq I}\Tr\bigl(M(\rho-\sigma)\bigr)
 =
 \frac12
 \max_{\substack{M\in\Herm(\cH)\\\norm{M}_\infty\le1}}
 \Tr\bigl(M(\rho-\sigma)\bigr).
\]
Hilbert--Schmidt self-duality similarly gives
\[
 \DHS(\rho,\sigma)
 =
 \max_{\substack{M\in\Herm(\cH)\\\norm{M}_2\le1}}
 \Tr\bigl(M(\rho-\sigma)\bigr).
\]
In each case, Sion's finite-dimensional minimax
theorem~\cite{Sio58} exchanges the minimum over
$\sigma\in\mathcal L$ with the maximum over the corresponding compact
convex witness set.  This proves both identities in
\eqref{eq:compact-duality-state-trace} and the first identity in
\eqref{eq:compact-duality-state-HS}.

To obtain the formula with $\Tr M=0$, let $q=\dim\cH$ and set
\[
 M_0:=M-\frac{\Tr M}{q}I.
\]
Replacing $M$ by $M_0$ leaves
$\Tr(M\rho)-h_{\mathcal L}(M)$ unchanged, while orthogonality to the
identity gives $\norm{M_0}_2\le\norm M_2$.  Thus the maximum may be
restricted to traceless witnesses.

Finally, take the maximum over $\rho\in\mathcal K$ in the state-to-set
formulas.  The two remaining maxima may be interchanged, and
maximization over $\rho$ replaces $\Tr(M\rho)$ by
$h_{\mathcal K}(M)$.  Equation~\eqref{eq:Hausdorff-outer-relaxation}
then gives \eqref{eq:compact-duality-sets-trace} and
\eqref{eq:compact-duality-sets-HS}.
\end{proof}

\subsection{Linear-algebra estimates}

\begin{fact}[Positive block estimates]
\label{fact:positive-block}
Let $A\succeq0$ act on $\cH_A$, let $B\succeq0$ act on $\cH_B$, and
let $X:\cH_B\to\cH_A$.  Then
\[
 \begin{pmatrix}
  A&X\\
  X^\dagger&B
 \end{pmatrix}
 \succeq0
\]
if and only if
\[
 X=A^{1/2}KB^{1/2}
\]
for a contraction $K:\cH_B\to\cH_A$.  In this case,
\begin{equation}
 \left\|
 \begin{pmatrix}
  A&X\\
  X^\dagger&B
 \end{pmatrix}
 \right\|_\infty
 \le
 \norm{A}_\infty+\norm{B}_\infty,
 \qquad
 \norm{X}_2^2
 \le
 \norm{B}_\infty\Tr A.
 \label{eq:positive-block-consequences}
\end{equation}
\end{fact}

\begin{proof}
The factorization statement is
\cite[Proposition~1.3.2]{Bha07}.  It remains to verify the norm
estimates.  For $a:=\norm{A}_\infty$ and $b:=\norm{B}_\infty$,
\[
 \bra{u\oplus v}
 \begin{pmatrix}A&X\\X^\dagger&B\end{pmatrix}
 \ket{u\oplus v}
 \le
 \left(\sqrt a\norm u+\sqrt b\norm v\right)^2
 \le
 (a+b)\left(\norm u^2+\norm v^2\right),
\]
which proves the operator-norm bound.  The factorization also gives
\[
 XX^\dagger
 =
 A^{1/2}KBK^\dagger A^{1/2}
 \preceq
 \norm{B}_\infty A,
\]
and taking traces proves the bound on $\norm{X}_2$.
\end{proof}

\begin{fact}[Discarding a positive block]
\label{fact:positive-block-discard}
Under the hypotheses of Fact~\ref{fact:positive-block}, suppose that
\[
 M:=
 \begin{pmatrix}
  A&X\\
  X^\dagger&B
 \end{pmatrix}
 \succeq0.
\]
Then
\begin{equation}
 \left\|
 M-
 \begin{pmatrix}
  A&0\\
  0&0
 \end{pmatrix}
 \right\|_2^2
 \le
 \norm{B}_\infty\bigl(2\Tr A+\Tr B\bigr).
 \label{eq:positive-block-discard}
\end{equation}
In particular, if $\Tr M=1$, then the right-hand side is at most
$2\norm{B}_\infty$.
\end{fact}

\begin{proof}
Block orthogonality in Hilbert--Schmidt inner product gives
\[
 \left\|
 \begin{pmatrix}
  0&X\\
  X^\dagger&B
 \end{pmatrix}
 \right\|_2^2
 =
 2\norm{X}_2^2+\norm{B}_2^2.
\]
Since $\norm{B}_2^2\le\norm{B}_\infty\Tr B$,
the bound on $\norm{X}_2$ in
\eqref{eq:positive-block-consequences} proves
\eqref{eq:positive-block-discard}.
\end{proof}

\subsection{Tensor norms and symmetric tensors}
\label{sec:tensor-norms}

\paragraph{Injective and projective tensor norms.}
Let $\cX$ and $\cY$ be finite-dimensional Hilbert spaces.  For
$z\in\cX\ot\cY$, define the injective and projective tensor norms by
\begin{align}
 \norm{z}_{\varepsilon}
 &:={}
 \sup_{\substack{\norm{x}=1\\\norm{y}=1}}
 \abs{\braket{x\ot y}{z}},
 \label{eq:injective-tensor-norm}\\
 \norm{z}_{\pi}
 &:={}
 \inf\left\{
  \sum_j\abs{a_j}:
  z=\sum_j a_jx_j\ot y_j,
  \ \norm{x_j}=\norm{y_j}=1
 \right\}.
 \label{eq:projective-tensor-norm}
\end{align}
The unadorned norm $\norm{z}$ remains the Hilbert-space norm.  After
choosing orthonormal bases of $\cX$ and $\cY$, let $Z$ be the coefficient
matrix of $z$.  The singular-value decomposition, equivalently the
Schmidt decomposition of $z$, gives
\begin{equation}
 \norm{z}_{\varepsilon}=\norm{Z}_\infty,
 \qquad
 \norm{z}=\norm{Z}_2,
 \qquad
 \norm{z}_{\pi}=\norm{Z}_1.
 \label{eq:tensor-norm-matrix}
\end{equation}
In particular, if $\operatorname{Schmidt-rank}(z)$ denotes the number
of nonzero Schmidt coefficients of $z$, Schatten-norm duality and
Cauchy--Schwarz give
\begin{equation}
 \abs{\braket{z}{w}}
 \le
 \norm{z}_{\pi}\norm{w}_{\varepsilon},
 \qquad
 \norm{z}_{\pi}
 \le
 \sqrt{\operatorname{Schmidt-rank}(z)}\,\norm{z}.
 \label{eq:tensor-norm-duality-rank}
\end{equation}

The orthogonal projection onto $\Sym^k\cH$ is the symmetrization
operator
\begin{equation}
 \Pi_{\mathrm{sym}}^{(k)}
 :=
 \frac1{k!}\sum_{\pi\in\mathfrak S_k}U_\pi.
 \label{eq:symmetrization-projector}
\end{equation}
For a standard treatment of the symmetric subspace, see
Harrow~\cite{Har13}.

\medskip
\paragraph{Orthogonal decompositions.}
An orthogonal decomposition of the one-site space induces a
corresponding decomposition of every symmetric tensor power; see, for
example,~\cite[Appendix~B, Eq.~(B.2)]{FH91}.  We record the precise
form and its number-operator interpretation used in
Section~\ref{sec:hs-definetti}.

\begin{fact}[Symmetric power of an orthogonal sum]
\label{fact:symmetric-orthogonal-sum}
Let $\cH=\cH_0\oplus\cH_1$ be an orthogonal decomposition.  For every
integer $n\ge1$, there is a natural orthogonal decomposition
\begin{equation}
 \Sym^n\cH
 \cong
 \bigoplus_{m=0}^n
 \Sym^m(\cH_0)\ot\Sym^{n-m}(\cH_1).
 \label{eq:symmetric-orthogonal-sum}
\end{equation}
If $P$ is the projection onto $\cH_0$, the summand indexed by $m$ is
the eigenvalue-$m$ eigenspace of $\sum_{j=1}^nP_j$ on $\Sym^n\cH$,
where
\[
 P_j
 =
 I^{\ot(j-1)}\ot P\ot I^{\ot(n-j)}.
\]
For $n=2$, the decomposition is
\begin{equation}
 \Sym^2\cH
 =
 \Sym^2(\cH_0)
 \oplus(\cH_0\vee\cH_1)
 \oplus\Sym^2(\cH_1),
 \label{eq:symmetric-two-summands}
\end{equation}
where $\cH_0\vee\cH_1$ is spanned by the vectors
\[
 p\vee q
 :=
 \frac{p\ot q+q\ot p}{\sqrt2},
 \qquad
 p\in\cH_0,\quad q\in\cH_1.
\]
\end{fact}

\medskip
\paragraph{Partial inner products.}

For $k\ge1$ and $u\in\cH$, define the partial-inner-product map
\[
 C_u:\Sym^k\cH\longrightarrow\Sym^{k-1}\cH,
 \qquad
 \Psi\longmapsto(\bra u\ot I)\Psi.
\]
For $k\ge2$ and $\phi\in\Sym^2\cH$, similarly define
\[
 C_\phi:\Sym^k\cH\longrightarrow\Sym^{k-2}\cH,
 \qquad
 \Psi\longmapsto(\bra\phi\ot I)\Psi.
\]
Here $I$ acts on the remaining tensor factors, and the symmetric
degree will always be clear from context.  When $C_u$ acts on a tensor
of symmetric degree zero, we set it equal to zero.

Equivalent calculations in the standard bosonic formalism appear
in~\cite[Sec.~4]{LNR15}; we use only partial-inner-product notation
here.

\begin{fact}[Partial inner products and marginals]
\label{fact:symmetric-partial-inner-products}
For $\Xi\in\Sym^{k-1}\cH$, the adjoint of $C_u$ is
\begin{equation}
 C_u^*(\Xi)
 =
 \Pi_{\mathrm{sym}}^{(k)}(u\ot\Xi)
 =
 \frac1k\sum_{j=1}^k U_{1j}(u\ot\Xi),
 \label{eq:symmetric-insertion}
\end{equation}
where $U_{1j}$ swaps the first and $j$th tensor factors, with
$U_{11}=I$.  Moreover, if $k\ge2$, $\Psi\in\Sym^k\cH$ is a unit vector, and
$\rho_1,\rho_2$ are the one- and two-site marginals of
$\ketbra{\Psi}$, then
\begin{equation}
 \bra u\rho_1\ket u=\norm{C_u(\Psi)}^2,
 \qquad
 \bra\phi\rho_2\ket\phi=\norm{C_\phi(\Psi)}^2.
 \label{eq:contraction-marginals}
\end{equation}
\end{fact}

\begin{fact}[Bounds for symmetric partial inner products]
\label{fact:symmetric-contraction-bessel}
For every $u\in\cH$ and $\Xi\in\Sym^{k-1}\cH$,
\begin{equation}
 \norm{C_u^*(\Xi)}^2
 =
 \frac1k
 \left(
  \norm{u}^2\norm{\Xi}^2
  +(k-1)\norm{C_u(\Xi)}^2
 \right).
 \label{eq:contraction-adjoint-norm}
\end{equation}
For every orthonormal family $(u_i)$ and every symmetric tensor
$\Psi$,
\begin{equation}
 \sum_i\norm{C_{u_i}(\Psi)}^2
 \le
 \norm{\Psi}^2.
 \label{eq:contraction-Bessel}
\end{equation}
Consequently, if $u_1,\ldots,u_\ell$ are orthonormal, then
\begin{equation}
 \sum_{i=1}^\ell\norm{C_{u_i}^*(\Xi)}^2
 \le
 \frac{\ell+k-1}{k}\norm{\Xi}^2.
 \label{eq:insertion-Bessel}
\end{equation}
In particular, the right-hand side is at most $2\norm{\Xi}^2$ when
$\ell\le k+1$.
\end{fact}

\begin{proof}
The case $k=1$ of \eqref{eq:contraction-adjoint-norm} is immediate.
For $k\ge2$, expanding the symmetrizer in
\eqref{eq:symmetric-insertion} gives
\begin{align*}
 \norm{C_u^*(\Xi)}^2
 &=
 \frac1k\sum_{j=1}^k
 \left\langle
  u\ot\Xi,U_{1j}(u\ot\Xi)
 \right\rangle\\
 &=\frac1k
 \left(
  \norm{u}^2\norm{\Xi}^2
  +(k-1)\norm{C_u(\Xi)}^2
 \right).
\end{align*}
Indeed, symmetry of $\Xi$ makes the terms with $j\ge2$ equal, and
their common value is
$\langle\Xi,(\ketbra{u}\ot I)\Xi\rangle
=\norm{C_u(\Xi)}^2$.  This proves
\eqref{eq:contraction-adjoint-norm}.

For \eqref{eq:contraction-Bessel}, let
$P_U:=\sum_i\ketbra{u_i}$ be the orthogonal projection onto the span
of the orthonormal family $(u_i)$.  Then
\begin{align*}
 \sum_i\norm{C_{u_i}(\Psi)}^2
 &=
 \sum_i
 \left\langle
  \Psi,\bigl(\ketbra{u_i}\ot I\bigr)\Psi
 \right\rangle\\
 &=
 \left\langle\Psi,(P_U\ot I)\Psi\right\rangle\\
 &\le
 \norm{\Psi}^2,
\end{align*}
because $0\preceq P_U\preceq I$.  Summing
\eqref{eq:contraction-adjoint-norm} over $u_1,\ldots,u_\ell$ and
applying \eqref{eq:contraction-Bessel} to $\Xi$ gives
\eqref{eq:insertion-Bessel}.
\end{proof}

\section{Two-site de Finetti upper bounds}
\label{sec:upper}

We first prove the bosonic two-site bound by argmax rounding.
Symmetric purification gives the exchangeable bound.  Rectangular
argmax rounding gives the two-sided Bose-symmetric bound, and a second
symmetric purification removes the Bose-symmetry assumption.

\subsection{Bosonic states: the argmax principle}

\begin{theorem}[Bosonic two-site de Finetti upper bound]
\label{thm:bos-main}
For every $N,d\ge2$,
\begin{equation}
 \DHtr\Bigl(
  \BOS_N^{(2)}(d),\,
  \SEPH^{(2)}(d)
 \Bigr)
 \le
 \frac{\sqrt{d-1}}{2(N-1)}.
 \label{eq:bos-upper-main}
\end{equation}
\end{theorem}

To prove the theorem, associate with each
$M\in\Herm(\Sym^2(\cH))$ the quantities $a_M$, $h(M)$, and the
normalized $N$-site lift $M^{[N]}$ defined by
\begin{align}
 a_M(u)&:=\bra{u^{\ot2}}M\ket{u^{\ot2}},\\
 h(M)&:=\max_{\norm{u}=1}a_M(u),
 \label{eq:hA}\\
 M^{[N]}&:=\binom N2^{-1}\sum_{1\le i<j\le N}M_{ij}
 \quad\text{on }\Sym^N(\cH).
 \label{eq:N-site-lift}
\end{align}
Here $M_{ij}$ acts as $M$ on sites $i,j$ and as the identity elsewhere.  Since every two-site marginal of a bosonic state is supported on $\Sym^2(\cH)$, no extension of $M$ outside that subspace is needed.

The proof has three steps.  First, trace-norm duality identifies the de
Finetti error with the largest gap
$\lambda_{\max}(M^{[N]})-h(M)$ over bounded Hermitian $M$.  Next, given
a top eigenvector $\psi$ of $M^{[N]}$, the argmax principle chooses as
the rounding object a unit vector $u$ maximizing
$\abs{\braket{u^{\ot N}}{\psi}}$.  The complex-sphere argmax lemma
provides the basic tools for analyzing this rounding object.  After
taking the partial inner product with $u$ on the other $N-2$ sites,
first-order optimality eliminates the mixed component of the resulting
two-site vector, while second-order optimality controls its component
in $\Sym^2(u^\perp)$.  Finally, Takagi decomposition converts this
structural information into the integrality-gap bound
\[
 \lambda_{\max}(M^{[N]})-h(M)
 \le
 \frac{\sqrt{d-1}}{N-1}\norm{M}_\infty.
\]

\subsubsection{An exact dual problem}

Trace-norm duality reduces the approximation problem to a gap detected
by a bounded Hermitian witness $M$.  Its maximum over Hartree states is
$h(M)$, whereas its maximum over two-site marginals of bosonic
$N$-site states is $\lambda_{\max}(M^{[N]})$.  The following lemma makes
this reduction exact.

\begin{lemma}[Exact trace-norm duality]
\label{lem:dual-bos}
For every $N,d\ge2$,
\begin{equation}
 \DHtr\Bigl(
  \BOS_N^{(2)}(d),\,
  \SEPH^{(2)}(d)
 \Bigr)
 =\frac12
 \sup_{\substack{M\in\Herm(\Sym^2\cH)\\
                  \norm{M}_\infty\le1}}
 \left\{\lambda_{\max}(M^{[N]})-h(M)\right\}.
 \label{eq:dual-bos}
\end{equation}
\end{lemma}

\begin{proof}
The support function of the Hartree mixtures is
\[
 h_{\SEPH^{(2)}(d)}(M)=h(M).
\]
For the bosonic marginal set, permutation symmetry gives
\[
 h_{\BOS_N^{(2)}(d)}(M)
 =
 \max_{\rho_N\in\D(\Sym^N\cH)}
 \Tr(M^{[N]}\rho_N)
 =
 \lambda_{\max}(M^{[N]}).
\]
By duality (Lemma~\ref{lem:compact-convex-distance-duality}), these two
identities prove \eqref{eq:dual-bos}.
\end{proof}

\subsubsection{The argmax rounding object}

The tensor-power vectors $\{u^{\ot N}:u\in\cH\}$ span
 $\Sym^N(\cH)$; see also~\cite[Theorem~3]{Har13}.  Thus, for every
 nonzero $\psi\in\Sym^N(\cH)$, the argmax principle gives a natural
 rounding object: choose
\[
 u\in\operatorname*{arg\,max}_{\norm{v}=1}
 \abs{\braket{v^{\ot N}}{\psi}},
\]
and use $u^{\ot N}$ as the rounding object.
Fix a unit vector $\psi\in\Sym^N(\cH)$, and let $u^{\ot N}$ be this
argmax rounding object.  After adjusting the phase of $u$, set
\begin{equation}
 c:=\braket{u^{\ot N}}{\psi}>0.
 \label{eq:c-def}
\end{equation}
Taking the partial inner product with $u$ on $N-2$ sites gives
\begin{equation}
 \eta_u:=
 (\bra{u}^{\ot(N-2)}\ot I_{\cH^{\ot2}})\psi
 \in\Sym^2(\cH).
 \label{eq:eta-def}
\end{equation}
Write
\[
 u\vee u^\perp
 :=\operatorname{span}\{u\ot v+v\ot u:v\perp u\}
 \subseteq\Sym^2(\cH).
\]

The next lemma provides the basic tools for analyzing this rounding
object.  It is the bosonic counterpart of the reweighted
pseudo-covariance lemma underlying our SoS argmax
rounding~\cite{JWX26}.  Here $\eta_u$ plays the role of the reweighted
pseudo-covariance.  The mechanism is concrete:
first-order optimality eliminates its mixed $u\vee u^\perp$ component,
while second-order optimality bounds the remaining
$\Sym^2(u^\perp)$ component against every square $v^{\ot2}$.
\begin{lemma}[Complex-sphere argmax estimate]
\label{lem:complex-sphere-argmax}
With the notation above,
\begin{equation}
 \eta_u=c\,u^{\ot2}+\zeta_u,
 \qquad
 \zeta_u\in\Sym^2(u^\perp),
 \label{eq:eta-decomp}
\end{equation}
and
\begin{equation}
 \sup_{\substack{v\perp u\\\norm{v}=1}}
 \abs{\braket{v^{\ot2}}{\zeta_u}}
 \le\frac{c}{N-1}.
 \label{eq:argmax-injective}
\end{equation}
\end{lemma}

\begin{proof}
Fix a unit vector $v\perp u$ and, for $t\in\C$, set
\[
 u_t:=\frac{u+t v}{\sqrt{1+\abs t^2}},
 \qquad
 g(t):=\braket{u_t^{\ot N}}{\psi}.
\]
By the choice of $u$, the function $\abs{g(t)}^2$ is maximized at
$t=0$.  For $0\le k\le N$, write
\[
 \gamma_k:=
 \braket{v^{\ot k}\ot u^{\ot(N-k)}}{\psi}.
\]
Symmetry of $\psi$ gives the exact expansion
\[
 g(t)
 =
 (1+\abs t^2)^{-N/2}
 \sum_{k=0}^N\binom Nk\overline t^{\,k}\gamma_k.
\]
In particular, $\gamma_0=c$.  Set
\begin{align*}
 \alpha&:=\gamma_1
 =\braket{v\ot u^{\ot(N-1)}}{\psi},\\
 \beta&:=\gamma_2
 =\braket{v^{\ot2}\ot u^{\ot(N-2)}}{\psi}.
\end{align*}
The exact expansion yields
\begin{equation}
 g(t)
 =(1+\abs t^2)^{-N/2}
 \left(c+N\overline t\,\alpha+\binom N2\overline t^{\,2}\beta
 +O(\abs t^3)\right).
 \label{eq:g-expansion}
\end{equation}
The normalization factor has no linear term, and hence the first-order
term in $\abs{g(t)}^2$ is
\[
 2Nc\,\Re(\overline t\alpha).
\]
To see why this forces $\alpha=0$, fix any $\xi\in\C$ and restrict to
the real one-parameter curve $t=r\xi$, $r\in\mathbb{R}$.  Since $t=0$ is a
local maximum of $\abs{g(t)}^2$, the differentiable function
$r\mapsto\abs{g(r\xi)}^2$ has derivative zero at $r=0$.  Thus
\[
 0
 =
 \left.\frac{d}{dr}\abs{g(r\xi)}^2\right|_{r=0}
 =
 2Nc\,\Re(\overline\xi\alpha).
\]
Taking $\xi=1$ gives $\Re(\alpha)=0$, while taking $\xi=i$ gives
$\Im(\alpha)=0$.  Therefore $\alpha=0$.  Keeping the normalization
factor explicit in the first line and then Taylor-expanding it in the
second, we obtain
\begin{align}
 \abs{g(t)}^2
 &=(1+\abs t^2)^{-N}
 \left(
 c^2
 +2\binom N2c\,\Re(\overline t^{\,2}\beta)
 +O(\abs t^3)
 \right)\notag\\
 &=
 \left(1-N\abs t^2+O(\abs t^4)\right)
 \left(
 c^2
 +2\binom N2c\,\Re(\overline t^{\,2}\beta)
 +O(\abs t^3)
 \right)\notag\\
 &=c^2-Nc^2\abs t^2
 +2\binom N2c\,\Re(\overline t^{\,2}\beta)
 +O(\abs t^3).
\label{eq:g2-expansion}
\end{align}
Here we used
$(1+\abs t^2)^{-N}=1-N\abs t^2+O(\abs t^4)$.  Its product with the
constant term $c^2$ contributes
$-Nc^2\abs t^2$.  Its products with terms of order $\abs t^2$ or
higher are $O(\abs t^4)$ and are therefore absorbed into the
$O(\abs t^3)$ remainder.

Now write $t=re^{i\theta}$ and choose $\theta$ so that
$e^{-2i\theta}\beta=\abs\beta$.  Along this real one-parameter
perturbation,
\[
 \abs{g(re^{i\theta})}^2
 =
 c^2+
 \left(
 -Nc^2+2\binom N2c\abs\beta
 \right)r^2
 +O(r^3).
\]
Since $r=0$ is a local maximum, the quadratic coefficient is
nonpositive:
\[
 2\binom N2c\abs\beta\le Nc^2.
\]
Therefore $\abs\beta\le c/(N-1)$.

It remains to translate these variational conditions into a statement
about $\eta_u$.  The orthogonal decomposition
\[
 \Sym^2(\cH)
 =
 \C u^{\ot2}
 \oplus (u\vee u^\perp)
 \oplus \Sym^2(u^\perp)
\]
is adapted to this partial inner product.  By definition,
\[
 \braket{u^{\ot2}}{\eta_u}=c,
 \qquad
 \braket{v\ot u}{\eta_u}=\alpha=0
 \quad(v\perp u).
\]
Since $\eta_u$ is symmetric, the second identity eliminates its entire
$u\vee u^\perp$ component.  Thus
$\eta_u=c\,u^{\ot2}+\zeta_u$ for some
$\zeta_u\in\Sym^2(u^\perp)$.  Finally, for every unit $v\perp u$,
\[
 \braket{v^{\ot2}}{\zeta_u}
 =\braket{v^{\ot2}}{\eta_u}
 =\beta,
\]
so the bound on $\beta$ is exactly \eqref{eq:argmax-injective}.
\end{proof}


\subsubsection{Bounding the integrality gap}

We now use the information supplied by the argmax lemma to bound the
integrality gap.  The remaining ingredient is the Takagi
decomposition, the symmetric form of the singular-value, or Schmidt,
decomposition.  It says that a symmetric two-tensor can be written as
\[
 \zeta=\sum_i s_i u_i^{\ot2},
\]
where the $u_i$ are orthonormal and $s_i\ge0$.  This form identifies
the relevant tensor norms with matrix Schatten norms.

To record the corresponding matrix statement, let $\cH$ be an
$m$-dimensional complex Hilbert space with an orthonormal basis
$(e_i)_{i=1}^m$.  Every $\zeta\in\Sym^2(\cH)$ has a unique representation
\begin{equation}
 \zeta=\sum_{i,j=1}^m Z_{ij}e_i\ot e_j,
 \qquad Z^\mathsf T=Z.
 \label{eq:matrix-tensor}
\end{equation}
We write $Z=\Mat(\zeta)$.  Although $Z$ need
not be Hermitian, its complex symmetry implies a Takagi
decomposition~\cite[Cor.~4.4.4]{HJ12}
\begin{equation}
 Z=U\Sigma U^\mathsf T,
 \label{eq:takagi-factorization}
\end{equation}
where $U$ is unitary and $\Sigma$ is diagonal with nonnegative entries.
The diagonal entries of $\Sigma$ are the singular values of $Z$.  A
change of orthonormal basis replaces $Z$ by a unitary congruence, so its
Schatten norms are basis independent.

The injective and projective tensor norms from
Section~\ref{sec:tensor-norms} have especially simple descriptions on
$\Sym^2(\cH)$.  Namely, Takagi decomposition shows that they may be
computed using only symmetric rank-one tensors.

\begin{fact}[Takagi decomposition and symmetric tensor norms]
\label{lem:takagi}
For every $\zeta\in\Sym^2(\cH)$,
\begin{align}
 \norm{\zeta}_{\varepsilon}
 &=
 \sup_{\norm{v}=1}
 \abs{\braket{v^{\ot2}}{\zeta}}
 =\norm{\Mat(\zeta)}_\infty,
 \label{eq:takagi-op}\\
 \norm{\zeta}_{\pi}
 &=
 \inf\left\{
  \sum_k\abs{a_k}:
  \zeta=\sum_k a_kv_k^{\ot2},\
  \norm{v_k}=1
 \right\}
 =\norm{\Mat(\zeta)}_1.
 \label{eq:takagi-projective}
\end{align}
\end{fact}

\begin{proof}
Write $Z=U\Sigma U^\mathsf T$ as in
\eqref{eq:takagi-factorization}.  Our inner-product convention gives
\[
 \braket{v^{\ot2}}{\zeta}
 =\overline v^{\mathsf T}Z\overline v.
\]
As $v$ ranges over the unit sphere, so does
$y=U^\mathsf T\overline v$.  Hence
\[
 \abs{\braket{v^{\ot2}}{\zeta}}
 =
 \abs{\sum_i\Sigma_{ii}y_i^2}
 \le
 \sum_i\Sigma_{ii}\abs{y_i}^2
 \le
 \max_i\Sigma_{ii}.
\]
Taking $y$ to be a coordinate vector corresponding to the largest
diagonal entry shows that the supremum over squares equals
$\norm{Z}_\infty$.  By \eqref{eq:tensor-norm-matrix}, this is precisely
$\norm{\zeta}_\varepsilon$, proving \eqref{eq:takagi-op}.

If $u_i$ are the columns of $U$, the Takagi decomposition gives
\[
 \zeta=\sum_i\Sigma_{ii}u_i^{\ot2},
\]
so the infimum over square decompositions is at most
$\sum_i\Sigma_{ii}=\norm{Z}_1=\norm{\zeta}_\pi$.  Conversely, every
square decomposition is admissible in the definition of the ordinary
projective norm \eqref{eq:projective-tensor-norm}, so the infimum is at
least $\norm{\zeta}_\pi$.  This proves
\eqref{eq:takagi-projective}.
\end{proof}

For $M\in\Herm(\Sym^2(\C^d))$, define
\begin{equation}
 \tau(M):=
 \sup_{\norm{u}=1}
 \norm{P_{\Sym^2(u^\perp)}Mu^{\ot2}}_{\pi},
 \label{eq:tau-def}
\end{equation}
where the projective norm is taken across the two copies of $u^\perp$.

\begin{proposition}[Integrality gap from argmax rounding]
\label{prop:tau-bound}
For every $N,d\ge2$ and every $M\in\Herm(\Sym^2(\C^d))$,
\begin{equation}
 \lambda_{\max}(M^{[N]})-h(M)
 \le\frac{\tau(M)}{N-1}.
 \label{eq:tau-spectral}
\end{equation}
Consequently,
\begin{equation}
 \lambda_{\max}(M^{[N]})-h(M)
 \le\frac{\sqrt{d-1}}{N-1}\norm{M}_\infty.
 \label{eq:uniform-spectral}
\end{equation}
\end{proposition}

\begin{proof}
Let $\psi\in\Sym^N(\cH)$ be a unit top eigenvector of $M^{[N]}$, with
eigenvalue $\lambda$.  Apply the argmax principle to $\psi$, and let
$u^{\ot N}$ be the resulting rounding object, with $c$ and $\eta_u$ as
in Lemma~\ref{lem:complex-sphere-argmax}.  Since both $u^{\ot N}$ and
$\psi$ are symmetric,
\begin{align}
 \lambda c
 &=\braket{u^{\ot N}}{M^{[N]}\psi}
 =\braket{u^{\ot N}}{M_{12}\psi}
 =\braket{Mu^{\ot2}}{\eta_u}.
 \label{eq:eigen-contraction}
\end{align}
Decompose orthogonally
\begin{equation}
 Mu^{\ot2}
 =a_M(u)u^{\ot2}+\xi_u+B_u,
 \qquad
 \xi_u\in u\vee u^\perp,
 \quad
 B_u\in\Sym^2(u^\perp).
 \label{eq:Au-decomp}
\end{equation}
Combining \eqref{eq:eigen-contraction} and \eqref{eq:eta-decomp} with
the orthogonal decomposition \eqref{eq:Au-decomp} gives
\begin{align}
 (\lambda-a_M(u))c
 &=\braket{Mu^{\ot2}}{\eta_u}-a_M(u)c
 \notag\\
 &=
 \braket{a_M(u)u^{\ot2}+\xi_u+B_u}
         {c\,u^{\ot2}+\zeta_u}
 -a_M(u)c
 \notag\\
 &=\braket{B_u}{\zeta_u}.
 \label{eq:key-pairing}
\end{align}
Lemma~\ref{lem:complex-sphere-argmax} and \eqref{eq:takagi-op} imply
\[
 \norm{\zeta_u}_{\varepsilon}
 \le\frac{c}{N-1}.
\]
Taking absolute values in \eqref{eq:key-pairing}, using
\eqref{eq:tensor-norm-duality-rank}, and dividing by $c>0$ gives
\[
 \abs{\lambda-a_M(u)}
 \le\frac{\norm{B_u}_{\pi}}{N-1}
 \le\frac{\tau(M)}{N-1}.
\]
Because $a_M(u)\le h(M)$,
\[
 \lambda-h(M)
 \le \lambda-a_M(u)
 \le \abs{\lambda-a_M(u)},
\]
which proves \eqref{eq:tau-spectral}.

Choose an orthonormal basis of $u^\perp$, and let $\Mat_u(B_u)$ be the
corresponding complex symmetric matrix from \eqref{eq:matrix-tensor}.
Its Schatten norms are independent of this choice.  By
\eqref{eq:takagi-projective}, and because $\Mat_u(B_u)$ has size
$(d-1)\times(d-1)$,
\begin{align*}
 \norm{B_u}_{\pi}
 &=\norm{\Mat_u(B_u)}_1\\
 &\le\sqrt{d-1}\norm{\Mat_u(B_u)}_2\\
 &=\sqrt{d-1}\norm{B_u}\\
 &\le\sqrt{d-1}\norm{M}_\infty.
\end{align*}
Taking the supremum over $u$ gives \eqref{eq:uniform-spectral}.
\end{proof}

\begin{proof}[Proof of Theorem~\ref{thm:bos-main}]
Insert \eqref{eq:uniform-spectral} into \eqref{eq:dual-bos}.  This gives
\[
 \DHtr\Bigl(
  \BOS_N^{(2)}(d),\,
  \SEPH^{(2)}(d)
 \Bigr)
 \le\frac{\sqrt{d-1}}{2(N-1)}.\qedhere
\]
\end{proof}

\begin{remark}
\label{rem:tau}
The dimension factor enters only through the inequality
$\norm{\xi}_{\pi}\le\sqrt{d-1}\norm{\xi}$ for
$\xi\in\Sym^2(u^\perp)$.  Thus the stronger estimate
\eqref{eq:tau-spectral} is dimension independent for classes of
Hermitian operators
satisfying $\tau(M)=O(\norm{M}_\infty)$.
\end{remark}

\subsubsection{Exchangeable states via symmetric purification}
\label{sec:perm}

\begin{corollary}[Exchangeable two-site de Finetti upper bound]
\label{thm:perm-main}
For every $N,d\ge2$,
\begin{equation}
 \DHtr\Bigl(
  \EXCH_N^{(2)}(d),\,
  \SEPP^{(2)}(d)
 \Bigr)
 \le
 \frac{\sqrt{d^2-1}}{2(N-1)}.
 \label{eq:perm-upper-main}
\end{equation}
\end{corollary}

The corollary follows from Theorem~\ref{thm:bos-main} through the
standard symmetric-purification lemma of Christandl, K\"onig,
Mitchison, and Renner~\cite[Lemma~II.5]{CKMR07}.

\begin{lemma}[Symmetric purification]
\label{lem:symmetric-purification}
Let $\sigma_N$ be exchangeable on $\cH^{\ot N}$, where $\dim\cH=d$.  There is a Hilbert space $\cH_E\cong\cH$ and a unit vector
\[
 \psi_N\in\Sym^N(\cH_E\ot\cH)
\]
such that
\[
 \Tr_{\cH_E^{\ot N}}\ketbra{\psi_N}=\sigma_N.
\]
\end{lemma}

\begin{proof}[Proof of Corollary~\ref{thm:perm-main}]
Apply Lemma~\ref{lem:symmetric-purification} and set
$\widetilde\rho_N=\ketbra{\psi_N}$.  This is a bosonic state with
one-site dimension $d^2$.  Theorem~\ref{thm:bos-main} gives a
probability measure $\mu$ on unit vectors $\phi\in\cH_E\ot\cH$ such
that
\[
 \Dtr\left(
 \widetilde\rho_N^{(2)},
 \int\ketbra{\phi}^{\ot2}\,d\mu(\phi)
 \right)
 \le\frac{\sqrt{d^2-1}}{2(N-1)}.
\]
For each $\phi$, define
$\tau_\phi=\Tr_{\cH_E}\ketbra\phi\in\D(\cH)$.  Sitewise partial trace
maps $\ketbra\phi^{\ot2}$ to $\tau_\phi^{\ot2}$, so contractivity gives
\[
 \Dtr\left(
 \sigma_N^{(2)},
 \int\tau_\phi^{\ot2}\,d\mu(\phi)
 \right)
 \le\frac{\sqrt{d^2-1}}{2(N-1)}.
\]
The approximating state belongs to $\SEPP^{(2)}(d)$, which proves
\eqref{eq:perm-upper-main}.
\end{proof}

\subsection{Two-sided extendibility}
\label{sec:two-sided-bose-trace}

\begin{theorem}[De Finetti upper bound for two-sided Bose-symmetric
extendible states]
\label{thm:two-sided-bose-extension}
For every $n,m,d_A,d_B\ge1$,
\begin{equation}
 \DHtr\Bigl(
  \BEXT_{n,m}(d_A,d_B),\,
  \SEP(d_A,d_B)
 \Bigr)
 \le
 \frac12
 \sqrt{\frac{\min\{d_A,d_B\}-1}{nm}}.
 \label{eq:two-sided-bext-distance}
\end{equation}
\end{theorem}

The proof follows the same pattern as the bosonic argument above:
rectangular argmax rounding bounds the integrality gap, and the same
estimate will also be used for higher marginals in
Section~\ref{sec:higher-marginals}.  Trace-distance duality then gives
the theorem, while a two-block symmetric purification removes the
Bose-symmetry assumption.

\subsubsection{Rectangular argmax rounding}
\label{sec:rectangular-argmax-main}

Let $\cH_A$ and $\cH_B$ be finite-dimensional Hilbert spaces, and write
$d_A:=\dim\cH_A$ and $d_B:=\dim\cH_B$.  For
$M\in\Herm(\cH_A\ot\cH_B)$, define
\begin{align}
 a_M(x,y)
 &:=\bra{x\ot y}M\ket{x\ot y},\\
 h_{\SEP}(M)
 &:=
 \max_{\substack{\norm{x}=1\\\norm{y}=1}}
 a_M(x,y).
 \label{eq:rectangular-hsep}
\end{align}
Since $\SEP(d_A,d_B)$ is the convex hull of pure product states,
\begin{equation}
 h_{\SEP}(M)
 =
 \max_{\sigma\in\SEP(d_A,d_B)}\Tr(M\sigma).
 \label{eq:rectangular-hsep-support}
\end{equation}
For $n,m\ge1$, its normalized rectangular lift is
\begin{equation}
 M^{[n,m]}
 :=
 \frac1{nm}
 \sum_{i=1}^n\sum_{j=1}^m M_{A_iB_j}
 \quad\text{on }\Sym^n(\cH_A)\ot\Sym^m(\cH_B).
 \label{eq:rectangular-lift}
\end{equation}
Here $M_{A_iB_j}$ acts as $M$ on $A_iB_j$ and as the identity on all
other registers.
If
$\omega_{n,m}\in
\D\bigl(\Sym^n(\cH_A)\ot\Sym^m(\cH_B)\bigr)$,
permutation symmetry within the $A$- and $B$-blocks gives
\begin{equation}
 \Tr\bigl(M^{[n,m]}\omega_{n,m}\bigr)
 =
 \Tr\bigl(M(\omega_{n,m})_{A_1B_1}\bigr).
 \label{eq:rectangular-lift-expectation}
\end{equation}

Given a unit top eigenvector of $M^{[n,m]}$, rectangular argmax rounding
selects a product vector from the corresponding overlap maximization.
The following estimate is the direct rectangular analogue of
Proposition~\ref{prop:tau-bound}.  Its proof is deferred to
Appendix~\ref{sec:rectangular-argmax-appendix}.

\begin{proposition}[Rectangular argmax estimate]
\label{prop:rectangular-argmax}
Let $M\in\Herm(\cH_A\ot\cH_B)$ and $n,m\ge1$, and let $\psi$ be a unit
top eigenvector of $M^{[n,m]}$.  Choose the argmax rounding object
$x\ot y$ by
\begin{equation}
 (x,y)
 \in
 \operatorname*{arg\,max}_{\substack{v\in\cH_A,\ w\in\cH_B\\
                                      \norm{v}=\norm{w}=1}}
 \abs{\braket{v^{\ot n}\ot w^{\ot m}}{\psi}}.
 \label{eq:rectangular-argmax-object}
\end{equation}
Then
\begin{equation}
 \lambda_{\max}\bigl(M^{[n,m]}\bigr)
 -a_M(x,y)
 \le
 \frac1{\sqrt{nm}}
 \norm{
  (P_{x^\perp}\ot P_{y^\perp})M(x\ot y)
 }_{\pi}.
 \label{eq:rectangular-argmax}
\end{equation}
\end{proposition}

As in the symmetric case, the pointwise estimate immediately gives an
integrality-gap bound.  Define
\begin{equation}
 \tau(M)
 :=
 \sup_{\substack{\norm{x}=1\\\norm{y}=1}}
 \norm{
  (P_{x^\perp}\ot P_{y^\perp})M(x\ot y)
 }_\pi.
 \label{eq:rectangular-tau-def}
\end{equation}

\begin{corollary}[Integrality gap from rectangular argmax rounding]
\label{cor:rectangular-integrality-gap}
For every $M\in\Herm(\cH_A\ot\cH_B)$ and $n,m\ge1$,
\begin{equation}
 \lambda_{\max}\bigl(M^{[n,m]}\bigr)-h_{\SEP}(M)
 \le
 \frac{\tau(M)}{\sqrt{nm}},
 \label{eq:rectangular-tau-gap}
\end{equation}
and hence
\begin{equation}
 \lambda_{\max}\bigl(M^{[n,m]}\bigr)-h_{\SEP}(M)
 \le
 \sqrt{\frac{\min\{\dim\cH_A,\dim\cH_B\}-1}{nm}}
 \norm{M}_\infty.
 \label{eq:rectangular-uniform}
\end{equation}
\end{corollary}

\begin{proof}
Apply Proposition~\ref{prop:rectangular-argmax} and use the definitions
of $h_{\SEP}(M)$ and $\tau(M)$.  For the second bound, the projected
vector
$(P_{x^\perp}\ot P_{y^\perp})M(x\ot y)$ has Schmidt rank at most
$\min\{\dim\cH_A,\dim\cH_B\}-1$.  Therefore,
\eqref{eq:tensor-norm-duality-rank} gives
\begin{align*}
 \norm{(P_{x^\perp}\ot P_{y^\perp})M(x\ot y)}_\pi
 &\le
 \sqrt{\min\{\dim\cH_A,\dim\cH_B\}-1}
 \norm{(P_{x^\perp}\ot P_{y^\perp})M(x\ot y)}\\
 &\le
 \sqrt{\min\{\dim\cH_A,\dim\cH_B\}-1}
 \norm{M}_\infty.
\end{align*}
\end{proof}

\subsubsection{Two-sided de Finetti bounds}

The rectangular gap immediately gives a de Finetti bound for arbitrary
extension levels on the two sides.

\begin{proof}[Proof of Theorem~\ref{thm:two-sided-bose-extension}]
Fix $\rho\in\BEXT_{n,m}(d_A,d_B)$ and choose an extension
$\widetilde\rho$ from \eqref{eq:two-sided-bext-definition}.
By duality (Lemma~\ref{lem:compact-convex-distance-duality}),
\begin{equation}
 \Dtr\bigl(\rho,\SEP(d_A,d_B)\bigr)
 =
 \frac12
 \max_{\substack{M\in\Herm(\cH_A\ot\cH_B)\\
                   \norm{M}_\infty\le1}}
 \left\{\Tr(M\rho)-h_{\SEP}(M)\right\}.
 \label{eq:two-sided-sep-dual}
\end{equation}
For every Hermitian $M$, the rectangular-lift identity
\eqref{eq:rectangular-lift-expectation} and
Corollary~\ref{cor:rectangular-integrality-gap} give
\begin{align}
 \Tr(M\rho)-h_{\SEP}(M)
 &=
 \Tr\bigl(M^{[n,m]}\widetilde\rho\bigr)-h_{\SEP}(M)
 \notag\\
 &\le
 \lambda_{\max}\bigl(M^{[n,m]}\bigr)-h_{\SEP}(M)
 \notag\\
 &\le
 \sqrt{\frac{\min\{d_A,d_B\}-1}{nm}}
 \norm{M}_\infty.
 \label{eq:two-sided-bext-rectangular-gap}
\end{align}
Inserting this estimate into \eqref{eq:two-sided-sep-dual} and taking
the supremum over $\rho$ proves \eqref{eq:two-sided-bext-distance}.
\end{proof}

We now remove the Bose-symmetry assumption by symmetric purification.

\begin{corollary}[De Finetti upper bound for two-sided extendible states]
\label{cor:two-sided-extension}
For every $n,m,d_A,d_B\ge1$,
\begin{equation}
 \DHtr\Bigl(
  \EXT_{n,m}(d_A,d_B),\,
  \SEP(d_A,d_B)
 \Bigr)
 \le
 \frac12
 \sqrt{\frac{\min\{d_A^2,d_B^2\}-1}{nm}}.
 \label{eq:two-sided-ext-distance}
\end{equation}
\end{corollary}

\begin{proof}[Proof of Corollary~\ref{cor:two-sided-extension}]
Let $\widetilde\rho$ be an $(n,m)$-extension of
$\rho\in\EXT_{n,m}(d_A,d_B)$.  Choose copies
$\cH_{E_A}\cong\cH_A$ and $\cH_{E_B}\cong\cH_B$, with copied
orthonormal bases, and let $\ket\Omega$ be the unnormalized maximally
entangled vector between
$\cH_{E_A}^{\ot n}\ot\cH_{E_B}^{\ot m}$ and
$\cH_A^{\ot n}\ot\cH_B^{\ot m}$.  The canonical purification
\[
 \ket\Psi
 :=
 \bigl(I\ot\sqrt{\widetilde\rho}\bigr)\ket\Omega
\]
satisfies
$\Tr_{E_A^{\ot n}E_B^{\ot m}}\ketbra\Psi=\widetilde\rho$.

Regroup the registers into the composite sites
\[
 \cX:=\cH_{E_A}\ot\cH_A,
 \qquad
 \cY:=\cH_{E_B}\ot\cH_B,
\]
where the composite registers $X_i=E_{A,i}A_i$ and
$Y_j=E_{B,j}B_j$ carry $\cX$ and $\cY$, respectively.
Fix $\pi\in\mathfrak S_n$ and regard the $B$- and $E_B$-registers as
spectators.  Invariance of $\widetilde\rho$ under the $A$-permutation
$U_\pi^A$ implies that $U_\pi^A$ commutes with
$\sqrt{\widetilde\rho}$.  The same permutation on the copied
$E_A$-registers and on the original $A$-registers fixes $\ket\Omega$,
so the canonical-purification identity gives
\[
 (U_\pi^{\cX}\ot I_{\cY^{\ot m}})\ket\Psi=\ket\Psi,
\]
where $U_\pi^{\cX}$ permutes the $n$ composite $X$-sites.  The same argument
with the two blocks exchanged gives
\[
 (I_{\cX^{\ot n}}\ot U_\tau^{\cY})\ket\Psi=\ket\Psi
 \qquad(\tau\in\mathfrak S_m).
\]
Thus the symmetrizing projectors for the $X$- and $Y$-sites both fix
$\ket\Psi$, and hence
\[
 \ket\Psi
 \in
 \Sym^n(\cX)
 \ot
 \Sym^m(\cY).
\]

The $X_1Y_1$ marginal therefore belongs to
$\BEXT_{n,m}(d_A^2,d_B^2)$ and reduces to $\rho$ after tracing out
$E_{A,1}E_{B,1}$.
Applying Theorem~\ref{thm:two-sided-bose-extension} in local dimensions
$d_A^2$ and $d_B^2$, then using contractivity of trace distance and
preservation of separability under partial trace, gives
\eqref{eq:two-sided-ext-distance}.
\end{proof}

\section{\texorpdfstring{$t$}{t}-site de Finetti upper bounds}
\label{sec:higher-marginals}

We now extend the two-site de Finetti bounds to $t$-site marginals.

\begin{theorem}[$t$-site de Finetti upper bounds]
\label{thm:higher-main}
For every $d\ge2$ and $3\le t\le N$,
\begin{align}
 \DHtr\Bigl(
 \BOS_N^{(t)}(d),\,
 \SEPH^{(t)}(d)
 \Bigr)
 &\le
 \frac{4t\sqrt{d-1}}{N},
 \label{eq:higher-bos-main}\\
 \DHtr\Bigl(
 \EXCH_N^{(t)}(d),\,
 \SEPP^{(t)}(d)
 \Bigr)
 &\le
 \frac{4t\sqrt{d^2-1}}{N}.
 \label{eq:higher-ex-main}
\end{align}
\end{theorem}

We first prove the bosonic bound; symmetric purification then gives the
exchangeable bound.
For $M\in\Herm(\Sym^t\cH)$, define
\[
 h_t(M)
 :=
 \max_{\norm{u}=1}
 \bra{u^{\ot t}}M\ket{u^{\ot t}}.
\]
This is the maximum expectation of $M$ over $\SEPH^{(t)}(d)$.

Trace-distance duality reduces the problem to bounding
$\Tr(M\rho_N^{(t)})-h_t(M)$.  We distinguish one tensor factor from the
other $t-1$ factors and view $\rho_N^{(t)}$ as a bipartite state on
\[
 \cH\ot\Sym^{t-1}(\cH).
\]
Grouping the $N$ sites into $n$ individual sites and $m$ blocks of
$t-1$ sites makes this state Bose-symmetric $(n,m)$-extendible.  This
two-sided extendibility lets rectangular argmax rounding produce a
product vector $x\ot y$, enforcing unentanglement across the
$1:(t-1)$ cut.  The challenge is to propagate this
unentanglement to the remaining cuts: $y$ need not equal
$x^{\ot(t-1)}$, so $x\ot y$ need not be a Hartree vector.

The symmetry test is what propagates this unentanglement.  We
incorporate it into the objective before performing argmax rounding.
The bosonic marginal passes the test with certainty, whereas the
rounded product loses value according to its symmetry-test rejection
probability.  The same rejection probability controls the cost of
replacing $x\ot y$ by $x^{\ot t}$.
Choosing the weight of the symmetry test to balance this cost against
the rectangular argmax rounding error gives the claimed rate.  The
rest of the section makes these steps precise.

We begin with the dual formulation of the bosonic bound.

\begin{lemma}[Exact trace-distance duality]
\label{lem:trace-distance-duality-higher}
For every $d\ge2$, $t>2$, and every $\gamma\in\D(\Sym^t\cH)$,
\begin{equation}
 \Dtr\bigl(\gamma,\SEPH^{(t)}(d)\bigr)
 =
 \sup_{\substack{M\in\Herm(\Sym^t\cH)\\
                  0\preceq M\preceq I}}
 \left\{
  \Tr(M\gamma)-h_t(M)
 \right\}.
 \label{eq:trace-distance-duality-higher}
\end{equation}
\end{lemma}

\begin{proof}
The support function of $\SEPH^{(t)}(d)$ is $h_t(M)$.  By duality
(Lemma~\ref{lem:compact-convex-distance-duality}), this gives
\eqref{eq:trace-distance-duality-higher}.  Since all states involved
are supported on $\Sym^t\cH$, the witness may be restricted to that
subspace.
\end{proof}

\subsection{Two-sided extendibility across the
\texorpdfstring{$1:(t-1)$}{1:(t-1)} cut}

Every Hartree vector factors across this cut as
\[
 u^{\ot t}=u\ot u^{\ot(t-1)}.
\]
The first step drops the requirement $y=x^{\ot(t-1)}$ and allows
arbitrary products $x\ot y$.  We choose the $1:(t-1)$ cut because its
smaller side has dimension $d$, which is the favorable dimension in the
rectangular argmax estimate, and because Bose symmetry supplies exactly
the corresponding two-sided extension.

Fix $t>2$ and set
\[
 \cH_A:=\cH,
 \qquad
 \cH_B:=\Sym^{t-1}(\cH).
\]
We use the natural inclusion
$\Sym^t(\cH)\subseteq\cH_A\ot\cH_B$ obtained by distinguishing the first
tensor factor.

\begin{fact}[Two-sided extendibility of the marginal]
\label{fact:one-tminusone-extension}
Let $n,m\ge1$ satisfy $n+(t-1)m\le N$.  For every
$\rho_N\in\D(\Sym^N\cH)$, the marginal $\rho_N^{(t)}$, viewed as a
state on $\cH_A\ot\cH_B$, is Bose-symmetric $(n,m)$-extendible.
\end{fact}

\begin{proof}
Trace out $N-n-(t-1)m$ sites and rebracket the remaining tensor factors
as $n$ individual sites and $m$ blocks of $t-1$ sites:
\[
 \underset{A_1}{\bullet}
 \ \cdots\
 \underset{A_n}{\bullet}
 \qquad\qquad
 \overset{t-1\text{ sites}}{\underset{B_1}{[\,\bullet\ \cdots\ \bullet\,]}}
 \ \cdots\
 \overset{t-1\text{ sites}}{\underset{B_m}{[\,\bullet\ \cdots\ \bullet\,]}}.
\]
Let $\widetilde\rho_{n,m}$ be the resulting marginal.  Because
$\rho_N$ is supported on $\Sym^N(\cH)$, this marginal is supported on
$\Sym^{n+(t-1)m}(\cH)$.  It is therefore supported on $\cH_B$ within each
$B_j$, and every permutation of the $n$ $A$-sites or of the $m$
$B$-blocks acts as the identity on its support.  Consequently,
$\widetilde\rho_{n,m}\in
\D\bigl(\Sym^n(\cH_A)\ot\Sym^m(\cH_B)\bigr)$.
Moreover, every $A_iB_j$ consists of $t$ distinct original sites, so
$(\widetilde\rho_{n,m})_{A_iB_j}=\rho_N^{(t)}$ for all $i,j$.
\end{proof}

\subsection{From product vectors to tensor powers}

Rectangular argmax rounding produces a product vector $x\ot y$ across
the chosen $1:(t-1)$ cut, but it need not satisfy
$y=x^{\ot(t-1)}$.  The symmetry test propagates this unentanglement to
the remaining cuts.

Let $\Pi_t$ be the orthogonal projector from
$\cH_A\ot\cH_B=\cH\ot\Sym^{t-1}(\cH)$ onto $\Sym^t(\cH)$.  It tests whether
the product vector lies in the symmetric subspace.  The following lemma
quantifies the test: a high probability of passing forces $y$ to be
close to $x^{\ot(t-1)}$.  It also bounds the component of
$\Pi_t(x\ot y)$ and $M(x\ot y)$ orthogonal to both rounded factors.

\begin{lemma}[Approximate symmetry of a product vector]
\label{lem:symmetry-projector}
Let $x\in\cH$ and $y\in\Sym^{t-1}(\cH)$ be unit vectors, and define
the symmetry-test rejection probability
\[
 \delta_{\mathrm{sym}}(x,y)
 :=
 1-\bra{x\ot y}\Pi_t\ket{x\ot y}.
\]
Set $P_\perp:=P_{x^\perp}\ot P_{y^\perp}$.
Then
\begin{equation}
 1-
 \abs{\braket{x^{\ot(t-1)}}{y}}^2
 \le
 t\,\delta_{\mathrm{sym}}(x,y).
 \label{eq:symmetry-overlap}
\end{equation}
If $0\preceq M\preceq I$ acts on $\Sym^t\cH$ and is extended by zero
to $\cH_A\ot\cH_B$, then
\begin{equation}
 \bra{x\ot y}M\ket{x\ot y}
 \le
 h_t(M)+\sqrt{t\,\delta_{\mathrm{sym}}(x,y)}.
 \label{eq:symmetry-objective}
\end{equation}
Moreover,
\begin{align}
 \norm{P_\perp\Pi_t(x\ot y)}_\pi
 &\le
 \sqrt{d-1}\,\sqrt{\delta_{\mathrm{sym}}(x,y)},
 \label{eq:symmetry-projector-component}\\
 \norm{P_\perp M(x\ot y)}_\pi
 &\le
 \sqrt{d-1}.
 \label{eq:symmetry-M-component}
\end{align}
\end{lemma}

\begin{proof}
Because $y$ is symmetric in its $t-1$ factors, averaging all
permutations reduces to averaging over the $t$ possible locations of
the distinguished first factor:
\[
 \Pi_t
 =
 \frac1t\sum_{j=1}^t U_{1j},
 \qquad U_{11}=I.
\]
Let $\gamma_y$ be the one-site marginal of $\ketbra y$.  Since $y$ is symmetric,
\[
 \bra{x\ot y}U_{1j}\ket{x\ot y}
 =
 \bra x\gamma_y\ket x
 \qquad(2\le j\le t).
\]
It follows that
\begin{equation}
 \bra{x\ot y}\Pi_t\ket{x\ot y}
 =
 \frac{1+(t-1)\bra x\gamma_y\ket x}{t}.
 \label{eq:symmetry-expectation}
\end{equation}
Put $Q_x=I-\ketbra x$.  On $\Sym^{t-1}(\cH)$, the operator
\[
 \mathcal N_x
 :=
 \sum_{j=1}^{t-1}(Q_x)_j
\]
counts how many of the $t-1$ factors lie in $x^\perp$.  Its kernel is
$\C x^{\ot(t-1)}$, and all other eigenvalues are at least one.  Therefore
\[
 I-\ketbra{x^{\ot(t-1)}}
 \preceq
 \mathcal N_x.
\]
Taking the expectation in $y$ and using \eqref{eq:symmetry-expectation},
\begin{align*}
 1-\abs{\braket{x^{\ot(t-1)}}{y}}^2
 &\le
 (t-1)\bigl(1-\bra x\gamma_y\ket x\bigr)\\
 &=
 t\,\delta_{\mathrm{sym}}(x,y),
\end{align*}
which proves \eqref{eq:symmetry-overlap}.

The trace distance between the pure states $\ketbra{x\ot y}$ and $\ketbra{x}^{\ot t}$ is
\[
 \sqrt{
  1-\abs{\braket{x^{\ot(t-1)}}{y}}^2
 }
 \le
 \sqrt{t\,\delta_{\mathrm{sym}}(x,y)}.
\]
Since $0\preceq M\preceq I$ and
$\bra{x^{\ot t}}M\ket{x^{\ot t}}\le h_t(M)$, trace-distance duality
gives \eqref{eq:symmetry-objective}.

For the bounds on the projected components, first note that orthogonality of $\Pi_t$
gives
\[
 \norm{(I-\Pi_t)(x\ot y)}^2
 =
 \delta_{\mathrm{sym}}(x,y).
\]
Since $P_\perp(x\ot y)=0$,
\[
 P_\perp\Pi_t(x\ot y)
 =
 -P_\perp(I-\Pi_t)(x\ot y).
\]
Thus its Hilbert-space norm is at most
$\sqrt{\delta_{\mathrm{sym}}(x,y)}$.  Every vector in
$x^\perp\ot y^\perp$ has Schmidt rank at most $d-1$, so
\eqref{eq:tensor-norm-duality-rank} proves
\eqref{eq:symmetry-projector-component}.  The same rank bound and
\[
 \norm{P_\perp M(x\ot y)}
 \le
 \norm{M}_\infty
 \le1
\]
prove \eqref{eq:symmetry-M-component}.
\end{proof}

A related, quantitatively weaker estimate appears in the
one-versus-many slicing de Finetti theorem of Jeronimo and
Wu~\cite[Theorem~23]{JW24}.  After normalizing $\Pi_t(x\ot y)$, that
theorem gives squared overlap
$1-O(t^3\delta_{\mathrm{sym}}(x,y))$ with some tensor power
$v^{\ot t}$.  Estimate~\eqref{eq:symmetry-overlap} improves the cubic
loss to a linear one and identifies the direction as $x$.  This sharper
dependence is essential for the linear dependence on $t$ in
Theorem~\ref{thm:higher-main}.

\subsection{Bounding the integrality gap}

We now combine rectangular argmax rounding with the symmetry estimate.
For $M_\theta=(1-\theta)\Pi_t+\theta M$, the grouped bosonic marginal is
feasible for the rectangular lift and retains the score
$1-\theta+\theta\Tr(M\rho_N^{(t)})$.  Rectangular argmax rounding
produces $x\ot y$, and Lemma~\ref{lem:symmetry-projector} controls its
replacement by $x^{\ot t}$.  Optimizing over $\theta$ balances the two
estimates and gives the integrality-gap bound required by
Lemma~\ref{lem:trace-distance-duality-higher}.

\begin{proposition}[Rounding across the $1:(t-1)$ cut]
\label{prop:one-tminusone-rounding}
Let $d\ge2$, $3\le t\le N$, and let $n,m\ge2$ satisfy
$n+(t-1)m\le N$.  For every
$\rho_N\in\D(\Sym^N\cH)$ and every
$0\preceq M\preceq I_{\Sym^t\cH}$,
\[
 \Tr\bigl(M\rho_N^{(t)}\bigr)-h_t(M)
 \le
 (\sqrt t+1)\sqrt{\frac{d-1}{nm}}.
\]
\end{proposition}

\begin{proof}
Let $\widetilde\rho_{n,m}$ be the grouped state from
Fact~\ref{fact:one-tminusone-extension}.
For $0<\theta<1$, extend $M$ by zero outside $\Sym^t\cH$ and consider
the mixed test
\[
 M_\theta
 :=
 (1-\theta)\Pi_t+\theta M
 \quad\text{on }\cH_A\ot\cH_B.
\]
Because $\rho_N^{(t)}$ is supported on $\Sym^t\cH$, the projector
$\Pi_t$ acts as the identity on its support.  Therefore
\begin{align*}
 \Tr\bigl(M_\theta\rho_N^{(t)}\bigr)
 &=
 (1-\theta)\Tr\bigl(\Pi_t\rho_N^{(t)}\bigr)
 +\theta\Tr\bigl(M\rho_N^{(t)}\bigr)\\
 &=
 1-\theta+\theta\Tr\bigl(M\rho_N^{(t)}\bigr).
\end{align*}
By Fact~\ref{fact:one-tminusone-extension},
$(\widetilde\rho_{n,m})_{A_1B_1}=\rho_N^{(t)}$.  Hence the
rectangular-lift identity \eqref{eq:rectangular-lift-expectation} gives
\begin{align}
 1-\theta+\theta\Tr\bigl(M\rho_N^{(t)}\bigr)
 &=
 \Tr\bigl(M_\theta\rho_N^{(t)}\bigr)
 =
 \Tr\bigl(M_\theta^{[n,m]}\widetilde\rho_{n,m}\bigr)\notag\\
 &\le
 \lambda_{\max}\bigl(M_\theta^{[n,m]}\bigr).
 \label{eq:penalized-lower}
\end{align}

\par\medskip\noindent\emph{Rectangular argmax rounding.}\quad
Proposition~\ref{prop:rectangular-argmax}, applied to $M_\theta$,
gives the first-stage product vector $x\ot y$.  Set
$P_\perp:=P_{x^\perp}\ot P_{y^\perp}$.  Moving its product-value term
to the right-hand side gives
\begin{equation}
 \lambda_{\max}\bigl(M_\theta^{[n,m]}\bigr)
 \le
 \bra{x\ot y}M_\theta\ket{x\ot y}
 +\frac1{\sqrt{nm}}\norm{P_\perp M_\theta(x\ot y)}_\pi.
 \label{eq:penalized-argmax-bound}
\end{equation}
We use the symmetry test to round this product vector to $x^{\ot t}$
and estimate the two terms on the right using the same rejection
probability.  Write
$\delta=\delta_{\mathrm{sym}}(x,y)$ for this probability and
\[
 \beta
 :=
 \sqrt{\frac{d-1}{nm}}.
\]
The product value of $M_\theta$ is
\begin{align}
 \bra{x\ot y}M_\theta\ket{x\ot y}
 &=
 (1-\theta)(1-\delta)
 +\theta\bra{x\ot y}M\ket{x\ot y}
 \notag\\
 &\le
 (1-\theta)(1-\delta)
 +\theta\bigl(h_t(M)+\sqrt{t\delta}\bigr)
 \notag\\
 &=
 1-\theta+\theta h_t(M)
 -(1-\theta)\delta
 +\theta\sqrt{t\delta},
 \label{eq:penalized-product-bound}
\end{align}
where the inequality is \eqref{eq:symmetry-objective}.  The triangle
inequality and the two bounds on projected components in
Lemma~\ref{lem:symmetry-projector} give
\begin{align}
 \frac1{\sqrt{nm}}
 \norm{P_\perp M_\theta(x\ot y)}_\pi
 &\le
 \frac1{\sqrt{nm}}
 \left(
  (1-\theta)\norm{P_\perp\Pi_t(x\ot y)}_\pi
  +\theta\norm{P_\perp M(x\ot y)}_\pi
 \right)
 \notag\\
 &\le
 \sqrt{\frac{d-1}{nm}}
 \bigl((1-\theta)\sqrt\delta+\theta\bigr)
 \notag\\
 &=
 \beta\bigl((1-\theta)\sqrt\delta+\theta\bigr).
 \label{eq:penalized-component-bound}
\end{align}
Both estimates concern the pair selected in
\eqref{eq:penalized-argmax-bound}, and hence the same quantity $\delta$.
Combining \eqref{eq:penalized-lower},
\eqref{eq:penalized-argmax-bound},
\eqref{eq:penalized-product-bound}, and
\eqref{eq:penalized-component-bound}, we obtain
\begin{align*}
 1-\theta+\theta\Tr\bigl(M\rho_N^{(t)}\bigr)
 &\le
 1-\theta+\theta h_t(M)
 -(1-\theta)\delta
 +\theta\sqrt{t\delta}
 +\beta(1-\theta)\sqrt\delta
 +\beta\theta.
\end{align*}
Subtracting $1-\theta+\theta h_t(M)$ from both sides gives
\begin{equation}
 \theta\left(\Tr\bigl(M\rho_N^{(t)}\bigr)-h_t(M)\right)
 \le
 -(1-\theta)\delta
 +\bigl(\theta\sqrt t+\beta(1-\theta)\bigr)\sqrt\delta
 +\beta\theta.
 \label{eq:higher-before-opt}
\end{equation}
Keeping the negative symmetry-test term and the two error terms coupled
through the same $\delta$ is what allows the optimization to retain
linear dependence on $t$.

\par\medskip\noindent\emph{Optimizing the error bound.}\quad
Since $0\le\delta\le1$, the variable $s:=\sqrt\delta$ lies in
$[0,1]$.  In terms of $s$, \eqref{eq:higher-before-opt} reads
\begin{equation*}
 \theta\left(\Tr\bigl(M\rho_N^{(t)}\bigr)-h_t(M)\right)
 \le
 -(1-\theta)s^2
 +\bigl(\theta\sqrt t+\beta(1-\theta)\bigr)s
 +\beta\theta.
\end{equation*}
For every real $s$, completing the square gives
\begin{equation*}
 -(1-\theta)s^2
 +\bigl(\theta\sqrt t+\beta(1-\theta)\bigr)s
 \le
 \frac{\bigl(\theta\sqrt t+\beta(1-\theta)\bigr)^2}
      {4(1-\theta)}.
\end{equation*}
After dividing by $\theta>0$ and expanding the square, we obtain
\begin{align}
 \Tr\bigl(M\rho_N^{(t)}\bigr)-h_t(M)
 &\le
 \frac{\bigl(\theta\sqrt t+\beta(1-\theta)\bigr)^2}
      {4\theta(1-\theta)}
 +\beta
 \notag\\
 &=
 \frac{\theta t}{4(1-\theta)}
 +\frac{\beta\sqrt t}{2}
 +\frac{\beta^2(1-\theta)}{4\theta}
 +\beta.
 \label{eq:higher-theta-bound}
\end{align}
Set
\[
 \lambda:=\frac{\theta}{1-\theta}.
\]
As $\theta$ ranges over $(0,1)$, the parameter $\lambda$ ranges over
$(0,\infty)$.  Since $\beta>0$, we may choose
\[
 \lambda=\frac{\beta}{\sqrt t},
 \qquad\text{equivalently}\qquad
 \theta=\frac{\beta}{\beta+\sqrt t}.
\]
With this choice, \eqref{eq:higher-theta-bound} becomes
\begin{align*}
 \Tr\bigl(M\rho_N^{(t)}\bigr)-h_t(M)
 &\le
 \frac{t\lambda}{4}
 +\frac{\beta\sqrt t}{2}
 +\frac{\beta^2}{4\lambda}
 +\beta\\
 &=
 \beta(\sqrt t+1).
\end{align*}
Recalling the definition of $\beta$, we conclude that
\[
 \Tr\bigl(M\rho_N^{(t)}\bigr)-h_t(M)
 \le
 (\sqrt t+1)\sqrt{\frac{d-1}{nm}}.
\]
This proves the proposition.
\end{proof}

\subsection{Finishing the proof}

It remains to choose the two extension levels so that $nm$ is large.

\begin{proof}[Proof of Theorem~\ref{thm:higher-main}]
Both trace distances are at most one.  If $N<4(t-1)$, then both
right-hand sides in \eqref{eq:higher-bos-main} and
\eqref{eq:higher-ex-main} exceed one, so the result is immediate.  We
may therefore suppose that $N\ge4(t-1)$.

For the bosonic bound, fix $\rho_N\in\D(\Sym^N\cH)$ and an operator
$M$ on $\Sym^t\cH$ satisfying $0\preceq M\preceq I$.  Define
\[
 m:=\left\lfloor\frac{N}{2(t-1)}\right\rfloor,
 \qquad
 n:=N-(t-1)m.
\]
Then $n,m\ge2$ and
\begin{equation}
 n\ge\frac N2,
 \qquad
 m\ge\frac{N}{4(t-1)},
 \qquad
 nm\ge\frac{N^2}{8(t-1)}.
 \label{eq:one-tminusone-product}
\end{equation}
Proposition~\ref{prop:one-tminusone-rounding}
and \eqref{eq:one-tminusone-product} give
\begin{align*}
 \Tr\bigl(M\rho_N^{(t)}\bigr)-h_t(M)
 &\le
 (\sqrt t+1)\sqrt{\frac{d-1}{nm}}\\
 &\le
 2\sqrt2\,
 \frac{(\sqrt t+1)\sqrt{(t-1)(d-1)}}{N}\\
 &<
 \frac{4t\sqrt{d-1}}{N}.
\end{align*}
Here the last inequality follows from
$(\sqrt t+1)^2\le2(t+1)$, which gives
$\sqrt{t-1}(\sqrt t+1)<\sqrt2\,t$.
Taking the supremum over the witnesses $0\preceq M\preceq I$ and using
Lemma~\ref{lem:trace-distance-duality-higher} gives the same bound for
$\Dtr(\rho_N^{(t)},\SEPH^{(t)}(d))$.  Taking the supremum over the
bosonic states $\rho_N$ therefore proves
\[
 \DHtr\Bigl(
 \BOS_N^{(t)}(d),\,
 \SEPH^{(t)}(d)
 \Bigr)
 \le
 \frac{4t\sqrt{d-1}}{N}.
\]

For the exchangeable bound, let $\sigma_N$ be exchangeable on
$(\C^d)^{\ot N}$.  By Lemma~\ref{lem:symmetric-purification}, it has a
symmetric purification
$\widetilde\rho_N\in\D(\Sym^N(\C^d\ot\C^d))$.  Applying the bosonic
bound just proved, now in one-site dimension $d^2$, gives a probability
measure $\mu$ on unit vectors $\phi\in\C^d\ot\C^d$ such that
\[
 \Dtr\left(
  \widetilde\rho_N^{(t)},
 \int\ketbra\phi^{\ot t}\,d\mu(\phi)
 \right)
 \le
 \frac{4t\sqrt{d^2-1}}{N}.
\]
For $\tau_\phi:=\Tr_E\ketbra\phi$, sitewise partial trace maps
$\ketbra\phi^{\ot t}$ to $\tau_\phi^{\ot t}$ and
$\widetilde\rho_N^{(t)}$ to $\sigma_N^{(t)}$.  Hence contractivity of
trace distance gives
\[
 \Dtr\left(
  \sigma_N^{(t)},
 \int\tau_\phi^{\ot t}\,d\mu(\phi)
 \right)
 \le
 \frac{4t\sqrt{d^2-1}}{N}.
\]
The integral belongs to $\SEPP^{(t)}(d)$.  Taking the supremum over
$\sigma_N$ proves \eqref{eq:higher-ex-main} and completes the proof.
\end{proof}

\section{Applications of the optimal de Finetti theorem}
\label{sec:optimal-definetti-applications}
\label{sec:trace-algorithms}

The optimal trace-norm theorem has three applications whose resource
bounds are subexponential in the local dimension at fixed accuracy.  The
disentangler construction uses the statement that every two-sided
Bose-symmetric extendible state is close to separable.  The same
two-sided extension relaxation approximates Best Separable State, and
its error guarantee also gives an algorithm for trace-norm separability
testing.

\subsection{Subexponential disentanglers}
\label{sec:disentanglers}

Recall that a channel
\[
 \Lambda:\D(\C^D)\longrightarrow\D(\C^d\ot\C^d)
\]
is an $(\epsilon,\delta)$-\emph{disentangler} if~\cite{ABDFS08}
\begin{align}
 \sup_{\rho\in\D(\C^D)}
 \Dtr\bigl(\Lambda(\rho),\SEP(d,d)\bigr)
 &\le\epsilon,
 \label{eq:disentangler-soundness}\\
 \sup_{\sigma\in\SEP(d,d)}
 \inf_{\rho\in\D(\C^D)}\Dtr\bigl(\Lambda(\rho),\sigma\bigr)
 &\le\delta.
 \label{eq:disentangler-coverage}
\end{align}
The first condition is soundness: every output is close to separable.
The second is coverage: every separable state is close to some output.
A sufficiently small and efficiently implementable disentangler would
simulate two unentangled quantum proofs using one ordinary QMA witness.
The original Watrous disentangler conjecture predicts
$\log D=\Omega(d)$ for fixed $\epsilon,\delta$ satisfying
$\epsilon+\delta<1$~\cite{ABDFS08}. The construction below therefore disproves the original Watrous disentangler
conjecture.

\begin{theorem}[Subexponential disentanglers]
\label{thm:disentangler-main}
\label{thm:disentangler-restated}
Let $d\ge2$ and $n\ge1$, and set $\cH_A=\cH_B=\C^d$.  There is a
channel
\[
 \Lambda_{d,n}:
 \D\bigl(\cH_A^{\ot n}\ot\cH_B^{\ot n}\bigr)
 \longrightarrow
 \D(\cH_A\ot\cH_B)
\]
which is an $(\epsilon_{d,n},0)$-disentangler, where
\begin{equation}
 \epsilon_{d,n}
 :=
 \min\left\{1,\frac{\sqrt{d-1}}{2n}\right\}.
 \label{eq:disentangler-quantitative-error}
\end{equation}
Consequently, for every fixed $\epsilon\in(0,1)$, there is an
$(\epsilon,0)$-disentangler with input dimension $D$ satisfying
$\log D=O_\epsilon(\sqrt d\log d)$.  In particular, one may take
$(\epsilon,\delta)=(1/4,0)$.
\end{theorem}

\begin{proof}
Let $\Pi_n^A$ and $\Pi_n^B$ be the orthogonal projections onto
$\Sym^n(\cH_A)$ and $\Sym^n(\cH_B)$, respectively, and set
\[
 P_n:=\Pi_n^A\ot\Pi_n^B.
\]
Choose unit vectors $a_0\in\cH_A$ and $b_0\in\cH_B$, and let
\[
 \omega_n:=\ketbra{a_0}^{\ot n}\ot\ketbra{b_0}^{\ot n}.
\]
The trace-preserving projection onto
$\Sym^n(\cH_A)\ot\Sym^n(\cH_B)$ is the channel
\begin{equation}
 \mathcal P_n(\rho)
 :=
 P_n\rho P_n
 +
 \Tr\bigl((I-P_n)\rho\bigr)\omega_n.
 \label{eq:symmetric-projection-channel}
\end{equation}
Indeed, the first term retains the symmetric component, while the
second replaces the rejected component by the fixed pure product state
$\omega_n$ in the same subspace.  Define the disentangler by retaining
the $A_1B_1$ registers:
\begin{equation}
 \Lambda_{d,n}(\rho)
 :=
 \Tr_{A_{2:n}B_{2:n}}\mathcal P_n(\rho)
 =
 \bigl(\mathcal P_n(\rho)\bigr)_{A_1B_1}.
 \label{eq:partial-trace-disentangler}
\end{equation}

The state $\mathcal P_n(\rho)$ is supported on
$\Sym^n(\cH_A)\ot\Sym^n(\cH_B)$.  Hence every output of
$\Lambda_{d,n}$ belongs to
$\BEXT_{n,n}(d,d)$, and
Theorem~\ref{thm:two-sided-bose-extension} gives
\begin{equation}
 \Dtr\bigl(\Lambda_{d,n}(\rho),\SEP(d,d)\bigr)
 \le
 \epsilon_{d,n}.
 \label{eq:partial-trace-disentangler-soundness}
\end{equation}
Coverage is exact: if
$\sigma=\sum_jp_j\ketbra{a_j}\ot\ketbra{b_j}$ is separable, then
\[
 \widetilde\sigma_n
 :=\sum_jp_j\ketbra{a_j}^{\ot n}\ot\ketbra{b_j}^{\ot n}
\]
is supported on the range of $P_n$.  Therefore
$\mathcal P_n(\widetilde\sigma_n)=\widetilde\sigma_n$ and
$\Lambda_{d,n}(\widetilde\sigma_n)=\sigma$.

For fixed $\epsilon\in(0,1)$, choose
\[
 n:=\left\lceil\frac{\sqrt{d-1}}{2\epsilon}\right\rceil.
\]
The input space has dimension $D_{d,n}=d^{2n}$, so
\[
 \log D_{d,n}
 =2n\log d
 =\Theta_{\epsilon}(\sqrt d\log d).
\]
This proves Theorem~\ref{thm:disentangler-main}.  Taking
$n=\lceil2\sqrt{d-1}\rceil$ gives the stated $(1/4,0)$ parameters.
\end{proof}

\subsection{Best Separable State without perfect completeness}
\label{sec:qma2-applications}

Let $\cH_A\cong\C^{d_A}$ and $\cH_B\cong\C^{d_B}$.  Best Separable
State (BSS) asks for the largest expectation of an operator
$0\preceq M\preceq I$ on a separable state.  When $M$ is the accepting
operator of a $\mathrm{QMA}(2)$ verifier, this is its largest
acceptance probability over two unentangled proofs.  For
$M\in\Herm(\cH_A\ot\cH_B)$, write
\begin{equation}
 h_{\SEP}(M)
 :=
 \max_{\sigma\in\SEP(d_A,d_B)}\Tr(M\sigma)
 =
 \max_{\substack{\norm{x}=1\\\norm{y}=1}}
 \bra{x\ot y}M\ket{x\ot y}.
 \label{eq:hsep-application}
\end{equation}
For an accepting operator $0\preceq M\preceq I$ and $n\ge1$, set
\begin{equation}
\begin{aligned}
 M^{[n,n]}
 &:=
 \frac1{n^2}
 \sum_{i=1}^n\sum_{j=1}^n M_{A_iB_j},
 \\
 h_{\BEXT_{n,n}}(M)
 &:=
 \max_{\rho\in\BEXT_{n,n}(d_A,d_B)}
 \Tr(M\rho)
 =
 \lambda_{\max}\bigl(M^{[n,n]}\bigr),
\end{aligned}
 \label{eq:two-sided-bss-relaxation}
\end{equation}
on $\Sym^n(\cH_A)\ot\Sym^n(\cH_B)$, with each summand acting as the
identity on all unlisted tensor factors.  The last equality follows from
permutation symmetry within each block.

\begin{theorem}[Subexponential BSS]
\label{thm:bss-main}
\label{thm:bss-two-sided-algorithm}
Let $0\preceq M\preceq I$ act on $\cH_A\ot\cH_B$.  For every $n\ge1$,
\begin{equation}
 h_{\SEP}(M)
 \le h_{\BEXT_{n,n}}(M)
 \le h_{\SEP}(M)
 +\frac{\sqrt{\min\{d_A,d_B\}-1}}{2n}.
 \label{eq:bss-integrality-gap}
\end{equation}
Consequently, for an explicitly given accepting operator $M$ on
$\C^d\ot\C^d$,
$h_{\SEP}(M)$ can be approximated to additive error
$\epsilon\in(0,1]$ deterministically.  Up to a factor polynomial in
the bit complexity of $M$, the running time is
\begin{equation}
 \exp\left(
 O\left(
 \frac{\sqrt d}{\epsilon}\log(2+\epsilon\sqrt d)
 \right)
 \right).
 \label{eq:bss-runtime-epsilon}
\end{equation}
The promise $h_{\SEP}(M)\ge c$ versus $h_{\SEP}(M)\le s$ is decidable within the same bound with $\epsilon=\Theta(c-s)$.
\end{theorem}

\begin{proof}
Write
\[
 \delta_n:=
 \frac{\sqrt{\min\{d_A,d_B\}-1}}{2n}.
\]
The inclusion \eqref{eq:sep-in-two-sided-bext} gives the lower bound in
\eqref{eq:bss-integrality-gap}.  For the upper bound, every
$\rho\in\BEXT_{n,n}(d_A,d_B)$ is within $\delta_n$ of a
separable state by Theorem~\ref{thm:two-sided-bose-extension}; since
$0\preceq M\preceq I$,
$\Tr(M\rho)\le h_{\SEP}(M)+\delta_n$.

The two-sided lift acts on a space of dimension
\[
 L_{d_A,d_B,n}=
 \binom{d_A+n-1}{n}
 \binom{d_B+n-1}{n}.
\]
Its entries are computable in time polynomial in $L_{d_A,d_B,n}$ and the
bit complexity of $M$.  For $d_A=d_B=d$, take
$n=\lceil\sqrt{d-1}/\epsilon\rceil$ and compute its largest eigenvalue
to precision $\epsilon/2$.  Then $\delta_n\le\epsilon/2$, and the standard
binomial estimate gives
\[
 \log L_{d,d,n}
 =
 O\left(
 \frac{\sqrt d}{\epsilon}\log(2+\epsilon\sqrt d)
 \right).
\]
This proves \eqref{eq:bss-runtime-epsilon}; taking
$\epsilon=\Theta(c-s)$ proves the promise statement.
\end{proof}

The algorithm of Barak, Kothari, and Steurer for explicitly given BSS
instances assumes perfect completeness, meaning $h_{\SEP}(M)=1$ in the
yes case.  Theorem~\ref{thm:bss-main} removes this assumption while
retaining a running time subexponential in the local
dimension~\cite[Remark~1.4]{BKS17}.

\subsection{Trace-norm separability testing}
\label{sec:separability-testing}

Given a bipartite state $\rho$, trace-norm separability testing asks
whether $\rho$ is separable or is at trace distance at least $\epsilon$
from every separable state.  This is the weak-membership problem for the
separable set in trace distance.  The BSS algorithm from the preceding
subsection gives a semidefinite relaxation that approximates
$h_{\SEP}(M)$ with an additive error independent of $M$.  General
reductions between weak optimization and weak membership can lose
polynomial factors in the
dimension~\cite{GLS93}.  Harrow and Montanaro proved the direction needed
here with only a constant-factor loss in
accuracy~\cite[Proposition~16]{HM13}.  Applying their reduction to our
BSS approximation gives the claimed algorithm.  By duality
(Lemma~\ref{lem:compact-convex-distance-duality}),
\begin{equation}
 \Dtr\bigl(\rho,\SEP(d_A,d_B)\bigr)
 =\max_{0\preceq M\preceq I}
 \left\{
 \Tr(M\rho)-h_{\SEP}(M)
 \right\}.
 \label{eq:sep-distance-dual}
\end{equation}

For every operator $0\preceq M\preceq I$,
Theorem~\ref{thm:bss-main} gives
\begin{equation}
 h_{\SEP}(M)
 \le h_{\BEXT_{n,n}}(M)
 \le h_{\SEP}(M)+\delta_n,
 \qquad
 \delta_n:=
 \frac{\sqrt{\min\{d_A,d_B\}-1}}{2n}.
 \label{eq:sep-support-approximation}
\end{equation}
Replacing $h_{\SEP}$ by this relaxation in
\eqref{eq:sep-distance-dual}, define
\begin{equation}
 \widehat D_n(\rho)
 :=\max_{0\preceq M\preceq I}
 \left\{
 \Tr(M\rho)-h_{\BEXT_{n,n}}(M)
 \right\}.
 \label{eq:dual-separability-relaxation}
\end{equation}

\begin{proposition}[From weak optimization to weak membership]
\label{prop:dual-separability-relaxation}
For every state $\rho\in\D(\cH_A\ot\cH_B)$ and every $n\ge1$,
\begin{equation}
 \Dtr\bigl(\rho,\SEP(d_A,d_B)\bigr)
 -\delta_n
 \le \widehat D_n(\rho)
 \le \Dtr\bigl(\rho,\SEP(d_A,d_B)\bigr).
 \label{eq:dual-separability-sandwich}
\end{equation}
In particular, $\widehat D_n(\rho)=0$ whenever $\rho$ is separable.  Moreover, \eqref{eq:dual-separability-relaxation} is the semidefinite program
\begin{equation}
\begin{aligned}
 \textnormal{maximize}\quad & \Tr(M\rho)-t\\
 \textnormal{subject to}\quad
 &0\preceq M\preceq I,\\
 &M^{[n,n]}
 \preceq
 tI_{\Sym^n(\cH_A)\ot\Sym^n(\cH_B)},
\end{aligned}
\label{eq:dual-separability-sdp}
\end{equation}
with variables $M\in\Herm(\cH_A\ot\cH_B)$ and $t\in\mathbb R$.
\end{proposition}

\begin{proof}
For every operator $0\preceq M\preceq I$,
\eqref{eq:sep-support-approximation} gives
\[
 \Tr(M\rho)-h_{\SEP}(M)-\delta_n
 \le \Tr(M\rho)-h_{\BEXT_{n,n}}(M)
 \le \Tr(M\rho)-h_{\SEP}(M).
\]
Maximizing over $M$ and applying \eqref{eq:sep-distance-dual} proves
\eqref{eq:dual-separability-sandwich}.  If $\rho$ is separable, the
upper bound gives $\widehat D_n(\rho)\le0$, while $M=0$ gives the reverse
inequality.  Hence $\widehat D_n(\rho)=0$.

By \eqref{eq:two-sided-bss-relaxation},
$h_{\BEXT_{n,n}}(M)
=\lambda_{\max}(M^{[n,n]})$.  For any Hermitian
operator $B$, the inequality $B\preceq tI$ holds exactly when
$\lambda_{\max}(B)\le t$.  Therefore,
\[
 M^{[n,n]}
 \preceq
 tI_{\Sym^n(\cH_A)\ot\Sym^n(\cH_B)}
 \quad\Longleftrightarrow\quad
 t\ge h_{\BEXT_{n,n}}(M).
\]
For fixed $M$, the objective $\Tr(M\rho)-t$ decreases as $t$ increases,
so its maximum is attained at the smallest feasible value,
$t=h_{\BEXT_{n,n}}(M)$.  Thus eliminating $t$ from
\eqref{eq:dual-separability-sdp} recovers exactly
\eqref{eq:dual-separability-relaxation}.  Finally,
$M\mapsto M^{[n,n]}$ is linear, so the matrix in the displayed
semidefinite constraint depends linearly on the optimization variables.
\end{proof}

\begin{theorem}[Subexponential trace-norm separability algorithm]
\label{thm:trace-norm-separability-algorithm}
Let $\rho$ be an explicitly given state on $\C^d\ot\C^d$.  For every
$\epsilon\in(0,1]$, the distance $\Dtr(\rho,\SEP(d,d))$ can be
approximated to additive error $\epsilon$ deterministically.  Up to a
factor polynomial in the bit complexity of $\rho$, the running time is
\begin{equation}
 \exp\left(
 O\left(
 \frac{\sqrt d}{\epsilon}
 \log(2+\epsilon\sqrt d)
 \right)
 \right).
 \label{eq:trace-norm-separability-runtime}
\end{equation}
In the same running time one can distinguish
\begin{equation}
 \rho\in\SEP(d,d)
 \qquad\text{from}\qquad
 \Dtr(\rho,\SEP(d,d))\ge\epsilon.
 \label{eq:trace-norm-separability-promise}
\end{equation}
\end{theorem}

\begin{proof}
Choose
\[
 n:=\left\lceil\frac{\sqrt{d-1}}{\epsilon}\right\rceil.
\]
Then $\delta_n\le\epsilon/2$, and
Proposition~\ref{prop:dual-separability-relaxation} gives
\begin{equation}
 \Dtr(\rho,\SEP(d,d))-\delta_n
 \le\widehat D_n(\rho)
 \le \Dtr(\rho,\SEP(d,d)).
 \label{eq:separability-algorithm-sandwich}
\end{equation}
Solve \eqref{eq:dual-separability-sdp} to additive precision
$\epsilon/2$, and denote the resulting value by
$\widetilde D_n(\rho)$.  By \eqref{eq:separability-algorithm-sandwich},
\[
 \left|\widetilde D_n(\rho)-\Dtr(\rho,\SEP(d,d))\right|
 \le\frac{\epsilon}{2}+\delta_n
 \le\epsilon.
\]
Thus $\widetilde D_n(\rho)$ is the required approximation.

The semidefinite constraint
$M^{[n,n]}\preceq tI_{\Sym^n(\C^d)\ot\Sym^n(\C^d)}$ has matrix dimension
\[
 L_{d,d,n}=\binom{d+n-1}{n}^{2}.
\]
The remaining constraints $0\preceq M\preceq I$ have dimension $d^2$.
The construction and dimension estimate in the proof of
Theorem~\ref{thm:bss-main}, together with standard
semidefinite optimization, therefore give
\eqref{eq:trace-norm-separability-runtime}.

For the promise problem, run the additive approximation just proved with
accuracy $\epsilon/3$.  Its output is at most $\epsilon/3$ when $\rho$ is
separable and at least $2\epsilon/3$ when
$\Dtr(\rho,\SEP(d,d))\ge\epsilon$.  Comparing the output with $\epsilon/2$
distinguishes the two cases.  Replacing $\epsilon$ by $\epsilon/3$
changes the running time only by constant factors in the exponent.
\end{proof}

\section{A dimension-free bosonic Hilbert--Schmidt de Finetti theorem}
\label{sec:hs-definetti}

Brand\~ao and Christandl identified local-dimension dependence as a
crucial bottleneck for quantum de Finetti theorems and showed that the
linear barrier can be broken under one-way LOCC measurements, obtaining
only logarithmic dependence on the local dimension~\cite{BC12}.  In
Hilbert--Schmidt distance, the theorem below removes the local-dimension
dependence entirely for bosonic approximation.

\begin{theorem}[Dimension-free bosonic Hilbert--Schmidt de Finetti upper bound]
\label{thm:HS-main}
For every $N,d\ge2$,
\begin{equation}
 \DHHS\Bigl(
  \BOS_N^{(2)}(d),\,
  \SEPH^{(2)}(d)
 \Bigr)
 \le
 \min\left\{
 \frac{\sqrt{d-1}}{N-1},
 \frac{8}{\sqrt{N-1}}
 \right\}.
 \label{eq:HS-upper}
\end{equation}
For each fixed $d$, the first term in \eqref{eq:HS-upper} retains the
faster rate $O(\sqrt d/N)$.
\end{theorem}

\subsection{Spectral truncation and the dimension-free Hilbert--Schmidt bound}

The dimension-dependent bound in \eqref{eq:HS-upper} requires no new
argument.  Theorem~\ref{thm:bos-main} gives a state
$\tau\in\SEPH^{(2)}(d)$ such that
\[
 \norm{\rho_N^{(2)}-\tau}_2
 \le\norm{\rho_N^{(2)}-\tau}_1
 \le\frac{\sqrt{d-1}}{N-1}.
\]
It remains to prove the dimension-free bound.

\proofparagraph{Setup.}
By convexity of
$\DHS(\,\cdot\,,\SEPH^{(2)}(d))$, it is enough to consider
a pure bosonic state
$\rho_N=\ketbra{\psi}$ with $\psi\in\Sym^N\cH$ a unit vector.  Set
\[
 \rho_1:=\rho_N^{(1)},
 \qquad
 \rho_2:=\rho_N^{(2)}.
\]
Split the spectrum of $\rho_1$ at the scale $1/N$:
\begin{equation}
 P:=\mathbf 1_{(1/N,\infty)}(\rho_1),
 \qquad
 Q:=I-P,
 \qquad
 r:=\rank P.
 \label{eq:HS-PQ-def}
\end{equation}
We call $P\cH$ the spectral head and $Q\cH$ the spectral tail.  The
cutoff has two immediate consequences: since $\Tr\rho_1=1$, the head
has dimension $r<N$, and the tail is spectrally flat,
\begin{equation}
 Q\rho_1Q\preceq\frac1NQ.
 \label{eq:HS-spectral-tail-flatness}
\end{equation}
By Fact~\ref{fact:symmetric-orthogonal-sum}, the induced two-site
decomposition is
\begin{equation}
 \Sym^2\cH
 =
 \Sym^2(P\cH)\oplus(P\cH\vee Q\cH)\oplus\Sym^2(Q\cH).
 \label{eq:HS-two-site-decomp}
\end{equation}
Let $\Pi_{PP}$, $\Pi_{PQ}$, and $\Pi_{QQ}$ be the projections onto the
three summands.  Group the two summands containing at least one tail
direction by setting
\[
 \Pi_{\mathrm{tail}}
 :=
 \Pi_{PQ}+\Pi_{QQ}
 =
 I-\Pi_{PP}.
\]
The corresponding diagonal blocks of $\rho_2$ are
\begin{equation}
 \rho_2^{\mathrm{head}}
 :=
 \Pi_{PP}\rho_2\Pi_{PP},
 \qquad
 \rho_2^{\mathrm{tail}}
 :=
 \Pi_{\mathrm{tail}}\rho_2\Pi_{\mathrm{tail}}.
 \label{eq:HS-tail-block-def}
\end{equation}
Let $X$ denote the off-diagonal head--tail block,
\[
 X:=\Pi_{PP}\rho_2\Pi_{\mathrm{tail}}.
\]
Relative to
$\operatorname{range}\Pi_{PP}\oplus\operatorname{range}\Pi_{\mathrm{tail}}$,
the state has the positive block form
\begin{equation}
 \rho_2
 =
 \begin{pmatrix}
  \rho_2^{\mathrm{head}}&X\\
  X^\dagger&\rho_2^{\mathrm{tail}}
 \end{pmatrix}
 \succeq0.
 \label{eq:HS-head-tail-decomp}
\end{equation}
Since $\Tr\rho_2=1$,
Fact~\ref{fact:positive-block-discard} gives
\[
 \norm{\rho_2-\rho_2^{\mathrm{head}}}_2^2
 \le
 2\norm{\rho_2^{\mathrm{tail}}}_\infty.
\]
Thus the tail need not have small trace, and the off-diagonal block
$X$ need not be estimated separately: an operator-norm bound on the diagonal
tail block controls the entire discarded part.  The rank bound $r<N$
makes the head accessible to the optimal trace-norm de Finetti theorem, while
\eqref{eq:HS-spectral-tail-flatness} will force
$\rho_2^{\mathrm{tail}}$ to be small in operator norm.

\proofparagraph{Proof architecture.}
The proof consists of two substantive estimates followed by a
trace-restoration step.  First, we construct a Hartree mixture
$\tau_{\mathrm{head}}$ that approximates
$\rho_2^{\mathrm{head}}$ and has the same trace.  Second, we transfer
the one-site flatness \eqref{eq:HS-spectral-tail-flatness} to the
two-site estimate
$\norm{\rho_2^{\mathrm{tail}}}_\infty=O(N^{-1})$.
Fact~\ref{fact:positive-block-discard} then removes the tail at
Hilbert--Schmidt cost $O(N^{-1/2})$.  The missing trace is restored by a
Hartree mixture
$\tau_{\mathrm{res}}$ supported on $\Sym^2(Q\cH)$.  Set
\begin{equation}
 \tau
 :=
 \tau_{\mathrm{head}}+\tau_{\mathrm{res}}.
 \label{eq:HS-final-mixture}
\end{equation}
Their traces will add to one, so $\tau\in\SEPH^{(2)}(d)$.  The proof
is organized around the following error decomposition:
\begin{align}
 \norm{\rho_2-\tau}_2
 &=
 \norm{
 (\rho_2^{\mathrm{head}}-\tau_{\mathrm{head}})
  +(\rho_2-\rho_2^{\mathrm{head}})
  -\tau_{\mathrm{res}}
 }_2\notag\\
 &\le
 \norm{\rho_2^{\mathrm{head}}-\tau_{\mathrm{head}}}_2
 +\norm{\rho_2-\rho_2^{\mathrm{head}}}_2
 +\norm{\tau_{\mathrm{res}}}_2\notag\\
 &\le
 \frac1{\sqrt{N-1}}
 +\sqrt{\frac{24}{N-1}}
 +\frac{2\sqrt2}{N}
 <\frac8{\sqrt{N-1}}.
 \label{eq:HS-proof-roadmap}
\end{align}
The three terms in the triangle inequality are bounded, in order, by
Lemmas~\ref{lem:HS-round-head}, \ref{lem:HS-remove-tail}, and
\ref{lem:HS-residual-mass}.

\subsubsection{Hartree approximation of the spectral head}

We begin with the low-dimensional head, where the optimal trace-norm
de Finetti theorem enters.

\begin{lemma}[Rounding the spectral head]
\label{lem:HS-round-head}
There is a Hartree mixture
$\tau_{\mathrm{head}}$, supported on $\Sym^2(P\cH)$, such that
\[
 \Tr\tau_{\mathrm{head}}
 =
 \Tr\rho_2^{\mathrm{head}},
 \qquad
 \norm{\rho_2^{\mathrm{head}}-\tau_{\mathrm{head}}}_2
 \le
 \frac1{\sqrt{N-1}}.
\]
\end{lemma}

\begin{proof}
Fact~\ref{fact:symmetric-orthogonal-sum}, applied to
$\cH=P\cH\oplus Q\cH$, gives
\[
 \Sym^N\cH
 \cong
 \bigoplus_{k=0}^N
 \Sym^k(P\cH)\ot\Sym^{N-k}(Q\cH).
\]
Let $\Pi_k$ be the projection onto the $k$th summand and set
\[
 p_k:=\Tr(\Pi_k\rho_N).
\]
The $PP$-compression is block diagonal with respect to this
decomposition and vanishes on the summands with $k<2$.  For $k\ge2$,
after projecting two sites into $P\cH$, the residual component lies in
$\Sym^{k-2}(P\cH)\ot\Sym^{N-k}(Q\cH)$.  These residual spaces are
orthogonal for distinct values of $k$, so only the diagonal blocks
$\Pi_k\rho_N\Pi_k$ contribute to $\rho_2^{\mathrm{head}}$.

For $p_k>0$, define
\[
 \sigma_k
 :=
 \frac{\Pi_k\rho_N\Pi_k}{p_k}
 \in\D(\operatorname{range}\Pi_k),
 \qquad
 \rho_{P,k}
 :=
 \Tr_{\Sym^{N-k}(Q\cH)}
 \left(\sigma_k\right)
 \in\D(\Sym^k(P\cH)).
\]
Write $\rho_{P,k}^{(2)}$ for its two-site marginal.
Within the $k$th slice, exactly $k$ of the $N$ sites lie in $P\cH$.
Thus the probability that two sites sampled without replacement both
lie in $P\cH$ is
\[
 \frac{\binom{k}{2}}{\binom{N}{2}}
 =
 \frac{k(k-1)}{N(N-1)}.
\]
Here $P_i$ denotes $P$ acting on the $i$th site.  Equivalently,
permutation symmetry and the identity
$\sum_{i\ne j}P_iP_j=k(k-1)I$ on
$\operatorname{range}\Pi_k$ give
\[
 \Tr(P_1P_2\sigma_k)
 =
 \frac{1}{N(N-1)}
 \Tr\left(\left(\sum_{i\ne j}P_iP_j\right)\sigma_k\right)
 =
 \frac{k(k-1)}{N(N-1)}.
\]
Under the canonical identification of the $k$th slice, its normalized
$PP$-block is $\rho_{P,k}^{(2)}$.  Therefore
\[
 \Pi_{PP}\Tr_{3,\ldots,N}(\sigma_k)\Pi_{PP}
 =
 \frac{k(k-1)}{N(N-1)}\rho_{P,k}^{(2)}.
\]
Using $\Pi_k\rho_N\Pi_k=p_k\sigma_k$ and summing the diagonal slice
contributions gives
\begin{equation}
 \rho_2^{\mathrm{head}}
 =
 \sum_{\substack{2\le k\le N\\p_k>0}}
 p_k\frac{k(k-1)}{N(N-1)}\rho_{P,k}^{(2)}.
 \label{eq:HS-P-slice-decomp}
\end{equation}

If $r\le1$, every term in \eqref{eq:HS-P-slice-decomp} is Hartree, so
we may take $\tau_{\mathrm{head}}:=\rho_2^{\mathrm{head}}$.  Assume
now that $r\ge2$.  Applying Theorem~\ref{thm:bos-main} to
$\rho_{P,k}$ gives a Hartree mixture $\tau_k$ supported on
$\Sym^2(P\cH)$ and satisfying
\begin{equation}
 \norm{\rho_{P,k}^{(2)}-\tau_k}_1
 \le\frac{\sqrt{r-1}}{k-1}.
 \label{eq:HS-P-trace-approx}
\end{equation}
Define the Hartree mixture
\begin{equation}
 \tau_{\mathrm{head}}
 :=\sum_{\substack{2\le k\le N\\p_k>0}}
 \frac{p_k k(k-1)}{N(N-1)}\tau_k.
 \label{eq:HS-tau-P}
\end{equation}
This operator has the same trace as $\rho_2^{\mathrm{head}}$.  The key
point is that the factor $k-1$ in the de Finetti error cancels the
same factor in the pair weight:
\begin{align}
 \norm{\rho_2^{\mathrm{head}}-\tau_{\mathrm{head}}}_2
 &\le
 \sum_{\substack{2\le k\le N\\p_k>0}}
 \frac{p_k k(k-1)}{N(N-1)}
 \norm{\rho_{P,k}^{(2)}-\tau_k}_1\notag\\
 &\le
 \frac{\sqrt{r-1}}{N(N-1)}
 \sum_{k=2}^N p_k k\notag\\
 &\le\frac{\sqrt{r-1}}{N-1}
 \le\frac1{\sqrt{N-1}}.
 \label{eq:HS-P-approx}
\end{align}
For $r\le1$, the same estimate holds with zero left-hand side.
\end{proof}

\subsubsection{Operator-norm control of the spectral tail}

The mixed and tail--tail diagonal blocks are bounded directly, after
which positive block estimates assemble them into an operator-norm
bound on $\rho_2^{\mathrm{tail}}$.

\begin{lemma}[Mixed and tail--tail block estimates]
\label{lem:HS-spectral-tail-blocks}
With the notation above,
\begin{align}
 \norm{\Pi_{PQ}\rho_2\Pi_{PQ}}_\infty
 &\le\frac4N,
 \label{eq:HS-PQ-op}\\
 \norm{\Pi_{QQ}\rho_2\Pi_{QQ}}_\infty
 &\le\frac8{N-1}.
 \label{eq:HS-QQ-op}
\end{align}
\end{lemma}

\begin{proof}
Recall that $\rho_N=\ketbra{\psi}$, where
$\psi\in\Sym^N\cH$ is a unit vector.

\proofparagraph{The mixed block.}
Let $\phi\in P\cH\vee Q\cH$ be a unit vector.  Its Schmidt decomposition is
\begin{equation}
 \phi=\sum_{i=1}^{r_\phi}\sqrt{\mu_i}\,p_i\vee q_i,
 \qquad
 \mu_i>0,
 \qquad
 \sum_i\mu_i=1,
 \label{eq:HS-PQ-Schmidt}
\end{equation}
where $(p_i)$ and $(q_i)$ are orthonormal families in $P\cH$ and $Q\cH$,
respectively, and $r_\phi\le r<N$.  Since $\psi$ is symmetric,
\[
 C_{p_i\vee q_i}(\psi)
 =
 \frac1{\sqrt2}
 \bigl(
 (\bra{p_i}\ot\bra{q_i}\ot I)
 +(\bra{q_i}\ot\bra{p_i}\ot I)
 \bigr)\psi
 =
 \sqrt2\,C_{p_i}\bigl(C_{q_i}(\psi)\bigr).
\]
Although the Schmidt directions $q_i$ need not be eigenvectors of
$\rho_1$, each is a unit vector in $Q\cH$, so
\eqref{eq:HS-spectral-tail-flatness} gives
\[
 \norm{C_{q_i}(\psi)}^2
 =
 \bra{q_i}\rho_1\ket{q_i}
 =
 \bra{q_i}Q\rho_1Q\ket{q_i}
 \le\frac1N.
\]
Since the Schmidt weights sum to one, it follows that
\begin{equation}
 \sum_i\mu_i\norm{C_{q_i}(\psi)}^2
 \le\frac1N\sum_i\mu_i
 =\frac1N.
 \label{eq:HS-PQ-weighted-tail}
\end{equation}
For every unit vector $\Xi\in\Sym^{N-2}\cH$, use the adjoint relation
\[
 \braket{\Xi}{C_{p_i}(C_{q_i}(\psi))}
 =
 \braket{C_{p_i}^*(\Xi)}{C_{q_i}(\psi)}.
\]
Together with Cauchy--Schwarz,
Fact~\ref{fact:symmetric-contraction-bessel}, specifically
\eqref{eq:insertion-Bessel}, with $k=N-1$ and $\ell=r_\phi$, and
\eqref{eq:HS-PQ-weighted-tail}, this gives
\begin{align*}
 \abs{\braket{\Xi}{C_\phi(\psi)}}
 &=
 \abs{\sum_i\sqrt{\mu_i}
 \braket{\Xi}{C_{p_i\vee q_i}(\psi)}}\\
 &=
 \sqrt2\abs{\sum_i\sqrt{\mu_i}
 \braket{\Xi}{C_{p_i}\bigl(C_{q_i}(\psi)\bigr)}}\\
 &=
 \sqrt2\abs{\sum_i\sqrt{\mu_i}
 \braket{C_{p_i}^*(\Xi)}{C_{q_i}(\psi)}}\\
 &\le\sqrt2
 \left(\sum_i\norm{C_{p_i}^*(\Xi)}^2\right)^{1/2}
 \left(\sum_i\mu_i\norm{C_{q_i}(\psi)}^2\right)^{1/2}\\
 &\le
 \sqrt2
 \sqrt{\frac{r_\phi+N-2}{N-1}}\,
 \frac1{\sqrt N}\\
 &\le\frac2{\sqrt N}.
\end{align*}
The last inequality uses $r_\phi<N$.
Taking the supremum over unit vectors
$\Xi\in\Sym^{N-2}\cH$ therefore gives
\[
 \norm{C_\phi(\psi)}
 \le
 \frac2{\sqrt N}.
\]
Since $\phi\in P\cH\vee Q\cH$, we have $\Pi_{PQ}\phi=\phi$.
Fact~\ref{fact:symmetric-partial-inner-products}, specifically the
partial-inner-product identity \eqref{eq:contraction-marginals}, now gives
\[
 \bra\phi\Pi_{PQ}\rho_2\Pi_{PQ}\ket\phi
 =
 \bra\phi\rho_2\ket\phi
 =
 \norm{C_\phi(\psi)}^2
 \le
 \frac4N.
\]
The compression $\Pi_{PQ}\rho_2\Pi_{PQ}$ is positive semidefinite, so
its operator norm is the supremum of the left-hand side over unit
vectors $\phi\in P\cH\vee Q\cH$.  This proves
\eqref{eq:HS-PQ-op}.

\proofparagraph{The tail--tail block.}
Let $\phi\in\Sym^2(Q\cH)$ be a unit vector.  A Takagi decomposition gives
\begin{equation}
 \phi=\sum_{i=1}^{m}\sqrt{\mu_i}\,q_i^{\ot2},
 \qquad
 \mu_1\ge\mu_2\ge\cdots\ge0,
 \qquad
 \sum_i\mu_i=1,
 \label{eq:HS-QQ-Takagi}
\end{equation}
for an orthonormal family $(q_i)\subseteq Q\cH$.  Split
\[
 \phi=\phi_{\le N}+\phi_{>N},
 \qquad
 \phi_{\le N}:=\sum_{i\le N}\sqrt{\mu_i}\,q_i^{\ot2},
 \qquad
 \phi_{>N}:=\sum_{i>N}\sqrt{\mu_i}\,q_i^{\ot2},
\]
with empty terms omitted.  The leading $N$ Takagi terms are controlled
by applying Cauchy--Schwarz with the weights $\mu_i$ multiplying
$\norm{C_{q_i}(\psi)}^2$, exactly as in the mixed block.  The same
pointwise tail estimate gives
\[
 \sum_{i\le N}\mu_i\norm{C_{q_i}(\psi)}^2
 \le\frac1N\sum_{i\le N}\mu_i
 \le\frac1N.
\]
Therefore, for every unit $\Xi\in\Sym^{N-2}\cH$,
\begin{align*}
 \abs{\braket{\Xi}{C_{\phi_{\le N}}(\psi)}}
 &\le
 \left(\sum_{i\le N}\norm{C_{q_i}^*(\Xi)}^2\right)^{1/2}
 \left(\sum_{i\le N}\mu_i\norm{C_{q_i}(\psi)}^2\right)^{1/2}\\
 &\le\sqrt{\frac2N}.
\end{align*}
Hence
\begin{equation}
 \norm{C_{\phi_{\le N}}(\psi)}^2\le\frac2N.
 \label{eq:HS-QQ-head}
\end{equation}

For the remaining Takagi terms, monotonicity of $(\mu_i)$ and
$\sum_i\mu_i=1$ imply
$\mu_{N+1}\le1/(N+1)$.  There may now be arbitrarily many terms, so we
instead attach the weights $\mu_i$ to
$\norm{C_{q_i}^*(\Xi)}^2$ and leave
$\norm{C_{q_i}(\psi)}^2$ unweighted.  Thus
\begin{align*}
 \abs{\braket{\Xi}{C_{\phi_{>N}}(\psi)}}
 &\le
 \left(\sum_{i>N}\mu_i\norm{C_{q_i}^*(\Xi)}^2\right)^{1/2}
 \left(\sum_{i>N}\norm{C_{q_i}(\psi)}^2\right)^{1/2}.
\end{align*}
To estimate the first factor, apply
Fact~\ref{fact:symmetric-contraction-bessel}, specifically
\eqref{eq:contraction-adjoint-norm}, with $k=N-1$:
\begin{align*}
 \sum_{i>N}\mu_i\norm{C_{q_i}^*(\Xi)}^2
 &=
 \frac1{N-1}
 \left(
 \sum_{i>N}\mu_i
 +(N-2)\sum_{i>N}\mu_i\norm{C_{q_i}(\Xi)}^2
 \right)\\
 &\le
 \frac1{N-1}
 \left(1+\frac{N-2}{N+1}\right)
 <\frac2{N-1}.
\end{align*}
Here $\sum_{i>N}\mu_i\le1$, while
\eqref{eq:contraction-Bessel} and
$\mu_i\le\mu_{N+1}\le1/(N+1)$ for $i>N$ give
\[
 \sum_{i>N}\mu_i\norm{C_{q_i}(\Xi)}^2
 \le
 \mu_{N+1}\sum_{i>N}\norm{C_{q_i}(\Xi)}^2
 \le\frac1{N+1}.
\]
The same Bessel bound, applied to $\psi$, gives
\[
 \sum_{i>N}\norm{C_{q_i}(\psi)}^2\le1.
\]
Therefore
\begin{equation}
 \norm{C_{\phi_{>N}}(\psi)}^2\le\frac2{N-1}.
 \label{eq:HS-QQ-tail}
\end{equation}
By linearity,
$C_\phi=C_{\phi_{\le N}}+C_{\phi_{>N}}$.  Hence the triangle
inequality, together with \eqref{eq:HS-QQ-head} and
\eqref{eq:HS-QQ-tail}, gives
\[
 \norm{C_\phi(\psi)}
 \le
 \norm{C_{\phi_{\le N}}(\psi)}
 +\norm{C_{\phi_{>N}}(\psi)}
 \le
 \sqrt{\frac2N}+\sqrt{\frac2{N-1}}.
\]
Since $\phi\in\Sym^2(Q\cH)$, we have $\Pi_{QQ}\phi=\phi$.
Fact~\ref{fact:symmetric-partial-inner-products}, specifically
\eqref{eq:contraction-marginals}, therefore implies
\[
 \bra\phi\Pi_{QQ}\rho_2\Pi_{QQ}\ket\phi
 =
 \bra\phi\rho_2\ket\phi
 =
 \norm{C_\phi(\psi)}^2
 \le
 \left(\sqrt{\frac2N}+\sqrt{\frac2{N-1}}\right)^2
 \le
 \frac8{N-1}.
\]
Finally, $\Pi_{QQ}\rho_2\Pi_{QQ}$ is positive semidefinite, so its
operator norm is the supremum of the left-hand side over unit vectors
$\phi\in\Sym^2(Q\cH)$.  This proves \eqref{eq:HS-QQ-op}.
\end{proof}

\begin{lemma}[Removing the spectral tail]
\label{lem:HS-remove-tail}
The head compression in \eqref{eq:HS-tail-block-def} satisfies
\[
 \norm{\rho_2-\rho_2^{\mathrm{head}}}_2
 \le
 \sqrt{\frac{24}{N-1}}.
\]
\end{lemma}

\begin{proof}
The range of $\Pi_{\mathrm{tail}}$ is the direct sum of the mixed and
tail--tail spaces in \eqref{eq:HS-two-site-decomp}.  Relative to
this decomposition,
\[
 \rho_2^{\mathrm{tail}}
 =
 \begin{pmatrix}
  \Pi_{PQ}\rho_2\Pi_{PQ}
  &
  \Pi_{PQ}\rho_2\Pi_{QQ}
  \\
  \Pi_{QQ}\rho_2\Pi_{PQ}
  &
  \Pi_{QQ}\rho_2\Pi_{QQ}
 \end{pmatrix}
 \succeq0.
\]
Fact~\ref{fact:positive-block} and
Lemma~\ref{lem:HS-spectral-tail-blocks} therefore give
\begin{align}
 \norm{\rho_2^{\mathrm{tail}}}_\infty
 &\le
 \norm{\Pi_{PQ}\rho_2\Pi_{PQ}}_\infty
 +\norm{\Pi_{QQ}\rho_2\Pi_{QQ}}_\infty
 \notag\\
 &\le\frac4N+\frac8{N-1}
 \le\frac{12}{N-1}.
 \label{eq:HS-tail-op}
\end{align}

Apply the trace-one case of
Fact~\ref{fact:positive-block-discard} to the block decomposition
\eqref{eq:HS-head-tail-decomp}.  This controls the diagonal tail and
the off-diagonal head--tail blocks simultaneously:
\begin{equation}
 \norm{\rho_2-\rho_2^{\mathrm{head}}}_2^2
 \le
 2\norm{\rho_2^{\mathrm{tail}}}_\infty
 \le\frac{24}{N-1}.
 \label{eq:HS-remove-tail}
\end{equation}
Thus positivity converts an $O(N^{-1})$ operator-norm estimate on the
discarded block into an $O(N^{-1/2})$ Hilbert--Schmidt estimate for
the entire discarded part, including the off-diagonal blocks.
\end{proof}

\subsubsection{Completion of the dimension-free bound}

At this point both substantive estimates are complete: the head has a
Hartree approximation of the correct trace, and the entire discarded
part is small in Hilbert--Schmidt norm.  It remains only to restore the
missing trace.  Spreading this mass uniformly over the symmetric tail
costs only $O(N^{-1})$.

\begin{lemma}[Restoring the missing trace]
\label{lem:HS-residual-mass}
There is a Hartree mixture
$\tau_{\mathrm{res}}$, supported on $\Sym^2(Q\cH)$, such that
\[
 \Tr\tau_{\mathrm{res}}
 =
 1-\Tr\rho_2^{\mathrm{head}},
 \qquad
 \norm{\tau_{\mathrm{res}}}_2
 \le
 \frac{2\sqrt2}{N}.
\]
\end{lemma}

\begin{proof}
Set
\[
 \delta:=1-\Tr\rho_2^{\mathrm{head}},
 \qquad
 s:=\dim(Q\cH).
\]
Since $\Pi_{\mathrm{tail}}\preceq Q\ot I+I\ot Q$ on
$\Sym^2\cH$, the spectral-tail bound
\eqref{eq:HS-spectral-tail-flatness} gives
\[
 \delta
 =
 \Tr(\Pi_{\mathrm{tail}}\rho_2)
 \le
 \Tr\bigl((Q\ot I+I\ot Q)\rho_2\bigr)
 =
 2\Tr(Q\rho_1)
 \le
 \frac{2s}{N}.
\]
If $s=0$, then $\delta=0$, and we take $\tau_{\mathrm{res}}=0$.
Suppose that $s\ge1$, and let $\mu_Q$ denote Haar probability measure
on the unit sphere of $Q\cH$.  Define
\[
 \tau_{\mathrm{res}}
 :=
 \delta\frac{\Pi_{QQ}}{\dim\Sym^2(Q\cH)}
 =
 \delta\int\ketbra{u}^{\ot2}\,d\mu_Q(u).
\]
Thus $\tau_{\mathrm{res}}$ is a Hartree mixture supported on
$\Sym^2(Q\cH)$, and its trace is $\delta$.  Since
$\dim\Sym^2(Q\cH)=s(s+1)/2$,
\begin{align*}
 \norm{\tau_{\mathrm{res}}}_2
 &=
 \delta\sqrt{\frac{2}{s(s+1)}} \le
 \frac{2\sqrt2}{N}\sqrt{\frac{s}{s+1}}
 \le
 \frac{2\sqrt2}{N}.\qedhere
\end{align*}
\end{proof}

\begin{proof}[Completion of the proof of the dimension-free bound]
Let $\tau_{\mathrm{head}}$ and $\tau_{\mathrm{res}}$ be supplied by
Lemmas~\ref{lem:HS-round-head} and
\ref{lem:HS-residual-mass}, and let $\tau$ be their sum as in
\eqref{eq:HS-final-mixture}.
The two summands are Hartree mixtures whose traces add
to one, so $\tau\in\SEPH^{(2)}(d)$.  The three estimates used in
\eqref{eq:HS-proof-roadmap} are precisely the conclusions of
Lemmas~\ref{lem:HS-round-head}, \ref{lem:HS-remove-tail}, and
\ref{lem:HS-residual-mass}.  Hence
\eqref{eq:HS-proof-roadmap} gives
$\norm{\rho_2-\tau}_2<8/\sqrt{N-1}$.
This is the dimension-free term in \eqref{eq:HS-upper}.
\end{proof}

\subsection{Two-sided Bose-symmetric extendible states}
\label{sec:HS-two-sided-bose}

Since Theorem~\ref{thm:HS-main} is dimension free, one might first try
a tensor-product construction: pair $A_iB_i$ into $n$ bosonic sites
with local space $\cH_A\ot\cH_B$ and apply the theorem.  The desired
cross marginal is then obtained from the resulting two-site marginal
by tracing out $B_1A_2$, and the image of each Hartree state under this
partial trace is separable.  However, partial trace has
Hilbert--Schmidt $2\to2$ norm $\sqrt{d_Ad_B}$, which destroys dimension
independence.

To avoid this loss, we instead use a direct-sum construction: the two
types of registers are placed in orthogonally labeled summands of a
single bosonic system, so that the desired state is a compression of a
two-site marginal rather than a partial trace.   The
direct-sum construction is inspired by the symmetrization idea of Liu
and Wu~\cite[Sec.~4.2.1]{LW26}.

This label-bundling
construction is a convenient technical workaround for the applications in the next section.  A more systematic
approach would repeat the spectral-truncation argument of the preceding
subsection directly for two-sided extensions, using
Theorem~\ref{thm:two-sided-bose-extension} on the low-dimensional
spectral head; we leave such a direct proof for future work.

Let $\cH_A=\C^{d_A}$ and
$\cH_B=\C^{d_B}$, and set
\[
 \cH_\oplus
 =
 (\ket0\ot\cH_A)\oplus(\ket1\ot\cH_B).
\]
Under the standard decomposition of a symmetric power of an
orthogonal sum, $\Sym^n(\cH_A)\ot\Sym^n(\cH_B)$ identifies with the
summand of $\Sym^{2n}(\cH_\oplus)$ containing $n$ sites of each label.

The subspace of $\Sym^2(\cH_\oplus)$ containing one site of each label
is identified with $\cH_A\ot\cH_B$ by the isometry
\begin{equation}
 V:\cH_A\ot\cH_B\longrightarrow\Sym^2(\cH_\oplus),
 \qquad
 V(x\ot y)
 =
 \frac{
  \ket{0,x}\ot\ket{1,y}
  +\ket{1,y}\ot\ket{0,x}
 }{\sqrt2}.
 \label{eq:label-bundling-isometry}
\end{equation}

\begin{lemma}[Label-bundling embedding]
\label{lem:label-bundling-embedding}
Let $\rho\in\BEXT_{n,n}(d_A,d_B)$.  There is a bosonic state
$\widehat\rho_{2n}\in\D(\Sym^{2n}(\cH_\oplus))$ whose two-site marginal
satisfies
\begin{equation}
 V^\dagger\widehat\rho_{2n}^{(2)}V
 =
 q_n\rho,
 \qquad
 q_n:=\frac{n}{2n-1}.
 \label{eq:cross-label-probability}
\end{equation}
Thus the cross-label block on the left has trace $q_n$.  Equivalently,
measuring the two labels gives one label of each
type with probability $q_n$, and conditioned on this event the state on
$\cH_A\ot\cH_B$ is exactly $\rho$.
\end{lemma}

\begin{proof}
Choose an extension $\widetilde\rho$ from
\eqref{eq:two-sided-bext-definition}.  By
Fact~\ref{fact:symmetric-orthogonal-sum},
\[
 \Sym^{2n}(\cH_\oplus)
 \cong
 \bigoplus_{k=0}^{2n}
 \Sym^k(\cH_A)\ot\Sym^{2n-k}(\cH_B).
\]
We may therefore regard $\widetilde\rho$ as a bosonic state
$\widehat\rho_{2n}$ supported on the $k=n$ summand.

It remains to compute the cross-label block of the two-site marginal.
Under this fixed-label-count identification, the block is the uniform
average of the $n^2$ cross marginals:
\[
 V^\dagger\widehat\rho_{2n}^{(2)}V
 =
 \binom{2n}{2}^{-1}
 \sum_{i=1}^n\sum_{j=1}^n
 \widetilde\rho_{A_iB_j}.
\]
Bose symmetry within each block makes every
$\widetilde\rho_{A_iB_j}$ equal to $\rho$.  Hence
\[
 V^\dagger\widehat\rho_{2n}^{(2)}V
 =
 \frac{n^2}{\binom{2n}{2}}\rho
 =
 \frac{n}{2n-1}\rho.
\]
The coefficient $n^2/\binom{2n}{2}=n/(2n-1)$ is precisely the
probability that a uniformly chosen pair contains one site of each
label.  This proves \eqref{eq:cross-label-probability}.
\end{proof}

We also record how the same isometry acts on Hermitian operators.

\begin{fact}
\label{fact:labeled-operator-embedding}
For $M\in\Herm(\cH_A\ot\cH_B)$, define
\begin{equation}
 A_M:=VMV^\dagger
 \label{eq:labeled-embedding}
\end{equation}
on $\Sym^2(\cH_\oplus)$, acting as zero on
$(\operatorname{range}V)^\perp$.  Then
\begin{equation}
 h(A_M)=\frac12\max\{h_{\SEP}(M),0\},
 \qquad
 \norm{A_M}_2=\norm{M}_2.
 \label{eq:labeled-operator-properties}
\end{equation}
If either $M\succeq0$ or $\Tr M=0$, then $h_{\SEP}(M)\ge0$ and hence
$h(A_M)=\frac12h_{\SEP}(M)$.
\end{fact}

\begin{proof}
Write a unit vector in $\cH_\oplus$ as
\[
 z=\sqrt p\,\ket{0,x}+\sqrt q\,\ket{1,y},
\]
where $x,y$ are unit vectors and $p,q\ge0$ satisfy $p+q=1$.  Since
$V^\dagger z^{\ot2}=\sqrt{2pq}\,x\ot y$,
\[
 \bra{z^{\ot2}}A_M\ket{z^{\ot2}}
 =2pq\bra{x\ot y}M\ket{x\ot y}.
\]
Maximizing over $p,q,x,y$ gives the first identity in
\eqref{eq:labeled-operator-properties}.  Since $V$ is an isometry,
$A_M$ is unitarily equivalent to $M\oplus0$, which proves the
Hilbert--Schmidt norm identity.  Finally, $h_{\SEP}(M)\ge0$ is
immediate when $M\succeq0$.  If $\Tr M=0$, then the average of
$\bra{x\ot y}M\ket{x\ot y}$ over independent Haar-random unit vectors
$x$ and $y$ is $\Tr(M)/(d_Ad_B)=0$, so its maximum is nonnegative.
\end{proof}

Theorem~\ref{thm:two-sided-bose-extension} gives the
dimension-dependent estimate below.  Combining the label-bundling
embedding in Lemma~\ref{lem:label-bundling-embedding} with
Theorem~\ref{thm:HS-main} gives the dimension-free estimate.

\begin{corollary}[Dimension-free Hilbert--Schmidt de Finetti bound for
two-sided Bose-symmetric extendible states]
\label{cor:HS-two-sided-bose-extension}
For every $n,d_A,d_B\ge1$,
\begin{equation}
 \DHHS\Bigl(
 \BEXT_{n,n}(d_A,d_B),\,
 \SEP(d_A,d_B)
 \Bigr)
 \le
 \min\left\{
 \frac{\sqrt{\min\{d_A,d_B\}-1}}{n},
 \frac{8\sqrt{2n-1}}{n}
 \right\}.
 \label{eq:HS-two-sided-bext-upper}
\end{equation}
\end{corollary}

\begin{proof}[Proof of Corollary~\ref{cor:HS-two-sided-bose-extension}]
Fix $\rho\in\BEXT_{n,n}(d_A,d_B)$.  By
Theorem~\ref{thm:two-sided-bose-extension}, there is a state
$\sigma\in\SEP(d_A,d_B)$ such that
\[
 \Dtr(\rho,\sigma)
 \le
 \frac{\sqrt{\min\{d_A,d_B\}-1}}{2n}.
\]
Therefore,
\[
 \DHS(\rho,\sigma)
 \le
 \norm{\rho-\sigma}_1
 =
 2\Dtr(\rho,\sigma)
 \le
 \frac{\sqrt{\min\{d_A,d_B\}-1}}{n}.
\]
This proves the dimension-dependent estimate.

For the dimension-free estimate, by duality
(Lemma~\ref{lem:compact-convex-distance-duality}),
\begin{equation}
 \DHS\bigl(\rho,\SEP(d_A,d_B)\bigr)
 =
 \max_{\substack{M=M^\dagger\\\norm{M}_2\le1}}
 \left\{\Tr(M\rho)-h_{\SEP}(M)\right\}.
 \label{eq:HS-separability-duality}
\end{equation}
The maximum may be restricted to traceless $M$: replacing $M$ by
$M-\frac{\Tr M}{d_Ad_B}I$ leaves the objective unchanged and can only
decrease its Hilbert--Schmidt norm.
For $A_M:=VMV^\dagger$, Fact~\ref{fact:labeled-operator-embedding} and
the tracelessness of $M$ give
\[
 h(A_M)=\frac12h_{\SEP}(M),
 \qquad
 \norm{A_M}_2=\norm{M}_2.
\]

Let $\widehat\rho_{2n}$ be the bosonic state supplied by
Lemma~\ref{lem:label-bundling-embedding}, and set
\[
 \varepsilon_{2n}
 :=
 \DHHS\Bigl(
  \BOS_{2n}^{(2)}(d_A+d_B),\,
  \SEPH^{(2)}(d_A+d_B)
 \Bigr).
\]
There is a Hartree mixture
$\tau\in\SEPH^{(2)}(d_A+d_B)$ such that
$\norm{\widehat\rho_{2n}^{(2)}-\tau}_2\le\varepsilon_{2n}$.
Since $\Tr(A_M\tau)\le h(A_M)$, Hilbert--Schmidt Cauchy--Schwarz and
Fact~\ref{fact:labeled-operator-embedding} give
\[
 \Tr\bigl(A_M\widehat\rho_{2n}^{(2)}\bigr)
 \le
 \frac12h_{\SEP}(M)+\varepsilon_{2n}\norm{M}_2.
\]
On the other hand, the label-bundling identity
\eqref{eq:cross-label-probability} gives
$\Tr(A_M\widehat\rho_{2n}^{(2)})=\frac{n}{2n-1}\Tr(M\rho)$.
Thus, for every traceless $M$ with $\norm{M}_2\le1$,
\begin{align*}
 \Tr(M\rho)-h_{\SEP}(M)
 &\le
 \frac{2n-1}{n}\varepsilon_{2n}\norm{M}_2
 -\frac1{2n}h_{\SEP}(M)\\
 &\le
 \frac{2n-1}{n}\varepsilon_{2n},
\end{align*}
where the last inequality uses $h_{\SEP}(M)\ge0$.  Taking the maximum
in \eqref{eq:HS-separability-duality} and then the supremum over $\rho$
gives, by \eqref{eq:Hausdorff-outer-relaxation},
\begin{equation*}
 \DHHS\Bigl(
  \BEXT_{n,n}(d_A,d_B),\,
  \SEP(d_A,d_B)
 \Bigr)
 \le
 \frac{2n-1}{n}\varepsilon_{2n}.
\end{equation*}
Theorem~\ref{thm:HS-main}, applied with $N=2n$, gives
$\varepsilon_{2n}\le8/\sqrt{2n-1}$ and therefore
\begin{equation*}
 \DHHS\Bigl(
  \BEXT_{n,n}(d_A,d_B),\,
  \SEP(d_A,d_B)
 \Bigr)
 \le
 \frac{8\sqrt{2n-1}}{n}.
\end{equation*}
Together with the dimension-dependent estimate above, this proves
\eqref{eq:HS-two-sided-bext-upper}.
\end{proof}

\section{Hilbert--Schmidt separability algorithms}
\label{sec:hs-applications}

The dimension-free estimate for two-sided Bose-symmetric extendible
states has two algorithmic consequences.  The two-sided extension
relaxation approximates $h_{\SEP}(M)$ for Hermitian objectives whose
traceless part has bounded Hilbert--Schmidt norm.  Through
Hilbert--Schmidt duality, the same approximation also gives
Hilbert--Schmidt separability testing.
At fixed accuracy, the required extension level is independent of the
local dimensions, so both programs have polynomial size in those
dimensions.

Let $\cH_A\cong\C^{d_A}$ and $\cH_B\cong\C^{d_B}$.  For $n\ge1$, write
\begin{equation}
 \delta_n^{\mathrm{HS}}(d_A,d_B)
 :=
 \min\left\{
  \frac{\sqrt{\min\{d_A,d_B\}-1}}{n},
  \frac{8\sqrt{2n-1}}{n}
 \right\}.
 \label{eq:HS-bext-error}
\end{equation}

\subsection{Separable optimization with bounded Hilbert--Schmidt norm}

For $M\in\Herm(\cH_A\ot\cH_B)$, set
\begin{equation}
 m:=\frac{\Tr M}{d_Ad_B},
 \qquad
 M_0:=M-mI.
 \label{eq:HS-BSS-centering}
\end{equation}
Recall from Subsection~\ref{sec:qma2-applications} that
\[
 h_{\BEXT_{n,n}}(M)
 :=
 \max_{\rho\in\BEXT_{n,n}(d_A,d_B)}\Tr(M\rho)
 =\lambda_{\max}\bigl(M^{[n,n]}\bigr).
\]
Adding a scalar multiple of the identity shifts the value of $M$ by
the same amount on every state.  Hence
\begin{equation}
\begin{aligned}
 h_{\SEP}(M)&=m+h_{\SEP}(M_0),\\
 h_{\BEXT_{n,n}}(M)&=m+h_{\BEXT_{n,n}}(M_0),\\
 \norm{M_0}_2&\le\norm{M}_2.
\end{aligned}
 \label{eq:HS-BSS-centering-properties}
\end{equation}
The last inequality holds because $M\mapsto M_0$ is the orthogonal
projection onto the traceless subspace in the Hilbert--Schmidt inner
product.

\begin{theorem}[Separable optimization with bounded Hilbert--Schmidt norm]
\label{thm:HS-BSS}
For every Hermitian $M$ and every $n\ge1$,
\begin{equation}
 h_{\SEP}(M)
 \le
 h_{\BEXT_{n,n}}(M)
 \le
 h_{\SEP}(M)
 +\delta_n^{\mathrm{HS}}(d_A,d_B)\norm{M_0}_2.
 \label{eq:HS-BSS-gap}
\end{equation}
Consequently, given an explicit rational description of $M$ and
$\epsilon\in(0,1]$, one can approximate $h_{\SEP}(M)$ to additive
error $\epsilon$ deterministically.  Up to a factor polynomial in
$\operatorname{bits}(M)$ and $\log(1/\epsilon)$, the running time is
\begin{equation}
 \exp\left(
 O\left(
 \left(
  1+
  \min\left\{
   \frac{\sqrt{\min\{d_A,d_B\}-1}\norm{M_0}_2}{\epsilon},
   \frac{\norm{M_0}_2^2}{\epsilon^2}
  \right\}
 \right)
 \log\left(d_A+d_B+\frac{\norm{M_0}_2}{\epsilon}+2\right)
 \right)
 \right).
 \label{eq:HS-BSS-runtime}
\end{equation}
In particular, for bounded $\norm{M_0}_2$ and fixed $\epsilon$, the
running time is polynomial in the local dimensions.
\end{theorem}

\begin{proof}
The lower bound in \eqref{eq:HS-BSS-gap} follows from
$\SEP(d_A,d_B)\subseteq\BEXT_{n,n}(d_A,d_B)$.  For the upper bound,
fix $\rho\in\BEXT_{n,n}(d_A,d_B)$.  Corollary~\ref{cor:HS-two-sided-bose-extension}
gives $\sigma\in\SEP(d_A,d_B)$ such that
\[
 \norm{\rho-\sigma}_2
 \le\delta_n^{\mathrm{HS}}(d_A,d_B).
\]
Since $M$ and $M_0$ differ by a scalar multiple of the identity,
Hilbert--Schmidt Cauchy--Schwarz gives
\begin{align*}
 \Tr(M\rho)-h_{\SEP}(M)
 &=\Tr(M_0\rho)-h_{\SEP}(M_0)\\
 &\le\Tr\bigl(M_0(\rho-\sigma)\bigr)\\
 &\le\delta_n^{\mathrm{HS}}(d_A,d_B)\norm{M_0}_2.
\end{align*}
Maximizing over $\rho$ proves the upper bound.

For the algorithm, take
\[
 n:=1+\left\lceil
 \min\left\{
  \frac{2\sqrt{\min\{d_A,d_B\}-1}\norm{M_0}_2}{\epsilon},
  \frac{512\norm{M_0}_2^2}{\epsilon^2}
 \right\}
 \right\rceil.
\]
Then the error term in \eqref{eq:HS-BSS-gap} is at most $\epsilon/2$.
Compute the largest eigenvalue of $M^{[n,n]}$ to additive precision
$\epsilon/2$.  The lift acts on a space of dimension
\[
 L_{d_A,d_B,n}
 =
 \binom{d_A+n-1}{n}
 \binom{d_B+n-1}{n}.
\]
The standard binomial estimate and
$n=O\left(
1+\min\left\{
 \frac{\sqrt{\min\{d_A,d_B\}-1}\norm{M_0}_2}{\epsilon},
 \frac{\norm{M_0}_2^2}{\epsilon^2}
\right\}
\right)$ give
\eqref{eq:HS-BSS-runtime}.
\end{proof}

\subsection{Hilbert--Schmidt separability testing}

Given a bipartite state $\rho$, Hilbert--Schmidt separability testing
asks whether $\rho$ is separable or is at Hilbert--Schmidt distance at
least $\epsilon$ from every separable state.  This is weak membership
for the separable set in Hilbert--Schmidt distance.  By duality
(Lemma~\ref{lem:compact-convex-distance-duality}),
\begin{equation}
 \DHS\bigl(\rho,\SEP(d_A,d_B)\bigr)
 =
 \max_{\substack{M\in\Herm(\cH_A\ot\cH_B)\\
                  \Tr M=0,\ \norm{M}_2\le1}}
 \left\{
  \Tr(M\rho)-h_{\SEP}(M)
 \right\}.
 \label{eq:HS-sep-distance-dual}
\end{equation}
We replace $h_{\SEP}$ in
\eqref{eq:HS-sep-distance-dual} by the relaxation from the preceding
subsection.

\begin{theorem}[Polynomial-time Hilbert--Schmidt separability testing]
\label{thm:HS-separability-testing}
\label{cor:HS-separability-ptas}
Given an explicit rational state
$\rho\in\D(\C^{d_A}\ot\C^{d_B})$ and $\epsilon\in(0,1]$, one can
approximate $\DHS(\rho,\SEP(d_A,d_B))$ to additive error
$\epsilon$ deterministically.  Up to a factor polynomial in
$\operatorname{bits}(\rho)$ and $\log(1/\epsilon)$, the running time is
\begin{equation}
 \exp\left(
 O\left(
 \left(
  1+
  \min\left\{
   \frac{\sqrt{\min\{d_A,d_B\}-1}}{\epsilon},
   \frac1{\epsilon^2}
  \right\}
 \right)
 \log(d_A+d_B+\epsilon^{-1})
 \right)
 \right).
 \label{eq:HS-separability-ptas-runtime}
\end{equation}
In the same running time one can distinguish
\begin{equation}
 \rho\in\SEP(d_A,d_B)
 \qquad\text{from}\qquad
 \DHS\bigl(\rho,\SEP(d_A,d_B)\bigr)\ge\epsilon.
 \label{eq:HS-separability-promise}
\end{equation}
Thus, for every fixed $\epsilon>0$, Hilbert--Schmidt separability
testing is polynomial-time solvable in the explicit input dimensions.
\end{theorem}

\begin{proof}
For $n\ge1$, define
\begin{equation}
 \widehat D_n^{\mathrm{HS}}(\rho)
 :=
 \max_{\substack{M\in\Herm(\cH_A\ot\cH_B)\\
                  \Tr M=0,\ \norm{M}_2\le1}}
 \left\{
  \Tr(M\rho)-h_{\BEXT_{n,n}}(M)
 \right\}.
 \label{eq:HS-dual-separability-relaxation}
\end{equation}
For every feasible $M$, Theorem~\ref{thm:HS-BSS} gives
\[
 h_{\SEP}(M)
 \le h_{\BEXT_{n,n}}(M)
 \le h_{\SEP}(M)+\delta_n^{\mathrm{HS}}(d_A,d_B).
\]
Subtracting these inequalities from $\Tr(M\rho)$, maximizing over $M$,
and applying \eqref{eq:HS-sep-distance-dual} gives
\begin{equation}
 \DHS\bigl(\rho,\SEP(d_A,d_B)\bigr)
 -\delta_n^{\mathrm{HS}}(d_A,d_B)
 \le
 \widehat D_n^{\mathrm{HS}}(\rho)
 \le
 \DHS\bigl(\rho,\SEP(d_A,d_B)\bigr).
 \label{eq:HS-bext-distance-sandwich}
\end{equation}

Since
$h_{\BEXT_{n,n}}(M)=\lambda_{\max}(M^{[n,n]})$,
the relaxation \eqref{eq:HS-dual-separability-relaxation} is the
semidefinite program
\begin{equation}
\begin{aligned}
 \textnormal{maximize}\quad &\Tr(M\rho)-t\\
 \textnormal{subject to}\quad
 &\Tr M=0,\\
 &\begin{pmatrix}
   I_{(d_Ad_B)^2} & \operatorname{vec}(M)\\
   \operatorname{vec}(M)^\dagger & 1
  \end{pmatrix}\succeq0,\\
 &M^{[n,n]}
 \preceq tI_{\Sym^n(\cH_A)\ot\Sym^n(\cH_B)},
\end{aligned}
 \label{eq:HS-separability-sdp}
\end{equation}
with variables $M\in\Herm(\cH_A\ot\cH_B)$ and $t\in\mathbb R$,
where $\operatorname{vec}(M)$ denotes the vectorization of $M$.
The first matrix inequality is equivalent to $\norm{M}_2\le1$ by the
Schur complement, while the second is equivalent to
$t\ge h_{\BEXT_{n,n}}(M)$.  Hence the SDP has value
$\widehat D_n^{\mathrm{HS}}(\rho)$.

For the approximation algorithm, take
\[
 n:=1+\left\lceil
 \min\left\{
  \frac{2\sqrt{\min\{d_A,d_B\}-1}}{\epsilon},
  \frac{512}{\epsilon^2}
 \right\}
 \right\rceil.
\]
Then $\delta_n^{\mathrm{HS}}(d_A,d_B)\le\epsilon/2$.
Solving \eqref{eq:HS-separability-sdp} to additive precision
$\epsilon/2$ therefore gives an additive-$\epsilon$ approximation by
\eqref{eq:HS-bext-distance-sandwich}.

The two semidefinite constraints have dimensions $(d_Ad_B)^2+1$ and
\[
 L_{d_A,d_B,n}
 =
 \binom{d_A+n-1}{n}
 \binom{d_B+n-1}{n},
\]
respectively.
Standard algorithms for semidefinite programs, together with the
binomial estimate, give \eqref{eq:HS-separability-ptas-runtime}.

For the promise problem, increase $n$ by a constant factor so that
$\delta_n^{\mathrm{HS}}(d_A,d_B)\le\epsilon/4$ and solve to precision
$\epsilon/8$.  The relaxation value is zero when $\rho$ is separable
and is at least
$3\epsilon/4$ when
$\DHS(\rho,\SEP(d_A,d_B))\ge\epsilon$.  Thus the numerical output is
at most $\epsilon/8$ in the separable case and at least $5\epsilon/8$
in the far case.  Comparing it with $\epsilon/2$ distinguishes the two
cases.
\end{proof}

For fixed additive accuracy, this improves the quasipolynomial
Hilbert--Schmidt weak-membership guarantee of Brand\~ao, Christandl,
and Yard~\cite{BCY11algorithm} to polynomial time in the local
dimensions.  These algorithms take explicit matrices as input and do
not give a copy-efficient test for an unknown quantum state.  Since
Hilbert--Schmidt distance is not contractive under general quantum
channels, they do not replace trace- or LOCC-norm separability testing.

\bigskip
\section*{Acknowledgment}
The authors used GPT-5.6 Pro for exploratory analysis, idea testing,
editorial polishing, and assistance with Lean formalization.
The authors reviewed and finalized all mathematical claims and proofs, and take full responsibility for the correctness, exposition, and attribution in the final manuscript.

\bibliographystyle{alpha}
\bibliography{reference}

\appendix

\section{Rectangular argmax rounding}
\label{sec:rectangular-argmax-appendix}

This appendix proves Proposition~\ref{prop:rectangular-argmax}.  The
eigenvector-contraction step is the same as in
Proposition~\ref{prop:tau-bound}; the remaining ingredient is an argmax
estimate on a product of two complex spheres.  This product-sphere
estimate is the complex pure-state form of the argmax argument
underlying our argmax-SoS analysis of Best Separable
State~\cite[Sec.~3]{JWX26}.

Let $\cH_A$ and $\cH_B$ be finite-dimensional Hilbert spaces, and use the
rectangular lift $M^{[n,m]}$ from
Subsubsection~\ref{sec:rectangular-argmax-main}.
We use throughout the injective and projective tensor norms defined in
Section~\ref{sec:tensor-norms}.

The vectors $x^{\ot n}\ot y^{\ot m}$ span
$\Sym^n(\cH_A)\ot\Sym^m(\cH_B)$.  Fix integers $n,m\ge1$ and a unit vector
$\psi$ in this space.  Choose unit vectors $x,y$ maximizing
$\abs{\braket{x^{\ot n}\ot y^{\ot m}}{\psi}}$.  After adjusting the
phase of $x$, set
\begin{equation}
 c
 :=
 \braket{x^{\ot n}\ot y^{\ot m}}{\psi}
 >0,
 \label{eq:rectangular-c-def}
\end{equation}
and contract all but one tensor factor in each block:
\begin{equation}
 \eta_{x,y}
 :=
 \Bigl[
  (\bra{x}^{\ot(n-1)}\ot I_{\cH_A})
  \ot
  (\bra{y}^{\ot(m-1)}\ot I_{\cH_B})
 \Bigr]\psi
 \in\cH_A\ot\cH_B.
 \label{eq:rectangular-contraction}
\end{equation}

\begin{lemma}[Product-sphere argmax estimate]
\label{lem:product-sphere-argmax}
With the notation above,
\begin{equation}
 \eta_{x,y}
 =
 c\,x\ot y+\zeta_{x,y},
 \qquad
 \zeta_{x,y}\in x^\perp\ot y^\perp,
 \label{eq:rectangular-contraction-decomp}
\end{equation}
and
\begin{equation}
 \norm{\zeta_{x,y}}_\varepsilon
 \le
 \frac{c}{\sqrt{nm}}.
 \label{eq:rectangular-injective-bound}
\end{equation}
\end{lemma}

\begin{proof}
For any unit vector $p\perp x$ and $z\in\C$, set
\[
 x_z:=\frac{x+zp}{\sqrt{1+\abs z^2}},
 \qquad
 g_p(z):=\braket{x_z^{\ot n}\ot y^{\ot m}}{\psi}.
\]
By the choice of $x,y$, the function $\abs{g_p(z)}^2$ is maximized at
$z=0$.  Its first-order expansion is
\[
 g_p(z)
 =c+n\overline z\,
 \braket{p\ot x^{\ot(n-1)}\ot y^{\ot m}}{\psi}
 +O(\abs z^2).
\]
The normalization factor has no linear term.  Varying the phase of
$z$ therefore gives
\[
 \braket{p\ot x^{\ot(n-1)}\ot y^{\ot m}}{\psi}=0.
\]
The analogous variation in the $y$-block gives, for every $q\perp y$,
\[
 \braket{x^{\ot n}\ot q\ot y^{\ot(m-1)}}{\psi}=0.
\]
Since these identities hold for every $p\perp x$ and $q\perp y$, and
$\braket{x\ot y}{\eta_{x,y}}=c$, they prove
\eqref{eq:rectangular-contraction-decomp}.

It remains to bound the component orthogonal to both $x$ and $y$.
If either $x^\perp$ or $y^\perp$ is zero-dimensional, then
$\zeta_{x,y}=0$ and the claim is immediate.  Otherwise, fix unit
vectors $p\perp x$ and $q\perp y$.  For $z,w\in\C$, set
\[
 x_z:=\frac{x+zp}{\sqrt{1+\abs z^2}},
 \qquad
 y_w:=\frac{y+wq}{\sqrt{1+\abs w^2}},
 \qquad
 g(z,w):=\braket{x_z^{\ot n}\ot y_w^{\ot m}}{\psi}.
\]
Again, $\abs{g(z,w)}^2$ is maximized at $(z,w)=(0,0)$.  Write
\[
 \beta
 :=
 \braket{p\ot x^{\ot(n-1)}\ot q\ot y^{\ot(m-1)}}{\psi}
 =\braket{p\ot q}{\zeta_{x,y}}.
\]
Expanding this overlap to second order gives
\begin{align}
 \abs{g(z,w)}^2
 &=c^2-c^2\bigl(n\abs z^2+m\abs w^2\bigr)
 \notag\\
 &\quad
 +2c\,\Re\left(
  \binom n2\overline z^2\alpha
  +\binom m2\overline w^2\gamma
  +nm\overline z\,\overline w\,\beta
 \right)
 +O\bigl((\abs z+\abs w)^3\bigr),
 \label{eq:rectangular-second-order}
\end{align}
When $n\ge2$, here
$\alpha:=\braket{p^{\ot2}\ot x^{\ot(n-2)}\ot y^{\ot m}}{\psi}$;
when $n=1$, the corresponding term is absent.  Define $\gamma$
analogously for the $y$-block.

Fix $a,b>0$ and choose $\theta$ so that
$e^{-2i\theta}\beta=\abs\beta$.  For each $\phi\in[0,2\pi]$, restrict
\eqref{eq:rectangular-second-order} to
\[
 z=ra e^{i(\theta+\phi)},
 \qquad
 w=rb e^{i(\theta-\phi)},
 \qquad r\in\mathbb R.
\]
Since $r=0$ is a local maximum, the coefficient of $r^2$ is
nonpositive.  Averaging this inequality over $\phi$ eliminates the
terms containing $\alpha$ and $\gamma$, while the term containing
$\beta$ is unchanged.  Hence
\[
 2cnmab\abs\beta
 \le
 c^2\bigl(na^2+mb^2\bigr).
\]
Taking $a=\sqrt m$ and $b=\sqrt n$ yields
\[
 \abs\beta\le\frac{c}{\sqrt{nm}}.
\]
Taking the supremum over unit $p\perp x$ and $q\perp y$ proves
\eqref{eq:rectangular-injective-bound}.
\end{proof}

\begin{proof}[Proof of Proposition~\ref{prop:rectangular-argmax}]
Let $\psi\in\Sym^n(\cH_A)\ot\Sym^m(\cH_B)$ be a unit top eigenvector of
$M^{[n,m]}$, with eigenvalue $\lambda$.  Choose $x,y,c,\eta_{x,y}$ as
in Lemma~\ref{lem:product-sphere-argmax}.  Because $\psi$ is invariant
under permutations of the $\cH_A$ factors and, independently, of the
$\cH_B$ factors,
\begin{align}
 \lambda c
 &=
 \braket{x^{\ot n}\ot y^{\ot m}}{M^{[n,m]}\psi}
 =
 \braket{M(x\ot y)}{\eta_{x,y}}.
 \label{eq:rectangular-eigen-contraction}
\end{align}
Decompose orthogonally
\begin{equation}
 M(x\ot y)
 =
 a_M(x,y)x\ot y+\xi_{x,y}+B_{x,y},
 \qquad
 \begin{aligned}
  \xi_{x,y}
  &\in
  (x^\perp\ot\C y)\oplus(\C x\ot y^\perp),\\
  B_{x,y}
  &\in x^\perp\ot y^\perp.
 \end{aligned}
 \label{eq:Axy-decomp}
\end{equation}
Combining \eqref{eq:rectangular-eigen-contraction},
\eqref{eq:rectangular-contraction-decomp}, and
\eqref{eq:Axy-decomp} gives
\begin{align}
 (\lambda-a_M(x,y))c
 &=
 \braket{M(x\ot y)}{\eta_{x,y}}-a_M(x,y)c
 \notag\\
 &=
 \braket{a_M(x,y)x\ot y+\xi_{x,y}+B_{x,y}}
         {c\,x\ot y+\zeta_{x,y}}
 -a_M(x,y)c
 \notag\\
 &=
 \braket{B_{x,y}}{\zeta_{x,y}}.
 \label{eq:rectangular-key-pairing}
\end{align}
By Lemma~\ref{lem:product-sphere-argmax} and
\eqref{eq:tensor-norm-duality-rank},
\[
 \abs{\lambda-a_M(x,y)}
 \le
 \frac1{\sqrt{nm}}\norm{B_{x,y}}_\pi.
\]
Since
$B_{x,y}=(P_{x^\perp}\ot P_{y^\perp})M(x\ot y)$, this proves
\eqref{eq:rectangular-argmax}.
\end{proof}

\section{Lower bounds for the de Finetti theorems}
\label{sec:lower-bounds}

\subsection{Optimality of the \texorpdfstring{$t$}{t}-site bounds}
\label{sec:higher-marginal-lower}

The two-site rectangular Werner examples of Christandl, K\"onig,
Mitchison, and Renner, together with their symmetric purifications,
already show that the dimension dependence of our bounds is
optimal~\cite[Lemma~II.5 and Corollary~III.9]{CKMR07}.  By
contractivity under partial trace, the same conclusion holds for every
fixed marginal order $t$.  We use their rectangular families below
only to prove the new point: the linear dependence on $t$ is also
necessary.

Fix $N=ar$.  Let $\omega_{ar,(a^r)}$ be the rectangular Werner
state: the normalized projector onto the $(a^r)$ Schur--Weyl block of
$(\C^r)^{\ot ar}$.  Let $\psi_{a,r}$ be its canonical symmetric
purification in
$\Sym^{ar}(\C^r\ot\C^r)$, equivalently the unit vector in the
one-dimensional Cauchy summand
\[
 S_{(a^r)}(\C^r)\ot S_{(a^r)}(\C^r).
\]
We write $\gamma_{a,r}^{(t)}$ and
$\omega_{ar,(a^r)}^{(t)}$ for their respective $t$-site marginals.
These are the rectangular families analyzed in~\cite{CKMR07}; the
calculation below extracts their dependence on $t$.

\subsubsection{Exact Schur-label distances}

The same rectangular families also detect the dependence on the
marginal order.  The cleanest formulation uses the Schur label of the
$t$-site marginal.  For a partition $\mu\vdash t$, let $f^\mu$ be the
dimension of the Specht module $V_\mu$, let $\mu'$ be the conjugate
partition, and write
\[
 c(\square)=j-i,
 \qquad
 h_\mu(\square)=\mu_i-j+\mu_j'-i+1
\]
for the content and hook length of a box
$\square=(i,j)\in\mu$.  We use
\[
 (x)_t=x(x-1)\cdots(x-t+1)
\]
for the falling factorial.

For $\mu\vdash t$ with $\ell(\mu)\le r$, define
\begin{align}
 p_{a,r,t}(\mu)
 &:=
 \frac{(f^\mu)^2}{t!}
 \frac{
  \displaystyle\prod_{\square\in\mu}
  (a-c(\square))(r+c(\square))
 }{(ar)_t},
 \label{eq:rectangular-down-weight}\\
 q_{r,t}(\mu)
 &:=
 \frac{(f^\mu)^2}{t!}
 \frac{
  \displaystyle\prod_{\square\in\mu}(r+c(\square))
 }{r^t},
 \label{eq:uniform-SW-weight}\\
 R_{a,r,t}(\mu)
 &:=
 \frac{p_{a,r,t}(\mu)}{q_{r,t}(\mu)}
 =
 \frac{
  \displaystyle\prod_{\square\in\mu}
  \left(1-\frac{c(\square)}a\right)
 }{
  \displaystyle\prod_{j=0}^{t-1}
  \left(1-\frac{j}{ar}\right)
 }.
 \label{eq:rectangular-likelihood-ratio}
\end{align}
A partition with $\mu_1>a$ has zero weight in
\eqref{eq:rectangular-down-weight}, since its first row contains a box
of content $a$.

\begin{lemma}[Rectangular branching weights]
\label{lem:rectangular-branching-weights}
Under the Cauchy decomposition
\[
 \Sym^t(\C^r\ot\C^r)
 =
 \bigoplus_{\substack{\mu\vdash t\\\ell(\mu)\le r}}
 S_\mu(\C^r)\ot S_\mu(\C^r),
\]
the block weights of $\gamma_{a,r}^{(t)}$ are
$p_{a,r,t}(\mu)$.  Under the Schur--Weyl decomposition of
$(\C^r)^{\ot t}$, the block weights of
$\omega_{ar,(a^r)}^{(t)}$ are the same numbers.  Moreover,
\begin{equation}
 p_{a,r,t}(\mu)
 =
 \frac{
  \dim S_\mu(\C^r)\,\dim S_{\mu'}(\C^a)
 }{\binom{ar}{t}}.
 \label{eq:rectangular-down-dimension}
\end{equation}
In particular, the numbers in
\eqref{eq:rectangular-down-weight} form a probability distribution.
\end{lemma}

\begin{proof}
Both restrictions are governed by the coherent down-transition measure
from the rectangular Specht module $V_{(a^r)}$:
\begin{equation}
 p_{a,r,t}(\mu)
 =
 \frac{f^\mu f^{(a^r)/\mu}}{f^{(a^r)}}.
 \label{eq:coherent-down-transition}
\end{equation}
For the Werner state, this follows by restricting the normalized
identity on $V_{(a^r)}$.

We give the corresponding argument for $\psi_{a,r}$.
Set $\lambda=(a^r)$ and choose a real orthonormal model of the Specht
module $V_\lambda$.  In the Schur--Weyl realization of the Cauchy
summand, the Specht factor of $\psi_{a,r}$ is the
normalized diagonal invariant
\[
 \Omega_\lambda
 :=
 \frac1{\sqrt{f^\lambda}}
 \sum_{j=1}^{f^\lambda}e_j\ot e_j
 \in V_\lambda\ot V_\lambda.
\]
Upon restriction to $S_t$,
\[
 V_\lambda\downarrow_{S_t}
 \cong
 \bigoplus_{\substack{\mu\vdash t\\\mu\subseteq\lambda}}
 V_\mu\ot M_{\lambda/\mu},
 \qquad
 \dim M_{\lambda/\mu}=f^{\lambda/\mu}.
\]
Choose the basis $(e_j)$ adapted to this orthogonal decomposition and
let $Q_\mu$ be the projection onto
$V_\mu\ot M_{\lambda/\mu}$.  Then
\[
 \norm{(Q_\mu\ot Q_\mu)\Omega_\lambda}^2
 =
 \frac{\rank Q_\mu}{f^\lambda}
 =
 \frac{f^\mu f^{\lambda/\mu}}{f^\lambda}.
\]
Because $\Omega_\lambda$ is invariant under the diagonal action
$g\ot g$ of $\mathfrak S_t$ and $Q_\mu$ commutes with that action,
$(Q_\mu\ot Q_\mu)\Omega_\lambda$ is itself diagonal invariant.  On
the $\mu\ot\mu$ isotypic component, the Cauchy-block projector for
the first $t$ paired factors is the projection onto
$(V_\mu\ot V_\mu)^{\operatorname{diag}(\mathfrak S_t)}$, tensored
with the identity on the multiplicity spaces.  It therefore acts as
the identity on $(Q_\mu\ot Q_\mu)\Omega_\lambda$, while the
projectors with other labels annihilate this component.  Consequently,
its expectation in the global vector, or equivalently in the
$t$-site marginal, is the displayed squared norm.  This proves
\eqref{eq:coherent-down-transition} for $\gamma_{a,r}^{(t)}$ as well.

The shifted-Schur skew-dimension formula, specialized to a rectangle,
gives~\cite[Theorem~8.1]{OO97}
\[
 \frac{f^{(a^r)/\mu}}{f^{(a^r)}}
 =
 \frac1{(ar)_t}
 \prod_{\square\in\mu}
 \frac{
  (a-c(\square))(r+c(\square))
 }{h_\mu(\square)}.
\]
Using the hook-length identity
$f^\mu=t!/\prod_{\square\in\mu}h_\mu(\square)$ gives
\eqref{eq:rectangular-down-weight}.  Applying the hook-content formula
to $S_\mu(\C^r)$ and $S_{\mu'}(\C^a)$ gives
\eqref{eq:rectangular-down-dimension}.  Summing the latter identity
over $\mu$ is the dual Cauchy decomposition of
$\bigwedge^t(\C^a\ot\C^r)$, whose dimension is
$\binom{ar}{t}$.
\end{proof}

For a probability vector $x=(x_1,\ldots,x_r)$, write
\begin{equation}
 \operatorname{SW}_{t,x}(\mu):=f^\mu s_\mu(x).
 \label{eq:SW-distribution}
\end{equation}
This is the Schur--Weyl distribution of a length-$t$ random word with
letter law $x$.  If $u\in\C^r\ot\C^r$ has squared Schmidt coefficients
$x$, then twirling $\ketbra{u}^{\ot t}$ under
$U(r)\times U(r)$ gives Cauchy-block weights
$\operatorname{SW}_{t,x}$.  Likewise, if a density operator $\sigma$
has spectrum $x$, then twirling $\sigma^{\ot t}$ under $U(r)$ gives
Schur--Weyl block weights $\operatorname{SW}_{t,x}$.  For the uniform
vector $\bar x=(1/r,\ldots,1/r)$, the hook-content formula gives
\[
 \operatorname{SW}_{t,\bar x}(\mu)=q_{r,t}(\mu).
\]

\begin{proposition}[Exact higher-marginal distance]
\label{prop:exact-higher-rectangular-distance}
For every $a,r\in\mathbb N$, $r\ge2$, and $2\le t\le ar$,
\begin{align}
 \Dtr\bigl(\gamma_{a,r}^{(t)},\SEPH^{(t)}(r^2)\bigr)
 &=
 \frac12
 \sum_{\substack{\mu\vdash t\\\ell(\mu)\le r}}
 q_{r,t}(\mu)
 \abs{R_{a,r,t}(\mu)-1},
 \label{eq:exact-higher-bos-distance}\\
 \Dtr\bigl(\omega_{ar,(a^r)}^{(t)},\SEPP^{(t)}(r)\bigr)
 &=
 \frac12
 \sum_{\substack{\mu\vdash t\\\ell(\mu)\le r}}
 q_{r,t}(\mu)
 \abs{R_{a,r,t}(\mu)-1}.
 \label{eq:exact-higher-ex-distance}
\end{align}
A closest invariant Hartree mixture in the first line is the twirl of
$\ketbra{\Phi_r}^{\ot t}$, where
$\Phi_r=r^{-1/2}\sum_i e_i\ot e_i$; a closest invariant state of the
form $\tau^{\ot t}$ in the second line is $(I/r)^{\ot t}$.  An optimal
separating operator $0\preceq A\preceq I$ is
the sum of the Schur projectors indexed by
\begin{equation}
 \mathcal L_{a,r,t}
 :=
 \{\mu:R_{a,r,t}(\mu)\ge1\}.
 \label{eq:NP-lower-ideal}
\end{equation}
\end{proposition}

\begin{proof}
Twirling is trace-norm contractive, fixes the two target states, and
preserves the respective approximation sets.  On the Cauchy side, the
$U(r)\times U(r)$ decomposition is multiplicity free.  On the Werner
side, the target and every twirled approximant commute with both the
$U(r)$ and $S_t$ actions and are therefore scalar on each joint
Schur--Weyl block $S_\mu(\C^r)\ot V_\mu$.  Consequently, trace distance
after twirling is the total-variation distance between the normalized
block weights.  It remains to show that $q_{r,t}$ is closest among
convex combinations of $\operatorname{SW}_{t,x}$.

Order partitions by dominance, with $\nu\unrhd\mu$ meaning that $\nu$
is more dominant.  The likelihood ratio in
\eqref{eq:rectangular-likelihood-ratio} is decreasing in this order.
Indeed, a covering move transfers one box from a lower row to a higher
row and strictly increases its content; the corresponding factor
$a-c(\square)$ decreases, while all other factors are unchanged.  If
the move leaves the $a\times r$ rectangle, the new likelihood ratio is
zero.  Consequently, $\mathcal L_{a,r,t}$ is a lower ideal in
dominance order.

Every probability vector $x$ majorizes $\bar x$.  The Schur--Weyl
majorization coupling states that if $x$ majorizes $y$, one may couple
$\mu_x\sim\operatorname{SW}_{t,x}$ and
$\mu_y\sim\operatorname{SW}_{t,y}$ so that
$\mu_x\unrhd\mu_y$ almost surely~\cite[Thm.~1.11]{OW16}.  Therefore
every lower ideal $\mathcal L$ satisfies
\begin{equation}
 \operatorname{SW}_{t,x}(\mathcal L)
 \le
 q_{r,t}(\mathcal L).
 \label{eq:lower-ideal-max-uniform}
\end{equation}
The same inequality holds for every convex combination of Schur--Weyl
distributions.

Let $p=p_{a,r,t}$ and $q=q_{r,t}$.  The set
$\mathcal L_{a,r,t}=\{p\ge q\}$ is the Neyman--Pearson set for $p$
versus $q$.  For every feasible invariant approximant with block
distribution $s$,
\begin{align*}
 \DTV(p,s)
 &\ge
 p(\mathcal L_{a,r,t})-s(\mathcal L_{a,r,t})\\
 &\ge
 p(\mathcal L_{a,r,t})-q(\mathcal L_{a,r,t})
 =
 \DTV(p,q).
\end{align*}
Since $q$ itself is feasible in both problems, equality follows, as
does the optimal separating operator.
\end{proof}

The exact formula reduces optimality to a scalar Schur-label
calculation.  We next extract a lower bound showing that the
linear dependence on $t$ in Theorem~\ref{thm:higher-main} is
necessary.  Let
\begin{equation}
 \mathrm{ct}(\mu)
 :=
 \sum_{\square\in\mu}c(\square),
 \qquad
 M_t:=\binom t2.
 \label{eq:content-statistic}
\end{equation}
The transposition class sum
$\mathsf T_t=\sum_{1\le i<j\le t}(ij)$ acts on $V_\mu$ as the scalar
$\mathrm{ct}(\mu)$.

\subsubsection{Optimal dependence on the marginal order}

\begin{lemma}[Plancherel content fluctuations]
\label{lem:Plancherel-content-fluctuations}
Let $t\ge2$, and let $\mu$ have Plancherel law
$\operatorname{Pl}_t(\mu)=(f^\mu)^2/t!$.  There is a universal
constant $c_0>0$ such that
\begin{equation}
 c_0t
 \le
 \mathbb E_{\operatorname{Pl}_t}
 \abs{\mathrm{ct}(\mu)}
 \le
 \sqrt{M_t}.
 \label{eq:Plancherel-content-L1}
\end{equation}
\end{lemma}

\begin{proof}
Use the normalized trace $\tau_{\mathrm{reg}}$ on the regular
representation of $\mathfrak S_t$, and write
\[
 \norm{X}_{p,\mathrm{reg}}
 :=
 \tau_{\mathrm{reg}}(\abs{X}^p)^{1/p}.
\]
Distinct group elements are orthonormal in the associated $L_2$ inner
product, so
\[
 \mathbb E\,\mathrm{ct}(\mu)=0,
 \qquad
 \mathbb E\,\mathrm{ct}(\mu)^2=M_t.
\]
Moreover, $\mathbb E\,\mathrm{ct}(\mu)^4$ is the number of ordered
quadruples of transpositions whose product is the identity.  Associate
to such a quadruple the four-edge multigraph formed by its
transpositions.  A vertex of degree one cannot occur in a word whose
product is the identity.  Hence every nonisolated vertex has degree at
least two, so the graph uses at most four vertices.  There are only
$O(t^4)$ such ordered quadruples, and therefore
\[
 \mathbb E\,\mathrm{ct}(\mu)^4\le C_0t^4
\]
for a universal $C_0$.  H\"older's inequality gives
\[
 \mathbb E\abs{\mathrm{ct}(\mu)}
 \ge
 \frac{
  \bigl(\mathbb E\,\mathrm{ct}(\mu)^2\bigr)^{3/2}
 }{
  \bigl(\mathbb E\,\mathrm{ct}(\mu)^4\bigr)^{1/2}
 }
 \ge
 c_0t.
\]
The upper bound is Cauchy--Schwarz.
\end{proof}

\begin{lemma}[Perturbative separation of the rectangular law]
\label{lem:rectangular-perturbative-lower}
There are universal constants $c_1,C_1>0$ such that, for every
integer $t\ge2$, whenever
\begin{equation}
 a\ge C_1t,
 \qquad
 r\ge C_1t^2,
 \label{eq:perturbative-parameter-range}
\end{equation}
one has
\begin{equation}
 \frac12
 \sum_\mu
 q_{r,t}(\mu)
 \abs{R_{a,r,t}(\mu)-1}
 \ge
 c_1\frac ta.
 \label{eq:perturbative-Schur-lower}
\end{equation}
\end{lemma}

\begin{proof}
We take $C_1\ge2$ throughout.  In particular, $a\ge t$ and $r\ge t$,
so every partition of $t$ fits inside the $a\times r$ rectangle and
all Schur labels appearing under Plancherel measure lie in the support
of $q_{r,t}$.
The distribution $q_{r,t}$ is the RSK shape of a uniform random word
in $[r]^t$.  With probability
\[
 \alpha_{r,t}
 :=
 \frac{(r)_t}{r^t}
 \ge
 1-\frac{M_t}{r},
\]
all letters are distinct.  Conditional on this event, their relative
order is a uniform permutation, so the shape has Plancherel law.  Thus
\begin{equation}
 q_{r,t}
 =
 \alpha_{r,t}\operatorname{Pl}_t
 +(1-\alpha_{r,t})\nu_{r,t}
 \label{eq:SW-Plancherel-mixture}
\end{equation}
for some probability distribution $\nu_{r,t}$.  Under
\eqref{eq:perturbative-parameter-range},
$\alpha_{r,t}\ge1/2$.

Let $J_k=\sum_{i<k}(ik)$ be the Jucys--Murphy elements and
$\mathsf T_t=\sum_kJ_k$.  The Jucys identity gives
\begin{equation}
 \prod_{k=1}^t
 \left(1-\frac{J_k}{a}\right)
 =
 \sum_{\pi\in\mathfrak S_t}
 \operatorname{sgn}(\pi)
 a^{-\abs{\pi}}\pi,
 \qquad
 \abs{\pi}
 :=
 t-\#\operatorname{cycles}(\pi).
 \label{eq:Jucys-identity-normalized}
\end{equation}
On $V_\mu$, the left-hand side has eigenvalue
$\prod_{\square\in\mu}(1-c(\square)/a)$.  Put
\[
 B_{a,r,t}
 :=
 \prod_{j=0}^{t-1}
 \left(1-\frac{j}{ar}\right).
\]
Thus the central element
\[
 \mathsf R_{a,r,t}
 :=
 B_{a,r,t}^{-1}
 \prod_{k=1}^t
 \left(1-\frac{J_k}{a}\right)
\]
has eigenvalue $R_{a,r,t}(\mu)$ on $V_\mu$.

Write
\[
 \prod_{k=1}^t
 \left(1-\frac{J_k}{a}\right)
 =
 I-\frac{\mathsf T_t}{a}+\mathsf H_a,
 \qquad
 \mathsf H_a
 :=
 \sum_{\abs{\pi}\ge2}
 \operatorname{sgn}(\pi)
 a^{-\abs{\pi}}\pi.
\]
If $n_k$ is the number of permutations of transposition length $k$,
then $n_k\le M_t^k$, because every such permutation has a
factorization into $k$ transpositions.  Orthogonality of the group
basis in regular $L_2$ therefore gives, provided $a^2\ge2M_t$,
\begin{equation}
 \norm{\mathsf H_a}_{2,\mathrm{reg}}^2
 =
 \sum_{k\ge2}\frac{n_k}{a^{2k}}
 \le
 2\frac{M_t^2}{a^4}.
 \label{eq:Jucys-remainder-L2}
\end{equation}
Also,
\[
 1-B_{a,r,t}
 \le
 \frac{M_t}{ar},
 \qquad
 B_{a,r,t}^{-1}\le2
\]
when the constants in
\eqref{eq:perturbative-parameter-range} are sufficiently large.  Since
\[
 \norm{\mathsf T_t}_{1,\mathrm{reg}}
 \le
 \norm{\mathsf T_t}_{2,\mathrm{reg}}
 =
 \sqrt{M_t},
\]
\eqref{eq:Jucys-remainder-L2} implies
\begin{align}
 \norm{
  \mathsf R_{a,r,t}-I+\mathsf T_t/a
 }_{1,\mathrm{reg}}
 &\le
 \frac{2M_t}{ar}
 \left(1+\frac{\sqrt{M_t}}a\right)
 +\frac{2\sqrt2M_t}{a^2}.
 \label{eq:R-linearization-L1}
\end{align}
Choose $C_1$ so that the right-hand side is at most
$(c_0/2)t/a$, where $c_0$ is from
Lemma~\ref{lem:Plancherel-content-fluctuations}.  Then
\begin{align*}
 \mathbb E_{\operatorname{Pl}_t}
 \abs{R_{a,r,t}(\mu)-1}
 &=
 \norm{\mathsf R_{a,r,t}-I}_{1,\mathrm{reg}}\\
 &\ge
 \frac1a
 \mathbb E_{\operatorname{Pl}_t}
 \abs{\mathrm{ct}(\mu)}
 -
 \norm{
  \mathsf R_{a,r,t}-I+\mathsf T_t/a
 }_{1,\mathrm{reg}}\\
 &\ge
 \frac{c_0t}{2a}.
\end{align*}
Combining this with \eqref{eq:SW-Plancherel-mixture} and
$\alpha_{r,t}\ge1/2$ proves
\eqref{eq:perturbative-Schur-lower}.
\end{proof}

\begin{theorem}[Optimal $t$-dependence of the rectangular examples]
\label{thm:rectangular-higher-lower}
There is a universal constant $C>0$ such that, for every
$a,r\in\mathbb N$ and $2\le t\le ar$ satisfying $r\ge Ct^2$,
\begin{align}
 \Dtr\bigl(\gamma_{a,r}^{(t)},\SEPH^{(t)}(r^2)\bigr)
 &=
 \Omega\left(\min\left\{1,\frac ta\right\}\right)
 =
 \Omega\left(\min\left\{1,\frac{t\sqrt d}{N}\right\}\right),
 \qquad
 d=r^2,
 \label{eq:matching-higher-bos-lower}\\
 \Dtr\bigl(\omega_{ar,(a^r)}^{(t)},\SEPP^{(t)}(r)\bigr)
 &=
 \Omega\left(\min\left\{1,\frac ta\right\}\right)
 =
 \Omega\left(\min\left\{1,\frac{td}{N}\right\}\right),
 \qquad
 d=r.
 \label{eq:matching-higher-ex-lower}
\end{align}
Together with Theorem~\ref{thm:higher-main}, these estimates match up
to universal constants.  In particular, the linear dependence on the
marginal order in \eqref{eq:higher-bos-main} and
\eqref{eq:higher-ex-main} cannot be improved in general.
\end{theorem}

\begin{proof}
By Proposition~\ref{prop:exact-higher-rectangular-distance}, the two
distances are the same scalar quantity.  If $a\ge C_1t$, apply
Lemma~\ref{lem:rectangular-perturbative-lower}.

Suppose $a<C_1t$.  Distance to the relevant de Finetti set cannot
increase under partial trace, so the $t$-site distance is at least the
$s$-site distance for every $2\le s\le t$.  If $a$ is larger than a
sufficiently large universal constant, take
\[
 s
 =
 \left\lfloor\frac{a}{2C_1}\right\rfloor.
\]
Then $2\le s\le t$, $a\ge C_1s$, and
$r\ge Ct^2\ge C_1s^2$ after increasing $C$.  The perturbative lemma
gives a universal constant lower bound because
$s/a\ge1/(4C_1)$.  If $a$ is bounded, use the two-site marginal
instead.  Proposition~\ref{prop:exact-higher-rectangular-distance},
specialized to two sites, gives for both rectangular families
\[
 D_t
 \ge
 D_2
 =
 \frac{r^2-1}{2r(ar-1)}
 \ge
 \frac{3}{8a},
\]
which is again bounded below by a universal constant.  This proves the
lower bounds in the regime $t/a\gtrsim1$ and hence the displayed capped
estimates.  The matching upper bounds follow from
Theorem~\ref{thm:higher-main}, together with the trace-distance
diameter bound, after substituting $N=ar$ and, respectively,
$d=r^2$ or $d=r$.
\end{proof}

\subsection{Dimension-free bosonic Hilbert--Schmidt lower bounds}
\label{sec:HS-lower-bound}

\begin{theorem}[Dimension-free bosonic Hilbert--Schmidt lower bound]
\label{thm:HS-lower}
As $N\to\infty$,
\[
 \sup_{d\ge2}
 \DHHS\Bigl(
  \BOS_N^{(2)}(d),\,
  \SEPH^{(2)}(d)
 \Bigr)
 =\Omega(N^{-1/2}).
\]
\end{theorem}

The lower bound is included for completeness.  It adapts the
Jiang--Tacla--Caves bosonic pair construction, based on symmetrized powers of a
two-site vector~\cite{JTC17}.  We give the short calculation in
symmetric-tensor notation.

\begin{proof}[Proof of Theorem~\ref{thm:HS-lower}]
Fix $m\ge1$, let $\cH_m=\C^{m+1}$ with orthonormal basis
$e_0,e_1,\ldots,e_m$, and define
\begin{equation}
 b_m:=\frac1{\sqrt m}\sum_{i=1}^m e_i^{\ot2},
 \qquad
 A_m:=\ket{e_0^{\ot2}}\bra{b_m}+\ket{b_m}\bra{e_0^{\ot2}}.
 \label{eq:HS-lower-witness}
\end{equation}
The vectors $e_0^{\ot2}$ and $b_m$ are orthonormal, so
\begin{equation}
 \norm{A_m}_2=\sqrt2.
 \label{eq:HS-Am-HS-norm}
\end{equation}
If $u=\alpha e_0+v$ with $v\perp e_0$, then
\[
 \abs{\bra{u^{\ot2}}A_m\ket{u^{\ot2}}}
 \le\frac{2}{\sqrt m}\abs{\alpha}^2\norm{v}^2
 \le\frac1{2\sqrt m}.
\]
Equality is attained at $u=(e_0+e_1)/\sqrt2$, and therefore
\begin{equation}
 h(A_m)=\frac1{2\sqrt m}.
 \label{eq:HS-Am-Hartree}
\end{equation}
Thus the Hartree value vanishes as $m\to\infty$.

For $q\ge0$, set
\begin{equation}
 x_q^{(m)}:=\Pi_{\mathrm{sym}}^{(2q)}b_m^{\ot q},
 \qquad
 \chi_q^{(m)}:=\frac{x_q^{(m)}}{\norm{x_q^{(m)}}},
 \label{eq:HS-pair-correlated-state}
\end{equation}
with $x_0^{(m)}=\chi_0^{(m)}=1$.  Expanding in the orthonormal
monomial basis of $\Sym^{2q}\C^m$ gives
\begin{equation}
 \norm{x_q^{(m)}}^2
 =\prod_{j=0}^{q-1}\frac{m+2j}{m(2j+1)}.
 \label{eq:HS-pair-correlated-norm}
\end{equation}
Let $C_{b_m}\xi:=(\bra{b_m}\ot I)\xi$ denote the partial inner product
with $b_m$ on two tensor factors.  Both $b_m$ and
$\chi_{q+1}^{(m)}$ are invariant under the natural $O(m)$ action, and
$C_{b_m}$ is $O(m)$-equivariant.  Hence
$C_{b_m}\chi_{q+1}^{(m)}$ lies in
$(\Sym^{2q}\C^m)^{O(m)}$.  This invariant subspace is one-dimensional,
spanned by $\chi_q^{(m)}$; equivalently, every $O(m)$-invariant
homogeneous polynomial of degree $2q$ is a scalar multiple of
$(z_1^2+\cdots+z_m^2)^q$.  Thus
\[
 C_{b_m}\chi_{q+1}^{(m)}=\kappa_q\chi_q^{(m)}.
\]
On symmetric tensors, the adjoint of $C_{b_m}$ is
$\xi\mapsto\Pi_{\mathrm{sym}}^{(2q+2)}(b_m\ot\xi)$.  Taking the inner
product with $\chi_{q+1}^{(m)}$ and using
\eqref{eq:HS-pair-correlated-norm} gives
\begin{equation}
 \kappa_q
 =\frac{\norm{x_{q+1}^{(m)}}}{\norm{x_q^{(m)}}}
 =\sqrt{\frac{m+2q}{m(2q+1)}},
 \qquad
 C_{b_m}\chi_{q+1}^{(m)}
 =\sqrt{\frac{m+2q}{m(2q+1)}}\,\chi_q^{(m)}.
 \label{eq:HS-pair-contraction}
\end{equation}

For $0\le q\le\lfloor N/2\rfloor$, define the normalized symmetric
tensor
\begin{equation}
 \psi_q^{(m,N)}
 :=
 \sqrt{\frac{N!}{(N-2q)!(2q)!}}\,
 \Pi_{\mathrm{sym}}^{(N)}\left(
 e_0^{\ot(N-2q)}\ot\chi_q^{(m)}
 \right)
 \in\Sym^N\cH_m.
 \label{eq:HS-pair-correlated-vectors}
\end{equation}
These vectors are mutually orthogonal.  Fix
$0\le q\le\lfloor(N-2)/2\rfloor$ and put $r=N-2q$.  After
contracting either $e_0^{\ot2}$ from $\psi_q^{(m,N)}$ or $b_m$ from
$\psi_{q+1}^{(m,N)}$, the normalized remaining tensor is the same;
denote it by $\theta_q^{(m,N)}$.  Directly from
\eqref{eq:HS-pair-correlated-vectors} and
\eqref{eq:HS-pair-contraction},
\begin{align}
 (\bra{e_0^{\ot2}}\ot I)\psi_q^{(m,N)}
 &=
 \sqrt{\frac{r(r-1)}{N(N-1)}}\,
 \theta_q^{(m,N)},
 \label{eq:HS-e0-contraction}\\
 (\bra{b_m}\ot I)\psi_{q+1}^{(m,N)}
 &=
 \frac{2\sqrt{(q+1)(q+m/2)}}{\sqrt{mN(N-1)}}\,
 \theta_q^{(m,N)}.
 \label{eq:HS-bm-contraction}
\end{align}
Since the vectors in \eqref{eq:HS-pair-correlated-vectors} are
symmetric, every summand in the normalized lift has the same
matrix element.  Equations \eqref{eq:HS-e0-contraction} and
\eqref{eq:HS-bm-contraction} therefore give
\begin{equation}
 \bra{\psi_{q+1}^{(m,N)}}A_m^{[N]}\ket{\psi_q^{(m,N)}}
 =t_{q,m,N},
 \label{eq:HS-t-q-def}
\end{equation}
where
\begin{equation}
 t_{q,m,N}
 :=\frac{2}{N(N-1)\sqrt m}
 \sqrt{(q+1)(q+m/2)}
 \sqrt{(N-2q)(N-2q-1)}.
 \label{eq:HS-t-q-formula}
\end{equation}
The diagonal matrix elements vanish.  Hence the bosonic state
\[
 \rho_{q,m,N}
 :=
 \frac12
 \ket{\psi_q^{(m,N)}+\psi_{q+1}^{(m,N)}}
 \bra{\psi_q^{(m,N)}+\psi_{q+1}^{(m,N)}}
\]
satisfies
$\Tr(A_m\rho_{q,m,N}^{(2)})=t_{q,m,N}$.

Now let $\tau\in\SEPH^{(2)}(m+1)$.  Hilbert--Schmidt
Cauchy--Schwarz, \eqref{eq:HS-Am-HS-norm}, and
$\Tr(A_m\tau)\le h(A_m)$ give
\[
 \norm{\rho_{q,m,N}^{(2)}-\tau}_2
 \ge
 \frac{\Tr(A_m\rho_{q,m,N}^{(2)})-\Tr(A_m\tau)}{\sqrt2}
 \ge
 \frac{t_{q,m,N}}{\sqrt2}-\frac1{2\sqrt{2m}}.
\]
Taking the infimum over $\tau$ proves
\begin{equation}
 \DHHS\Bigl(
  \BOS_N^{(2)}(m+1),\,
  \SEPH^{(2)}(m+1)
 \Bigr)
 \ge\frac{t_{q,m,N}}{\sqrt2}-\frac1{2\sqrt{2m}}.
 \label{eq:HS-finite-m-lower}
\end{equation}
For fixed $N$ and $q$,
\[
 \lim_{m\to\infty}\frac{t_{q,m,N}}{\sqrt2}
 =\frac{\sqrt{(q+1)(N-2q)(N-2q-1)}}{N(N-1)}.
\]
Since the second term in \eqref{eq:HS-finite-m-lower} tends to
zero, taking the supremum over $d=m+1$ gives, for every admissible $q$,
\[
 \sup_{d\ge2}
 \DHHS\Bigl(
  \BOS_N^{(2)}(d),
  \SEPH^{(2)}(d)
 \Bigr)
 \ge
 \frac{\sqrt{(q+1)(N-2q)(N-2q-1)}}{N(N-1)}.
\]
For $N\ge6$, choose $q=\lfloor N/6\rfloor$.  Then
\[
 q+1\ge\frac N6,
 \qquad
 N-2q\ge\frac{2N}{3},
 \qquad
 N-2q-1\ge\frac N2.
\]
Substitution into the preceding display gives the claimed
$\Omega(N^{-1/2})$ lower bound.
\end{proof}

\begin{theorem}[Hilbert--Schmidt lower bound for two-sided
Bose-symmetric extendible states]
\label{thm:HS-two-sided-bose-extension-lower}
As $n\to\infty$,
\[
 \sup_{d\ge2}
 \DHHS\Bigl(
  \BEXT_{n,n}(d,d),\,
  \SEP(d,d)
 \Bigr)
 =\Omega(n^{-1/2}).
\]
\end{theorem}

\begin{proof}
Use the witness $A_m$ from
\eqref{eq:HS-lower-witness}.  If
$x=\alpha e_0+v$ and $y=\beta e_0+w$, then
\[
 \abs{\bra{x\ot y}A_m\ket{x\ot y}}
 \le
 \frac{2}{\sqrt m}
 \abs{\alpha}\norm{v}\abs{\beta}\norm{w}
 \le\frac1{2\sqrt m}.
\]
Equality holds for $x=y=(e_0+e_1)/\sqrt2$, and hence
\begin{equation}
 h_{\SEP}(A_m)=\frac1{2\sqrt m}=h(A_m).
 \label{eq:HS-Am-separable-support}
\end{equation}
Set $q=\lfloor n/3\rfloor$ and take the bosonic $2n$-site state
$\rho_{q,m,2n}$ from the proof of Theorem~\ref{thm:HS-lower}.  Grouping
its first and last $n$ sites gives a two-sided Bose-symmetric extension
whose cross-block marginal is $\rho_{q,m,2n}^{(2)}$.  The
Hilbert--Schmidt-normalized witness $A_m/\sqrt2$ and
\eqref{eq:HS-Am-separable-support} give the
finite-$m$ estimate \eqref{eq:HS-finite-m-lower}, now with
$N=2n$ and the separable set in place of $\SEPH^{(2)}(m+1)$.
With $q=\lfloor n/3\rfloor$, the right-hand side tends, as
$m\to\infty$, to at least $1/(6\sqrt n)$ by the calculation in the
proof of Theorem~\ref{thm:HS-lower}.  For each $n\ge3$, choose a
sufficiently large finite $m$ and set $d_n=m+1$.  The resulting
distance is at least $1/(12\sqrt n)$, and taking the supremum over
finite dimensions proves the claim.
\end{proof}

\end{document}